\def\mnotes{0}
\def\final{0}
\def\llncs{0}
\documentclass[11pt]{article}

\usepackage{amsfonts,amsmath,epsfig,fullpage, graphicx, multirow}
\usepackage{color}

\def\colorson{0}

\usepackage{times}
\usepackage{color}
\usepackage{amsfonts,amsmath}
\usepackage{url}
 \ifnum\llncs=0

\usepackage{amsthm}
\usepackage{latexsym,amssymb}
\usepackage{graphicx}
\usepackage{wrapfig}

\usepackage[colorlinks=true,breaklinks=true,linkcolor=black,
citecolor=black,filecolor=black,menucolor=black,
pagecolor=black,urlcolor=black]{hyperref}

\advance \topmargin by -\headheight

\advance \topmargin by -\headsep

\fi %matches ifnum\llncs=0

\ifnum\mnotes=0
\newcommand{\mnote}[1]{}
\else
\newcounter{mynotes}
\newcommand{\mnote}[1]{\addtocounter{mynotes}{1}{{\bf !}}%
\marginpar{\scriptsize  {\arabic{mynotes}.\ {\sf \textcolor{red}{#1}}}}}
\fi %matches \ifnum \final=1

\newcommand{\snote}[1]{\mnote{S: #1}}
\newcommand{\anote}[1]{\mnote{A: #1}}
\newcommand{\enote}[1]{\mnote{E: #1}}

\ifnum\colorson=1
\newenvironment{todo}{\noindent
\sf \footnotesize \textcolor{blue}{To go here}:
\begin{CompactItemize}\color{blue}}
{\color{black}\end{CompactItemize}\rm \normalsize}
\else

\fi
\ifnum\llncs=0

\newtheorem{theorem}{Theorem}[section]
\newtheorem{lemma}[theorem]{Lemma}

\newtheorem{proposition}[theorem]{Proposition}
\newtheorem{claim}[theorem]{Claim}

\newtheorem{fact}[theorem]{Fact}
\newtheorem{corollary}[theorem]{Corollary}

\newtheorem{definition}{Definition}[section]

\theoremstyle{definition}

\theoremstyle{remark}

\newtheorem{myremark}{Remark} [section]
\newenvironment{remark}{\begin{myremark}}{$\diamondsuit$\end{myremark}}

\newtheorem{myexample}{Example}

\else
\spnewtheorem{protocol}{Protocol}{\bfseries}{\rmfamily}
\spnewtheorem{algm}{Algorithm}{\bfseries}{\rmfamily}
\spnewtheorem{fact}{Fact}{\bfseries}{\rmfamily}
\spnewtheorem{myclaim}{Claim}{\bfseries}{\itshape}
\fi

\newcommand{\thmref}[1]{Theorem~\ref{thm:#1}}

\newcommand{\lemref}[1]{Lemma~\ref{lem:#1}}

\newcommand{\corref}[1]{Corollary~\ref{cor:#1}}

\newcommand{\defref}[1]{Definition~\ref{def:#1}}

\newcommand{\secref}[1]{Section~\ref{sec:#1}}

\newcommand{\comment}[1]{}
\newcommand{\ignore}[1]{}

\newenvironment{CompactItemize}{
  \vspace{-4pt}
 \begin{list}{$\bullet$}{%
      \setlength{\leftmargin}{12pt}%
      \setlength{\itemsep}{-2pt}
      }}
  {\end{list}}

\newenvironment{proofof}[1]{\begin{proof}[Proof of {#1}]%\small
}{\end{proof}}

\newcommand{\ith}{i^{th}}

\newcommand{\kth}{k^{th}}

\newcommand{\eps}{\epsilon}

\DeclareMathOperator{\polylog}{polylog}

\DeclareMathOperator{\poly}{poly}

\newcommand{\defeq}{\stackrel{{\mbox{\tiny def}}}{=}}

\newcommand{\beq}{\begin{equation}}
\newcommand{\eeq}{\end{equation}}
\newcommand{\bml}{{\begin{multline}}}
\newcommand{\eml}{{\end{multline}}}

\newcommand{\etal}{{\it et~al.\ }}

\newcommand{\splong}[1]{$#1$-transitive-closure-spanner}
\newcommand{\spshort}[1]{$#1$-TC-spanner}
\newcommand{\spshorts}[1]{$#1$-TC-spanners}
\newcommand{\spshorter}{TC-spanner}

\newcommand{\tcspanner}{{\sc $k$-TC-Spanner}}
\newcommand{\directedspanner}{{\sc Directed $k$-Spanner}}
\newcommand{\clientserver}{{\sc Client/Server Directed $k$-Spanner}}
\newcommand{\maxisol}{m}
\newcommand{\Maxisol}{M}

\begin{document}

\title{Transitive-Closure Spanners} %, with Applications to Access Control, Data Structures, and Property Testing}
\author{Arnab Bhattacharyya\thanks{Massachusetts Institute of Technology, USA. Email:{\tt\{abhatt,elena\_g,kmjung\}@mit.edu}.} \and Elena Grigorescu\protect
\footnotemark[1] \and Kyomin Jung\protect
\footnotemark[1] \and
Sofya Raskhodnikova\thanks{Pennsylvania State University, USA. Email: {\tt sofya@cse.psu.edu}. Supported by National Science Foundation (NSF/CCF grant number 0729171).} \and David P. Woodruff\thanks{IBM Almaden Research Center, USA. Email:{ \tt dpwoodru@us.ibm.com}.}
}

\date{}
\maketitle\begin{abstract}
We define %introduce
the notion of a transitive-closure spanner of a directed graph. Given a directed graph $G = (V,E)$ and an integer $k \geq 1$, a {\it $k$-transitive-closure-spanner ($k$-TC-spanner)} of $G$ is a directed graph $H = (V, E_H)$ that has (1) the same transitive-closure as $G$ and (2) diameter at most $k$. These spanners were studied implicitly in access control, property testing, and data structures, and properties of these spanners have been rediscovered over the span of 20 years. % with as few edges as possible.
We bring these areas under the unifying framework of TC-spanners. We abstract the common task implicitly tackled
%in order to obtain efficient solutions
in these diverse applications as the problem of constructing sparse TC-spanners.
%show that the problems addressed in these diverse applications can all be reduced to constructing sparse TC-spanners.

We study the approximability of the size of the sparsest $k$-TC-spanner for a given digraph. Our technical contributions fall into three categories: algorithms for general digraphs, inapproximability results, and structural bounds for a specific graph family which imply an efficient algorithm with a good approximation ratio for that family.

\vspace*{-3mm}
\paragraph{Algorithms.} We present two efficient deterministic algorithms that find $k$-TC-spanners of size approximating the optimum.
%For $k>2$, we present an efficient algorithm that finds a $k$-TC-spanner with size within $\tilde{O}(n^{1-1/k})$ of the optimum.
 The first algorithm gives an $\tilde{O}(n^{1-1/k})$-approximation for $k>2$. Our method, based on a combination of convex programming and sampling,
yields the first sublinear approximation ratios for (1) {\sc Directed
  $k$-Spanner}, a well-studied generalization of {\sc $k$-TC-Spanner},
and (2) its variants \clientserver, and the {\sc
  $k$-Diameter Spanning Subgraph}. This resolves the main open
question of Elkin and Peleg (IPCO, 2001). %whereas previously nothing better than $O(n)$ was known for $k > 3$ for any of these problems. For $k = 3$ the best approximation ratio was $\tilde{O}(n^{2/3})$, due to Elkin and Peleg (IPCO, 2001).
The second algorithm, specific to the $k$-TC-spanner problem, gives an $\tilde{O}(n/k^2)$-approximation. It shows that for $k =\Omega(\sqrt{n})$, our problem has a provably better approximation ratio than {\sc Directed
  $k$-Spanner} and its variants. This algorithm also resolves an open question of Hesse (SODA, 2003).
%that performs better than our first algorithm for $k =\Omega (\log n / \log\log n)$.

\vspace*{-3mm}
\paragraph{Inapproximability.} Our main technical contribution is a pair of strong inapproximability results. We resolve the approximability of
%the important case of
$2$-TC-spanners, showing that it is $\Theta(\log n)$ unless $P = NP$.
For constant $k\geq 3$, we prove that the size of the sparsest $k$-TC-spanner is hard to approximate within $2^{\log^{1-\eps} n}$, for any $\epsilon > 0$, unless NP $\subseteq$ DTIME$(n^{\polylog n})$.
% and show $\Omega \left ((\log n)/k \right )$-inapproximability.
 Our hardness result helps explain the difficulty in designing general efficient solutions for the applications above, and it cannot be improved without resolving a long-standing open question in complexity theory. It uses an involved application of generalized butterfly and broom graphs, as well as noise-resilient transformations of hard problems, which may be of independent interest.

\vspace*{-3mm}
\paragraph{Structural bounds.} Finally, we study the size of the sparsest TC-spanner for $H$-minor-free digraphs, which include planar, bounded genus, and bounded tree-width graphs, explicitly investigated in applications above. We show that every $H$-minor-free digraph has an efficiently constructable $k$-TC-spanner of size $\tilde{O}(n)$, which implies an $\tilde{O}(1)$-approximation algorithm for this family.
Furthermore, using our insight that 2-TC-spanners yield property testers, we obtain a monotonicity tester with $O(\log^2 n /\epsilon)$ queries for any poset whose transitive reduction is an $H$-minor free digraph. This improves and generalizes the previous $\Theta(\sqrt{n} \log n/\epsilon)$-query tester of Fischer {\it et al} (STOC, 2002).

 \ignore{
 Along the way we resolve central open questions of earlier work:
\begin{enumerate}
\item We obtain an $O(n^{1-1/k} \log^{1-1/k}n)$-approximation ratio for the sparsest directed spanner and $k$-diameter spanning subgraph problems for any $k$. This extends the $\tilde{O}(n^{2/3})$ approximation ratio for $k = 3$ of Elkin and Peleg (IPCO, 2001). Previously, no sublinear ratio for $k > 3$ was known.
\item We give a monotonicity tester with $O(\log^2 n /\epsilon)$ queries for any poset whose transitive reduction is an $H$-minor free digraph. This class includes the interesting and previously studied cases of planar, bounded genus, and bounded tree-width graphs. Our result greatly improves the testers of Fischer {\it et al} (STOC, 2002), who require $\Omega(\sqrt{n} \log n/\epsilon)$ queries for such graphs.
\end{enumerate}
}
\end{abstract}

\thispagestyle{empty}
\setcounter{page}{0}
\newpage

%%%%%%%%%%%%%%%%%%%%%%%%%%%%%%%%%%%%%%
\ifnum\final=1
\tableofcontents
\fi
%%%%%%%%%%%%%%%%%%%%%%%%%%%%%%%%%%%%%%

\newpage
\section{Introduction}\label{sec:intro}
\vspace{-1mm}
A {\em spanner} can be thought of as a sparse backbone of a graph that approximately preserves distances
between every pair of vertices. More precisely, a subgraph $H = (V,E_H)$ is a {\em $k$-spanner} of
$G=(V,E)$ if for every pair of vertices $u,v \in V$, the shortest path distance $d_H(u,v)$ from $u$ to $v$ in
$H$ is at most $k \cdot d_G(u,v)$.  \ifnum\final=0 Since
they were introduced by Peleg and Sch\"affer \cite{PelegSchaffer} in the context of distributed
computing, spanners for undirected graphs have been extensively studied.  The tradeoff between the
parameter $k$, called the {\em stretch}, \else Spanners for undirected graphs, introduced by Peleg
and Sch\"affer \cite{PelegSchaffer}, have been extensively studied. The tradeoff between the
parameter $k$ \fi and the number of edges in a spanner is relatively well understood: for every $k
\geq 1$, any undirected graph on $n$ vertices has a $(2k-1)$-spanner with $O(n^{1+1/k})$ edges
\cite{ADDJS93,Pel00, approx-distance-oracles}.  This is known to be tight for $k=1,2,3,5$ and is
conjectured to be tight for all $k$ (see, for example a survey by Zwick \cite{zwick-survey01}).
Undirected spanners have numerous applications, such as efficient routing \cite{ Cowen01, CowenWagner, PelUp89, roundtrip-spanners, TZ01b}, simulating synchronized protocols in unsynchronized
networks \cite{PelegU89}, parallel and distributed algorithms for approximating shortest paths
\cite{coh98, coh00, e01}, and algorithms for distance oracles \cite{BasSen, approx-distance-oracles}.

\ifnum\final=1
Given the importance of spanners for undirected graphs, it is natural to try extending
the notion to the directed setting.
The above definition of spanners generalizes readily to arbitrary directed graphs, but becomes somewhat uninteresting: whenever an edge $(u,v)$ forms a unique path from $u$ to $v$, it has to be present in the $k$-spanner for every $k$. Consider, for example, a directed line
$L_n$ with vertices $\{1,2,\ldots, n\}$ and edges $\{(i, {i+1}) : 1 \leq i < n\}$. No strict
subgraph of $L_n$ is a $k$-spanner for it, for any $k$.
%
%% S: This is not clearly explained -- the issue is not multiple paths
%%(e.g., hypercube has multiple paths between most pairs of nodes
%unless there are multiple directed
%paths between pairs of vertices in the graph.  (Otherwise, if for a directed graph $G$,
%there is at most one path between every pair of vertices,
 %then any $k$-spanner of $G$ must equal $G$.)
Roditty \etal
\cite{roundtrip-spanners}, building on the ideas of Cowen and Wagner \cite{CowenWagner},
 circumvent this problem by considering the {\em roundtrip distance} between the nodes instead of the conventional distance.
 The roundtrip distance from $u$ to $v$ is defined as the sum of the distance from $u$ to $v$ and from $v$ to $u$. They propose a notion of a {\em $k$-roundtrip-spanner}
 which expands {\em roundtrip} distances between all pairs of vertices by at most a factor of~$k$.
 This notion is well-suited for (weighted) graphs with a small number of strongly connected components. However,
 a roundtrip spanner ensures that the distance between $u$ and $v$
is small only when both $u$ is reachable from $v$ {\em and} $v$ is reachable from $u$.  If there is no path from $u$ to
$v$, the distance from $u$ to $v$ is taken to be $\infty$. So the notion of a roundtrip spanner is not interesting, for example, for directed
acyclic graphs (DAGs).
\else
In the directed setting, two notions of spanners have been considered in the literature: the direct generalization of the above definition \cite{PelegSchaffer} and {\em roundtrip spanners} \cite{CowenWagner, roundtrip-spanners}. %\snote{have to keep references to previous definitions}
\fi
\ifnum\final=1
In this paper, we consider a different %introduce a new
definition of spanners which is appropriate for all directed graphs.
Our definition captures the notion that a spanner should have a small diameter but preserve the connectivity of
the original graph.
That is, $u$ is reachable from $v$ in the spanner
if and only if it is reachable in the original graph.
\else
In this paper, we introduce a new definition of directed spanners that captures the notion that a spanner should have a small diameter but preserve the connectivity of
the original graph.
\fi
\ifnum\final=1
Recall that the {\em transitive closure} of a
graph $G=(V,E)$ is a graph $(V,E_{TC})$ where $(u,v)\in E_{TC}$ if and only if there is a directed
path from $u$ to $v$ in $G$.
\fi
%Formally, we define a {\em transitive-closure spanner} for $G$ as follows.%to be a conventional spanner for the transitive closure of $G$:
% Our formulation uses one-way distance instead of roundtrip
% distance, but the resulting spanner is a subgraph of the transitive closure of
% the given graph, not necessarily the given graph itself.  The
% transitive closure of a directed graph usually has several paths
% between a pair of vertices, making our definition interesting.
\begin{definition}[\spshorter]\label{def:main}
Given a directed graph $G=(V,E)$ and an integer $k \geq 1$, a
\textbf{\splong{k}} (\textbf{\spshort{k}}) is a directed graph $H=(V,E_H)$
with the following properties:
\begin{enumerate}
\item $E_H$ is a subset of the edges in the transitive closure of $G$.
\item For all vertices $u,v\in V$, if $d_G(u,v) < \infty$, then $d_H(u,v) \leq k$.
\end{enumerate}
\end{definition}
Notice that a \spshort{k} of $G$ is just a directed spanner of the transitive-closure of $G$ with stretch~$k$. Nevertheless, a \spshort{k} is interesting in its own right due to the numerous TC-spanner-specific applications we present in Section~\ref{sec:applications}.
\ifnum\final=1
\noindent
The first property ensures that only the vertices that are connected in $G$ will remain connected
in the spanner. The second property guarantees that distances between connected pairs of vertices are
small.
\fi
%The contributions of this paper are three-fold: (1) We bring several diverse applications, including property testing, access control and data structures, under the unifying framework of TC-spanners. (2) We initiate the study of the computational problem of finding the size of the sparsest $k$-TC-spanner for a given digraph, which we refer to as \tcspanner. We present several results on \tcspanner\ and the related well-studied problem, called \directedspanner, of finding the size of the sparsest $k$-spanner for a given (not necessarily transitively closed) digraph.
%(3) Finally, we give structural bounds on the size of the sparsest $k$-TC-spanners for $H$-minor free graphs. Our constructions yield new algorithms and data structures for the application areas we discuss in Section~\ref{sec:applications}.
%

One of the focuses of this paper is the study of the computational problem of finding the size of the sparsest $k$-TC-spanner for a given digraph, referred to as \tcspanner. It is a special case of the problem of finding the size of the sparsest directed spanner, called \directedspanner, that has been previously studied. %It is easy to show that these
Both problems are NP-hard (see Appendix~\ref{sec:nphard}).
\vspace*{-2mm}
\subsection{Related Work}\label{sec:previous-work}
\vspace*{-1mm}
Thorup \cite{tho92} considered a special case of TC-spanners of graphs $G$ %the tense has to match the rest of the sentence
%%"matched shortcuttings" is his terminology
that have at most twice as many edges as $G$, and conjectured that for all directed graphs $G$ with $n$ vertices there are such TC-spanners with stretch polylogarithmic in $n$. He proved his conjecture for planar graphs \cite{thorup95}, but later Hesse \cite{h03} gave a counterexample to Thorup's conjecture for general graphs. TC-spanners were also studied for directed trees: implicitly in \cite{AlonSchieber87,AFB05, cha87, DGLRRS99,Yao82} and explicitly in \cite{thorup97}.
For the directed line, \cite{AlonSchieber87} (and later, \cite{AFB05}) showed that the size of the sparsest $k$-TC-spanner is
$\Theta(n\cdot
\lambda_k(n)),$ where $\lambda_k(n)$ is the $\kth$-row inverse Ackermann function. \cite{AlonSchieber87,cha87,thorup97} gave the same bounds for rooted directed trees.

\vspace*{-3mm}
\paragraph{Approximability of directed spanner problems.}
%Finding the sparsest \spshorter s is a non-trivial task for many graphs, including a line $L_n$, defined above.
%\noindent Our main conceptual contribution is the initiation of a
%systematic study of sparse \spshorter s.
%This work initiates a systematic study of the sparsest \spshorter s for directed graphs.
\iffalse
\ignore{
Many variants of this problem has been considered in literature. We will discuss {\sc Client/Server Directed $k$-Spanner} \cite{ElkinPeleg01}, where the input specifies a set of {\em client} edges that have to be spanned and the set of {\em server} edges that may be used for spanning.
and {\sc $k$-Diameter Spanning Subgraph}.
}
\fi
%Observe that {\sc $k$-TC-Spanner} is a special case of {\sc Directed $k$-Spanner}.
%\vspace*{-3mm}
%\paragraph{Approximability of directed spanner problems.}

All algorithms for {\sc Directed $k$-Spanner} immediately yield algorithms for {\sc $k$-TC-Spanner} with the same approximation ratio.
%, but, as we argue in Section~\ref{sec:techniques}, existing hardness results for {\sc Directed $k$-Spanner} do not apply, and existing techniques are not sufficient to prove hardness of {\sc $k$-TC-Spanner}.
%In particular, we give new algorithms for {\sc Directed $k$-Spanner}, and hence also for {\sc $k$-TC-Spanner}.
%Note that for a digraph $G$ on $n$ vertices, the size of the sparsest \spshort{k} %, denoted by $S_k(G)$,
%is $O(n^2)$, at most the size of the transitive closure of $G$. %This upper bound is tight, in general, as witnessed by a complete bipartite graph with $n/2$ vertices in each part, and all edges directed from one part to the other.
%Thus, $O(n)$-approximation for {\sc $k$-TC-Spanner} is trivial.
Kortsarz and Peleg \cite{KortPel} give an $O(\log n)$-approximation algorithm for {\sc Directed-$2$-Spanner}, and Kortsarz \cite{Kort} shows that this approximation ratio cannot be improved unless P=NP. For $k=3$, Elkin and Peleg \cite{ElkPel} present an $\tilde{O}(n^{2/3})$-approximation algorithm. Their algorithm is complicated, and the $\polylog$ factor hidden in the $\tilde{O}$ notation is not analyzed. %for digraphs.
%Applying these algorithms to the transitive-closure of a digraph $G$ yields an $O(\log n)$-approximation for the sparsest \spshort{2} problem and an $\tilde{O}(n^{2/3})$-approximation for the sparsest \spshort{3} problem.
For $k\geq 4$, sublinear factor approximation algorithms are known only in the undirected setting \cite{PelegSchaffer}. We note that Dodis and Khanna \cite{dk99} and Chekuri \etal \cite{cegs08} study algorithms that might seem relevant to \tcspanner.
%Dodis and Khanna \cite{dk99} study the problem of finding the minimum-cost subset of missing edges that can be added %to a (directed) graph $G$, with costs and
%lengths associated to the missing edges, so as to ensure that that there is a path of length at most $k$ between every pairs of nodes (not only those connected in $G$).
%Chekuri \etal \cite{cegs08} give an $O(p^{1/2 + \epsilon})$-approximation algorithm for the directed
%Steiner network.
In Appendix~\ref{app:previous-work} we explain why these algorithms do not work for \tcspanner.

For any constant $k> 2$ and $0<\eps<1$, it is hard to approximate {\sc Directed $k$-Spanner}
within a factor of $2^{\log^{1-\eps} n},$ assuming  NP$\not \subseteq $DTIME($n^{poly \log n}$) \cite{ElkPel}.
Moreover, \cite{ElkinPeleg07} extend this result to $3\leq k=O(n^{1-\delta})$ for any $0<\delta < 1$. Thus, according to Arora and Lund's classification \cite{ApproxAlgs} of NP-hard problems, {\sc Directed $k$-Spanner} is in class III, for $3\leq k = O(n^{1-\delta})$. Moreover, \cite{ElkinPeleg07} show that proving that {\sc Directed $k$-Spanner} is in class IV, that is, inapproximable within $n^\delta$ for some $0<\delta<1$, would resolve a long standing open question in complexity theory, and cause classes III and IV to collapse into a single class.

\ifnum\final=1
For a directed graph $G$ on $n$ vertices, the size of the sparsest \spshort{k}, denoted by $S_k(G)$, is at most
the size of the transitive closure of $G$, which is $O(n^2)$. This upper bound is tight, in general, as witnessed
by a complete bipartite graph with $n/2$ vertices in each part, and all edges directed from one part to the other.

In addition to the work of Thorup \cite{tho92} and Hesse \cite{h03} described after Definition~\ref{def:main},
the articles mentioned in Section~\ref{sec:applications},
implicitly contain bounds on \spshorter s for simple families of graphs, such as a line $L_n$ and a tree.
Alon and Schieber \cite{AlonSchieber87} implicitly give tight bounds on $S_k(L_n)$. They show that
the size of the sparsest \spshort k for a line, $S_k(L_n)=\Theta(n\cdot \lambda_k(n)),$ where $\lambda_k(n)$ is the $\kth$-row
inverse Ackermann function.
In particular, the sizes of the sparsest \spshort{k}s for a line for small $k$ are: $S_2(L_n)=
\Theta(n\log n)$, $S_3(L_n)=\Theta(n \log \log n)$ and $S_4(L_n)=\Theta(n \log^* n)$.
\cite{AlonSchieber87}
also show that the smallest $k$ for which $S_k(L_n)=O(n)$ is $O(\alpha(n))$. Note that $S_k(L_n)\geq n-1$ for all $k$,
since all edges of the form $(i,i+1)$ have to be present in a \spshorter\ to ensure the same connectivity as in $L_n$.
\ignore{
In our terminology, they proved the following.
\begin{theorem}[\cite{AlonSchieber87}]\label{thm:as87}
For any $k \geq 2$, the size of the smallest \spshort{k} for $P_n$ is $\Theta(n\cdot \lambda(k,n))$,
where $\lambda(k,n)$ is the inverse of a function at the $\lfloor k/2 \rfloor$-th level of the
primitive recursive hierarchy.  Precisely, $\lambda(0,n) = \lceil n/2\rceil$, $\lambda(1,n) = \lceil
\sqrt{n} \rceil$, and for $k \geq 2$, $\lambda(k,n) = \min \{i : \lambda^{(i)}(k-2,n)\leq 1\}$ where
$\lambda^{(1)}(k,n) = \lambda(k,n)$ and $\lambda^{(i)}(k,n) = \lambda(k,\lambda^{(i-1)}(k,n))$.
\end{theorem}
}%ignore
%At first glance, it is really quite surprising that such small spanners exist.
The techniques in \cite{AlonSchieber87} can be used to give the same upper bounds on $S_k$ for a more general class of graphs:
rooted trees on $n$ nodes where all edges are directed away from the root. As mentioned in Section~\ref{sec:applications},
\cite{AFB05} also implicitly develop  \spshort{k}s for $k=2, 3$ and $\log\log n$ for this class. In the light of \cite{AlonSchieber87}'s
lower bound for the line, their results are optimal. \snote{However, the $\alpha(n)$-spanner for a tree also has $O(n)$ edges.}
\fi

\ifnum\final=0
%Implicit in \cite{AFB05}
%are \spshorter s for rooted trees on $n$ nodes with edges directed away from the root: a \spshort{2} of size $O(n^{3/2})$, a \spshort{3} of size $O(n \log \log n)$
%and an \spshort{O(\log \log n)} of size $O(n)$.
%Given the abundance of applications of sparse spanners, there has been a  vast literature studying the complexity of  building spanners, as well as the algorithmic implications.
%

\vspace*{-2mm}
%\paragraph{Approximability of related problems.}
%Dodis and Khanna \cite{dk99} and Chekuri \etal \cite{cegs08} study approximation algorithms for more general problems.
%Dodis and Khanna \cite{dk99} study the problem of finding the minimum-cost subset of missing edges that can be added %to a (directed) graph $G$, with costs and
%lengths associated to the missing edges, so as to ensure that that there is a path of length at most $k$ between every pairs of nodes (not only those connected in $G$).
%Chekuri \etal \cite{cegs08} give an $O(p^{1/2 + \epsilon})$-approximation algorithm for the directed
%Steiner network.
%In Appendix~\ref{app:previous-work} we explain why the application of these algorithms to {\sc $k$-TC-Spanner} gives approximation ratio no better than $O(n)$.
%
%\vspace*{-1mm}
\subsection{Our Contributions}\label{sec:results}
\vspace*{-2mm}
The contributions of this paper are the following: (1) we bring several diverse applications, including property testing, access control and data structures, under the unifying framework of TC-spanners, (2) we obtain strong bounds on the approximability of \tcspanner \ and \directedspanner \ as well as some well-studied variants of these problems, and (3) we characterize the exact size of TC-spanners and obtain better bounds for the family of $H$-minor free graphs, which include planar, bounded-treewidth, and bounded genus graphs. Our results on the approximability of \tcspanner \ are summarized in Table~\ref{table:results}.
\begin{table}
\begin{tabular}{|l||l|l|l|}
  \hline
  % after \\: \hline or \cline{col1-col2} \cline{col3-col4} ...
  Setting of $k$ & Implied by previous work & This paper & Notes \\
  \hline
  \hline
 $k=2$  & $O(\log n)$ \cite{KortPel}  & $\Omega(\log n)$ % unless P=NP
 &\\
  \hline
 constant $k>2$    & &   $\Omega(2^{\log^{1-\eps} n})$ %unless  NP$\subseteq $DTIME($n^{poly \log n}$)
 &\\
 \hline
 \hline
 $k=3$  & $O(n^{2/3} \polylog n)$ \cite{ElkPel}  & $O((n\log n)^{2/3})$
 %,   $\Omega(2^{\log^{1-\eps} n})$ unless  NP$\subseteq $DTIME($n^{poly \log n}$)
 & \multirow{2}{*}{applies to {\sc Directed $k$-Spanner}}\\
 \cline{1-3}
 $k>3$  & $O(n)$ [trivial]  & $O((n\log n)^{1-1/k})$
 %,   $\Omega(2^{\log^{1-\eps} n})$
 %unless  NP$\subseteq $DTIME($n^{poly \log n}$)
 &\\
  \hline
 $k=\Omega\left(\frac{\log n}{\log\log n}\right)$  & $O(n)$ [trivial]  & $O \left (\frac{n \log n}{k^2 + k \log n}\right )$ & separation from {\sc Directed $k$-Spanner}\\
  \hline
\end{tabular}
\caption{Summary of Results on Approximability of {\sc $k$-TC-Spanner}}\label{table:results}\end{table}

\vspace*{-3mm}
\paragraph{Algorithms  for {\bf {\sc $k$-TC-Spanner}} and Related Problems.} %In section \ref{sec:general},
We present two deterministic polynomial time approximation algorithms for \tcspanner. Our first algorithm uses a new combination of convex programming and sampling, and gives an $O((n\log n)^{1-1/k})$-ratio for \tcspanner. Moreover, our method yields the same approximation ratio for \directedspanner \ and its well-studied variants: {\sc Client/Server Directed $k$-Spanner}, and {\sc $k$-Diameter Spanning Subgraph} (see \cite{ElkinPeleg01} for definitions). %This resolves the ``challenging direction'' of Elkin and Peleg \cite{ElkinPeleg05} % I did not know that direction of David Peleg
This resolves the open question
of finding a sublinear approximation ratio for these problems for $k > 3$, described as a "challenging direction" for research on directed spanners by Elkin and Peleg \cite{ElkinPeleg05}. Our algorithm for $k=3$ is arguably simpler than the $O(n^{2/3} \polylog n)$-approximation algorithm of \cite{ElkinPeleg05}.
 %We use a new linear programming formulation. In addition to the \spshort{k} problem, the algorithm extends to other %spanner problems, such as client/server variants and the $k$-diameter spanning subgraph problem
%(\cite{ElkinPeleg05}), none of which were known to have $o(n)$-approximability for $k > 3$.

Our second algorithm has an $\tilde{O}(n/k^2)$ ratio for {\sc $k$-TC-Spanner}.  This demonstrates a separation
between {\sc $k$-TC-Spanner} and {\sc Directed $k$-Spanner}:
for $k = \sqrt{n}$, it gives $O(\log n)$-approximation for {\sc $k$-TC-Spanner} while
\cite[Theorem 6.6]{ElkinPeleg07}  showed that {\sc Directed $\sqrt{n}$-Spanner} is $2^{\log^{1-\eps}n}$-inapproximable. Moreover, Hesse \cite{h03} asks for an algorithm to add $O(|G|)$ "shortcuts" to a digraph and reduce its diameter to $\sqrt{n}$. Our second algorithm returns a \spshort{\sqrt{n}} of size $O(|G|+\log n)$, %=O(|G|)$,
answering his question.

% for any $\eps \in (0,1)$.
%Our second algorithm has a better ratio than our first for $k = \Omega(\frac{\log n}{\log \log n})$.
%
%We also need to carefully interleave several operations on a base {\sc MIN-REP} instance to coordinate the instance with our butterfly and broom graphs.

\vspace*{-3mm}
\paragraph{Inapproximability of {\bf {\sc $k$-TC-Spanner}}.} %In section \ref{sec:khardness},
We present two results on the hardness of {\sc $k$-TC-Spanner}. First, we prove for $k=2$ that the $O(\log n)$ ratio of \cite{KortPel} is optimal unless P=NP. Next, we show that for constant $k>2$, {\sc $k$-TC-Spanner} is inapproximable within a factor of $2^{\log^{1-\eps} n}$, for any $\epsilon > 0$, unless  NP$\subseteq $DTIME($n^{\polylog n}$). This %hardness
result is our main technical contribution. Observe that a stronger inapproximability result for $k>2$ would imply the same inaproximability for {\sc Directed-$k$-Spanner}, and as shown in \cite{ElkinPeleg07}, collapse classes III and IV in Arora and Lund's classification.

Our $2^{\log^{1-\eps}n}$-hardness matches the known hardness for {\sc Directed $k$-Spanner}. As is
the case for {\sc Directed $k$-Spanner}, we start by building a directed graph from a well-known
hard problem called {\sc MIN-REP}, which has the same inapproximability as {\sc Symmetric Label
  Cover}. However, %the similarity between the proofs stops there. Indeed,
  as illustrated in Section
\ref{sec:lb}, all known hard instances for {\sc Directed $k$-Spanner} cannot imply anything better
than $\Omega(1)$-hardness for \tcspanner. Intuitively, our lower bound is much harder to prove than the one for
{\sc Directed $k$-Spanner} since our instance must be transitively-closed, and thus, many more
``shortcut'' routes between pairs of vertices exist. Our construction uses a novel application of
the generalized butterfly and broom graphs, together with several transformations of
the {\sc MIN-REP} problem, which make it {\em noise-resilient}. We call a {\sc MIN-REP} instance noise-resilient to indicate that its structure is preserved under small perturbations. The paths in the generalized butterfly are well-structured, which
 allows us to analyze the many different routes possible in the transitive closure.  To realize these ideas, we perform various transformations on a {\sc MIN-REP} instance to coordinate it with multiple copies of butterflies and brooms.

\ignore{
In section \ref{sec:khardness}, we study the complexity of approximating the minimal \spshort{k} size.  We show unless $P = NP$, it is infeasible to approximate the minimal \spshort{k} size to
within a ratio of $O(\frac{1}{k} \log n)$.  Thus, for the important case of $k=2$ (yielding efficient monotonicity testers) the $O(\log n)$ ratio of \cite{KortPel} is optimal. \anote{Hopefully, we will get better results here soon.}
Our hardness is based on an embedding of specialized \textsf{Set Cover} instances
into butterfly graphs. Butterfly graphs were recently used in \cite{woo06} in the context of
spanners, and here again, they play a useful role. We also show NP-hardness of the decision
problem for any $k < n^{1-\epsilon}$ for any $\epsilon > 0$.
}
%ignore
% Namely, we first show that if \textsc{NP}$\not =$\textsc{P} then it is \textsf{NP}-hard to
% approximate $S_k(G)$ within a ratio of $\Omega(\frac{1}{k} \log n)$. For $k=2$ this result
% establishes that the $O(\log n)$ approximation algorithm of \cite{KortPel} for $2$-spanners on
% directed graphs is essentially optimal.  In terms of general upper bounds, we exhibit two algorithms
% with various approximation ratios depending on $k$. The first one is a deterministic $\tilde{O}
% (n/k^2)$ -ratio approximation algorithm which performs better than the other when $k\geq \log
% n$. The second one achieves an approximation ratio of $\tilde{O}(n^{1-1/k})$ and performs better
% when $k\leq \log n / \log \log n$.  The latter approximation ratio applies to the more general
% directed spanner problem, for every $k$. This extends the $\tilde{O}(n^{2/3})$ approximation ratio
% for $k = 3$ of Elkin and Peleg \cite{ElkPel}. Previously, no sublinear ratio for $k > 3$ was known.

\vspace*{-3mm}
\paragraph{Structural Results.} Finally, we study the minimum \spshort{k} size for
a specific graph family with sparse \spshorts{k}: $H$-minor-free graphs. % and hypercubes.
A graph $H$ is a {\em minor}  of $G$ if $H$ is a subgraph of a graph obtained from $G$ by a sequence of edge contractions and deletions.
For a fixed graph $H$ (e.g., $K_5$), the family of {\em $H$-minor-free graphs} is a minor-closed family that excludes $H$.  % sofya: almost verbaitum from http://www-math.mit.edu/~hajiagha/graphminoralgorithm.pdf
\ignore{Given an edge
$e = \{v,w\}$ in a graph $G$, the contraction of $e$ in $G$ is the
result of identifying vertices $v$ and $w$ in $G$ and removing
all loops and duplicate edges. A graph $H$ obtained by a
sequence of such edge contractions starting from $G$ is said
to be a contraction of $G$. A graph $H$ is a minor of $G$ if
$H$ is a subgraph of some contraction of $G$. A graph family
$F$ is minor-closed if any minor of any graph in $C$ is also a
member of $C$. A minor-closed graph family $F$ is $H$-minor-free
if $H\notin F$.}%ignore
Examples of such families include planar graphs, bounded treewidth graphs, and bounded genus graphs, explicitly studied in applications in Section~\ref{sec:applications}.
%In section \ref{sec:separable}
 For $H$-minor-free graphs, we efficiently construct \spshort{2}s of size $O(n\log^2 n)$, and  \spshort{k}s of size $O(n\cdot \log n \cdot \lambda_k(n))$, where $\lambda_k(\cdot)$ is the $k$-row inverse Ackermann function. The main idea is to use the path separators for undirected $H$-minor free graphs due to Abraham and Gavoille \cite{AbrahamGavoille}. However, although the separators are paths, in our digraph they may be the union of many dipaths, and so we cannot efficiently recurse using the sparse \spshorts{k} for the directed line of Alon and Schieber \cite{AlonSchieber87}. We observe that these separators satisfy a stronger property than claimed in \cite{AbrahamGavoille}, effectively allowing us to encode the direction of edges in a cost function associated with the separators.

\vspace*{-2mm}
\subsection{Applications of \spshorter s}\label{sec:applications}
\vspace*{-1mm}
%\spshorter s naturally arise in settings where one needs to drastically decrease the diameter of a directed
%graph, while ensuring that the resulting graph is sparse and has the same connectivity as the original graph.
\ifnum\final=1
For instance, in the {\em bosses and subordinates}\snote{Keep?} problem, the input is a graph of
a large organization, where each node represents an employee,
and an edge $(u,v)$ means that $u$ is $v$'s manager on some project.
%Now the organization decided that it is inefficient to pass orders along a long chain from top level
%bosses to bottom level subordinates.
Instead of playing broken telephone between top level bosses and low level subordinates, the CEO
would like to ensure that each order makes at most $k$ hops, namely, passes through at most $k-1$ intermediate employees.
However, the intermediate employees have to be on the chain of command: that is, an order from $u$ to $v$ has to be passed via people on
a path from $u$ to $v$. The goal is to minimize the number of boss-subordinate relationships in the organization. The bosses and
subordinates problem directly corresponds to finding a sparsest \spshort{k} for the organization graph.

%Below we observe that in some well-studied problems, \spshorter s are applicable.
In some of the applications, it is obvious that \spshorter s are
the right abstraction for the problem; in others, the connection to \spshorter s is more subtle and requires some work to prove. One of the
contributions of this paper is bringing all these diverse applications under the unifying framework of \spshorter s.
\fi
\paragraph{Monotonicity Testing.}
Monotonicity of functions \cite{AilonChazelle, BRW, DGLRRS99, EKKRV, Fischer04, FLNRRS02, GoldreichGLRS00, HalevyK04}
is one of the most studied properties in property testing \cite{GGR, RS}. Fischer \etal \cite{FLNRRS02} prove that testing monotonicity is equivalent to several other testing problems.
Let $V_n$ be a poset of $n$ elements and $G_n=(V_n,E)$ be
the relation graph, i.e., the Hasse diagram, for $V_n$. A function $f: V_n \rightarrow \mathbb{R}$
is called {\em monotone} if $f(x) \leq f(y)$ for all $(x,y)\in E$.
%% S: \preceq requires a definition
%\Longleftrightarrow x \preceq y$.
We say $f$ is $\eps$-far from monotone if $f$ has to be changed
on $\geq\eps$ fraction of the domain to become monotone, that is, $\min_{\text{monotone } g}|\{x : f(x) \neq
g(x)\}| \geq \eps n$.
A monotonicity tester on $G_n$ is an algorithm that, given an oracle for a
function $f: V_n \rightarrow \mathbb{R}$, passes if $f$
is monotone but fails with probability $\geq\frac{2}{3}$ if $f$ is {\em $\epsilon$-far} from monotone.
% For instance, if $G_n$ is a directed line, the tester needs to determine whether the input sequence,
% specified by $f$, is sorted or $\eps$-far from sorted. If $G_n$ is a 2-dimensional grid (with vertex set $\{1,...,m\}\times
% \{1,...,m\}$ and edge set $\{(x,y) \ | \ x_1 = y_1 \text{ and } x_2+1 =  y_2 \}\cup \{(x,y) \ | \ x_1 +1 = y_1 \text{ and } x_2=y_2 \}$),
% the goal is to determine whether the input matrix has non-decreasing rows and columns.
The optimal monotonicity tester for the directed line $L_n$, consisting of vertices $\{1,2,\ldots, n\}$ and edges $\{(i, {i+1}) : 1 \leq i \leq n-1\}$, proposed by Dodis \etal \cite{DGLRRS99},
is based on the sparsest \spshort{2} for that graph. Implicit in the proof of Proposition~9 in \cite{DGLRRS99} is a lemma relating the complexity of a monotonicity tester for $L_n$ to the size of a \spshort{2} for $L_n$. We generalize this by observing that a sparse \spshort{2} for any partial order graph $G_n$ implies an efficient monotonicity tester on $G_n$. In Appendix~\ref{app:monotonicity}, we prove the following lemma.
\begin{lemma}\label{lem:spanners-to-monotonicity-tests}
If a directed acyclic graph $G_n$ has a \spshort{2} with $s(n)$ edges, then
there exists a monotonicity tester on $G_n$ that runs in time $O\left(\frac{s(n)}{\epsilon n}\right)$.
\end{lemma}
\noindent
Therefore, all the \spshort{2} constructions described in this paper yield monotonicity
testers for functions defined on the corresponding posets. Moreover, for $H$-minor free graphs, the resulting tester has much better query complexity than the previously known, due to Fischer \etal \cite{FLNRRS02}. Indeed, we achieve testers with $O(\log^2 n /\epsilon)$ queries, whereas previous testers required $\Theta(\sqrt{n}/\epsilon)$ queries.

\vspace*{-3mm}
\paragraph{Key Management in an Access Hierarchy.} In the problem of key management in an access hierarchy, i.e., access control, there is a
partially ordered set (poset) of access classes and a key associated with each class.
This is modeled by a directed graph $G$ whose nodes are classes and whose edges indicate
an ordering. A user is entitled to access a certain class and all classes reachable from
it. This problem arises in content
distribution, operating systems, and project development (see,
e.g., the references in \cite{AFB05}). One approach to the access control problem \cite{ABF06, AFB05, SFM07} is to associate
public information $P(i,j)$ with each edge $(i,j) \in G$ and a secret key $k_i$ with each node $i$.
There is an efficient algorithm $A$ which takes $k_i$ and $P(i,j)$ and generates
$k_j$. However, for each $(i,j)$ in $G$, it is computationally hard to generate $k_j$ without knowledge of $k_i$. To obtain a key $k_v$ from a key $k_u$, algorithm~$A$ is run $d_G(u,v)$ times.
To speed this up, \cite{AFB05} suggest adding edges to $G$ to increase
connectivity. To preserve the access hierarchy of $G$, new edges must be from the transitive closure of $G$. The number of edges added corresponds to
the space complexity of the scheme, while the shortest-path distances correspond to the time complexity.
%
%Ideally, one would like to add a small number of edges,
%corresponding to the space complexity of the scheme, while making the shortest-path
%distances as small as possible.
% If one had a sparse \spshort{k} $H$ of $G$, one could add
% the edges of $H$ to $G$ and reduce all distances to at most $k$.
Implicit in \cite{AFB05} are TC-spanners
for directed trees with $k=3$ and size $O(n \log \log n)$ and also with $k=O(\log \log
n)$ and size $O(n)$. Our results for $H$-minor free graphs extend the known posets for which access control schemes have $O(n\polylog n)$ storage and $O(1)$ key derivation time. Our approximation algorithms yield sparser \spshorts{k} for general posets.
%In this paper we consider
%more general graphs (e.g., directed grids, the hypercube, and planar
%graphs), and our corresponding $k$-spanners can be used in the framework
%of \cite{AFB05} to give the best known time/space tradeoffs for these
%hierarchies.

%\paragraph{Example name.}\anote{This example will be changed.}\snote{We can change it to prereq example
%if it's not too bogus. Another bogus example we can add is bosses-subordinates}
%Suppose there is a collection of webpages where each webpage links to a subset of the
%others.  This collection defines a graph $\mathcal{L}$; the vertices of $\mathcal{L}$ are the pages and there is an
%edge in $\mathcal{L}$ from page $A$ to page $B$ if $A$ links to $B$.  Then if there exists a path in the
%graph from page $A$ to page $B$, $B$ is relevant to $A$.  Now, suppose that we want to add an additional set
%of links to the webpages so that it is possible to surf from any page to any other page relevant to it by
%clicking only $k$ times; additionally, we want that the intermediate pages be relevant to the starting page.
%Our goal is to minimize the number of links added.  Clearly, the solution to this problem corresponds to the
%sparsest $k$-spanner for $\mathcal{L}$.  In general, $k$-spanner constructions are useful in increasing
%the connectivity of networks where there is a natural notion of directedness.
%

\vspace*{-3mm}
\paragraph{Partial Products in a Semigroup.} Yao \cite{Yao82} and Alon and Schieber \cite{AlonSchieber87} study
space-efficient data structures for the following problem:
Preprocess elements $\{s_1,\ldots,s_n\}$ of a semigroup $(S,\circ)$, such as $(\mathbb{R},\min)$, to be able to compute
partial products $s_{i}\circ s_{i+1}\circ \cdots \circ s_j$ for all $1 \leq i <
j \leq n$ with at most $k$ queries to a small database of pre-computed partial products.
This problem reduces to finding a sparsest
\spshort{k} for a directed line
$L_{n+1}$. %, consisting of vertices $\{1,2,\ldots, n+1\}$ and edges $\{(i, {i+1}) : 1 \leq i \leq n\}$.
% If the database stores a product $s_{u}\circ \cdots \circ s_v$ for each \spshort{k} edge $(u,v+1)$,
% every product $s_{i}\circ \cdots \circ s_j$ can be computed by multiplying the products corresponding to the edges
% on a path of length at most $k$ from $i$ to $j+1$ in the \spshort{k} for $L_{n+1}$. Conversely, each
% preprocessing  scheme with the databases of size $s(n)$ for computing
% partial products with at most $k$ queries translates into a \spshort k with $s(n)$ edges for a line.
\ifnum\final=0
Chazelle \cite{cha87} and Alon and Schieber also consider a generalization of the above problem, where the input is an (undirected) tree $T$ with
an element $s_i$ of a semigroup associated with each vertex $i$. The goal is to create a space-efficient data structure
that allows one to compute the product of elements associated with all
vertices on the path from $i$ to $j$, for all vertex pairs $i,j$ in $T$.
%The more general problem reduces to finding a sparsest \spshort{k} for a rooted directed tree.
The generalized problem reduces to finding a sparsest \spshort{k} for a certain directed tree $T'$ obtained from $T$. We describe the reduction in Appendix~\ref{app:partial-products}. %Implicitly, \cite{cha87,AlonSchieber87} give matching upper and lower bounds on \spshort{k}
%sizes for all $k$ for this graph family.  As mentioned,
%\cite{AFB05} rediscovered some of the same results in the context of
%access control.
The same reduction to \spshorts{k} can be used to design space-efficient data structures for any digraph with a unique path between pairs of nodes. Our structural results imply new space-efficient data structures for $H$-minor free graphs with unique paths, and our approximation algorithms yield more space-efficient data structures for general digraphs with unique paths.
\else
 As we describe in Section~\ref{sec:previous-work}, implicit in \cite{AlonSchieber87} are tight bounds on the number of
 edges in a sparsest \spshort k for a line.
 %That work also developed techniques that helped us build \spshort{k}s
 %for directed trees and served as inspiration for generalizing this construction to planar graphs.
\fi
\\\\
{\it Organization.}  Section \ref{sec:ub} contains an overview of our algorithms. In Section \ref{sec:lb}, we give an overview of our lower bounds and the techniques involved. Section \ref{sec:minor} contains an overview of our bounds for minor-free graphs.
%
%we give a more formal sketch of the proof of the $2^{\log^{1-\epsilon} n}$-inapproximability for \spshort{3}. The missing details of the connection of \spshorts{k} to partial products and monotonicity testing are given in Appendix~\ref{app:intro-details}.
We defer the details and proofs of our results to the Appendix. In Appendix~\ref{app:intro-details}, we discuss the previous work and applications to monotonicity testing and partial products in a semigroup. Appendix~\ref{app:algorithms} contains our algorithms. In Appendix~\ref{app:khardness}, we give our $2^{\log^{1-\epsilon} n}$-inapproximability for $k > 2$, and in Appendix~\ref{app:2hardness} we give our $\Omega(\log n)$-inapproximability for $k=2$. In Appendix~\ref{app:separable}, we give the proofs of our structural results.
%
%We also defer the details of our $2^{\log^{1-\epsilon} n}$-inapproximability for \spshort{k} for general $k > 3$ to Appendix \ref{sec:lbk}. Our $\Omega(\log n)$-inapproximability for \spshort{2} is given in Appendix \ref{sec:lb2}. The proof of our results for $H$-minor free digraphs is given in Appendix \ref{sec:Hminor}.
\\\\
{\it Notation.} The {\em transitive closure} of a graph $G=(V, E)$, denoted $TC(G)$, is defined as the directed graph $(V, E')$, where $E'=\{(u, v): u\leadsto_G v\}$. Vertices $u$ and $v$ are {\it comparable} if either $(u,v) \in TC(G)$ or $(v,u) \in TC(G)$. The {\em transitive reduction} of $G$, denoted $TR(G)$, is
a digraph $G'$ with the fewest edges for which $TC(G') = TC(G)$. As shown by Aho {\it et al} \cite{agu72}, $TR(G)$ can be computed efficiently via a greedy algorithm. %by contracting each strongly connected component $C$ to a vertex $v(C)$ to get a supergraph $H$, obtaining a supergraph $H'$ by greedily removing arcs in $H$ that do not change its transitive-closure, and finally uncontracting the $v(C)$, choosing a representative vertex of $C$ to be incident to the edges incident to $v(C)$ in $H'$.
For directed acyclic graphs $TR(G)$ is unique, and $G$ is {\em transitively reduced} if $TR(G)=G$.

For $k \geq 1$, we define: $S_k(G)=\min_H \{|H| : H \text{ is a } k \text{-TC-spanner of }G\}$, and we
call  $H$ that achieves this minimum a {\em sparsest \spshort{k}}. Clearly, $S_k(G) \geq |TR(G)|$.
%
%We define a few related decision problems that we will study.
%Sofya: should be defined right away in the intro.
%We consider the optimization version of the following decision problems:  {\em \textsf{Directed
%    k-Spanner}} $= \{(G, M) : \text{Digraph }G \text{ has a } k\text{-spanner of}$
%$\text{ size at most } M\}$ and {\em \textsf{\spshort{k}}} $= \{(G, M) : \text{Digraph }G \text{ has a \spshort{k} of %size at most } M\}$.
%
%We use the directed {\em butterfly graph}, defined in \cite{woo06}, in several constructions.
%\begin{definition}\label{def:butterfly}
%The directed butterfly graph of diameter $k$ and width $w$, denoted $BF(k,w)$, has
%vertex set $[k+1] \times [w^{1/k}]^k$ and  edge set
%$\{((i,a_1,\ldots,a_k),(i+1,b_1,\ldots,b_k)) : 1 \leq i \leq k, \forall j \neq i, b_j = a_j\}$.
%\end{definition}
%
The  {\em Ackermann function} \cite{a28} is defined by: $A(1, j)=2^j$, $A(i+1, 0)=A(i, 1), $ $A(i+1, j+1)=A(i, 2^{2^{A(i+1, j)}})$.
The inverse Ackermann function is $\alpha(n)=\min \{i: A(i, 1)\geq n\}$ and the $\ith$-row inverse is $\lambda_i(n)=\min \{j: A(i, j)\geq n\}$.

\section{Overview of Algorithms for \tcspanner\ and Related Problems}\label{sec:ub}
Our $O((n \log n)^{1-1/k})$-approximation for \tcspanner \ for arbitrary $k$ is based on a new combination of convex programming and sampling. The technique also achieves an $O((n \log n)^{1-1/k})$ ratio for {\sc Directed $k$-Spanner}, {\sc Client/Server Directed $k$-Spanner}, and {\sc $k$-Diameter Spanning Subgraph}. Here we describe the result for {\sc Directed $k$-Spanner}. To achieve the same result for \tcspanner, it suffices to run the algorithm on the transitive-closure of the input digraph. Missing proofs are in Appendix~\ref{app:algorithms}.

\begin{theorem}
For any (not necessarily constant) $k > 2$, there is a deterministic polynomial-time algorithm achieving an $O((n\log n)^{1-1/k})$-approximation for {\sc Directed $k$-Spanner}.%, the client/server directed $k$-spanner problem, the $k$-diameter spanning subgraph problem, and {\sc $k$-TC-Spanner}.
\end{theorem}
We start by formulating the problem as an integer program. We briefly explain the problems with this approach and the ideas required to make it work. One can introduce binary edge variables $x_e$ for each edge $e$ in the transitive closure, and binary path variables $y_P$ for each path $P$ of length $\leq k$ in the transitive closure. One enforces the constraints $y_P \leq x_e$ for each $e \in P$, which allow a path $P$ in the spanner only if all edges along it are present. The final constraint is $\sum_P y_p \geq 1$ for all edges $(u,v) \in G$, where the sum is over paths $P$ of length $\leq k$ from $u$ to $v$. Finally, one can relax the problem to an LP, and try to round the solution.

The first problem %difficulty?
 is that the integrality gap is huge, which may be why an LP approach had not been considered before. Indeed, if there are $n$ paths of length at most $k$ between $u$ and $v$, the LP might assign each of them a value of $\Theta(1/n)$. However, we observe that if there are $r = n^{1-1/k}$ distinct paths from $u$ to $v$ of length $\leq k$, there must be $\geq r^{1/(k-1)}$ distinct vertices $w$ for which $u \leadsto w \leadsto v$. Let $BFS(v)$ denote a shortest path tree of edges directed away from $v$, together with a shortest path tree of edges directed towards $v$. Then if we sample $\tilde{O}(n/r^{1/(k-1)})$ vertices, and grow $BFS(w)$ of $2(n-1)$ edges around each sample $w$,
  we will sample a $w$ for which $u \leadsto w \leadsto v$, and the path from $u$ to $v$ along the edges in $BFS(w)$ has length at most $k$. We take the spanner $H$ to be the union of the outputs of the LP and sampling-based algorithms.
\begin{center}
\fbox{
\parbox{6in} {
%\underline{{$k$-Spanner Generation}($G$)}:
\vspace*{-4mm}
\begin{enumerate}
\addtolength{\itemsep}{-4mm}
\item $H \leftarrow \emptyset$.
\item For each edge $e \in G$, if $x_e \geq \frac{1/2}{(n \log n)^{1-1/k}}$, $H \leftarrow H \cup \{e\}$.
\item Randomly sample $r = O((n \log n)^{1-1/k})$ vertices $z_1, z_2, \ldots, z_r \in G$.
\item $H \leftarrow H \cup \left (\cup_i BFS(z_i) \right )$.
%\item
Output $H$.
\end{enumerate}
\vspace*{-5mm}
}}
\end{center}
With high probability, an edge $(u,v)$ is covered by either the LP relaxation or the sampling.
\begin{lemma}
With probability at least $1-1/n$, $H$ is a $k$-TC-spanner of $G$.
\end{lemma}
The spanner has at most $r \cdot OPT + \frac{n^2}{r^{1/(k-1)}}$ edges, where $OPT$ is the optimum of the LP. By observing that any spanner must have size $\min(OPT, n-1)$, one can guarantee that this is an $\tilde{O}(n^{1-1/k})$-approximation. Note that we assume that $G$ is connected, as otherwise we can run the algorithm separately on each component. A more careful analysis gives an $O((n \log n)^{1-1/k})$-approximation, and a simple greedy algorithm derandomizes the sampling.
\begin{lemma}
$|H| = O((n \log n)^{1-1/k}OPT)$.
\end{lemma}
The problem with this approach is that the number of variables and the size of each of the constraints grows exponentially with $k$. We replace the variables $y_P$ with $\min_{e \in P} x_e$, reducing the number of variables to $O(n^2)$. The resulting program is convex, and we use the ellipsoid algorithm with a separation oracle. The oracle, given $\vec{x}$, just needs to find one pair of vertices $(u,v)$ for which the constraint $\sum_{P:u \leadsto v} \min_{e \in P} x_e \geq 1$ is violated. It can do this by sorting the coordinates of $\vec{x}$, and counting the number of $u$-$v$ paths $P$ for which some particular $x_e$ is the minimum edge variable along $P$. For this, it iteratively removes edges $e$ from $G$ for which $x_e$ is smallest, and uses matrix multiplication to count the $u$-$v$ paths that remain in the graph.
\begin{lemma}
For any $k$, there exists a separation oracle which runs in time $\poly(n)$.
\end{lemma}
Our $\tilde{O}(n/k^2)$-approximation algorithm, which is specific to \tcspanner, works by sampling $\tilde{O}(n/k)$ vertices and selectively including $O(n/k)$ edges in the transitive closure adjacent to the samples. We also include the edges of $TR(G)$ in the spanner. A simple greedy algorithm derandomizes the sampling.
%Suppose we randomly sample $\tilde{O}(n/k)$ vertices and include shortcut edges from the samples to every comparable vertex in $G$, and from every comparable vertex in $G$ to the samples. Fix any two vertices $u$ and $v$ with $u \leadsto v$. If there is a path in $TR(G)$ of length at most $k$, we are done, since any \spshort{k} must include the edges in $TR(G)$. Otherwise, there are at least $k$ vertices $w$ for which $u \leadsto w \leadsto v$, and by sampling, with high probability we will choose one of these vertices and obtain a path of length $2$ from $u$ to $v$ through $w$. To do better, we observe that we need not include a shortcut between $w$ and every comparable vertex, but rather we may ``skip'' $r = \Theta(k)$ levels at a time. That is, we include edges between $w$ and vertices $v$ at distance exactly $i \cdot r$ from $w$ in $TR(G)$, where $i = 1, \ldots, n/r$. By averaging, there is a choice of $r = \Theta(k)$ for which this results in the addition of at most $O(n/k)$ edges from each sample. It can be shown this still suffices to obtain an \spshort{k}, since we can walk along edges in $TR(G)$, then follows two edges adjacent to a sample, and then again walk along edges in $TR(G)$. We have $\tilde{O}(n/k)$ samples, but each only contributes $O(n/k)$ edges, and using that any \spshort{k} has size $\Omega(n, |TR(G)|)$, this gives an $\tilde{O}(n^2/k^2)/n = \tilde{O}(n/k^2)$-approximation.
\begin{theorem}
For any $k$, there exists a deterministic approximation algorithm for the \textsc{$k$-TC-Spanner} problem with approximation ratio $O((n \log n)/(k^2 + k \log n))$.
\end{theorem}
%
%for any $k$. Though the program is no longer linear, it is still convex. However, the constraint that for all $u,v$ the sum of $\min_{e \in P} x_e$ over $u$-$v$ paths of length at most $k$ be at least $1$ is not of polynomial size. Nevertheless, we can still use the ellipsoid method with a separation oracle to find a point which has exponentially small Euclidean distance from a feasible point in polynomial time. The oracle, given a point outside of the convex body, finds a hyperplane separating the point from the body. Due to the path structure of the underlying graph problem, we can use matrix multiplication to efficiently count the number of paths $P$ for which $\min_{e \in P} x_i = \gamma$, where $\gamma$ is any real number. This is enough to implement the oracle.%The final idea is to exploit the faster but less accurate multiplicative weights algorithm for convex programming (see, e.g., the survey \cite{ahk} by Arora {\it et al}), rather than the ellipsoid algorithm. We
%can do it because our rounding procedure can be made to work even if some constraints are satisfied only approximately with a relatively large approximation error.

\fi
\section{Overview of Hardness Results for \tcspanner}\label{sec:lb} %Our $2^{\log^{1-\epsilon} n}$-inapproximability is significantly more complex and is our main technical contribution. %athan the corresponding inapproximability for {\sc Directed $k$-Spanner}. Our $2^{\log^{1-\epsilon} n}$-inapproximability is our main result. %and uses a completely different hard instance. %This is our main technical contribution.
This section outlines the proof of Theorem~\ref{thm:khardness}, which is our main technical contribution.  Missing details are
in Appendix~\ref{app:khardness}. At the end we briefly describe the ideas behind the inapproximability result for {\sc 2-TC-Spanner} that appears in Appendix~\ref{app:2hardness}.
\begin{theorem}\label{thm:khardness}
For any fixed $\eps \in (0,1)$, the size of the sparsest \spshort{k} cannot be approximated to within a
factor of  $2^{\log^{1-\eps}n}$ unless $NP \subseteq DTIME(n^{\polylog n})$.
\end{theorem}
\subsection{The Construction and its Motivation}
Since {\sc $k$-TC-Spanner} is a special case of {\sc Directed $k$-Spanner}, which is $\Theta(\log n)$-inapproximable for $k=2$ and $2^{\log^{1-\epsilon}n}$-inapproximable for $k\geq 3$, it is natural to ask whether the hard instances of {\sc Directed $k$-Spanner} from \cite{Kort,ElkPel,ElkinPeleg07} can be used to prove hardness for {\sc $k$-TC-Spanner}. It turns out that all these instances have very small $k$-TC-spanners. We demonstrate it for the instance used in the proof of $\Omega(\log n)$-hardness for {\sc Directed $k$-Spanner}, %This is the instance used by Kort... for k=2. For k>2 his instances are more complicated, but still have small TC-spanners.
which works via a reduction from {\sc Set-Cover}.

Let $G$ be a bipartite digraph for {\sc Set-Cover} with $n$ vertices (``sets'') on the left, $n$ vertices (``elements'') on the right, and edges from left to right. Let $I$ be a set of $i$ new independent vertices, for some value $i$, and let $L$ be a directed line on $k-1$ new vertices. Call the first vertex of $L$ the head, and the last vertex the tail. Include directed edges (1) from the tail of $L$ to every set in $G$, (2) from every vertex of $I$ to the head of $L$, and (3) from every vertex of $I$ to the sets and the elements of $G$. Call the constructed digraph~$G'$.

Observe that in $G'$, all directed edges except those from $I$ to $G$ must be included in the directed $k$-spanner, as such edges form the unique path between their endpoints. %The number of such edges that must be included is thus $t + (k-2) + n + L$, where $L$ is the number of edges in $G$.
At this point, the only pairs of vertices at distance larger than $k$ are those from a vertex in $I$ to an element of $G$. Since these vertices are adjacent in $G'$, there must be a path of length at most $k$ in the spanner. The only possible path is from the vertex in $I$ to a vertex of $G$. It is easy to see that adding exactly $OPT$ edges from each vertex in $I$ to the sets of $G$ is necessary and sufficient to obtain a spanner, where $OPT$ is the size of the minimum set-cover. By making $i$ sufficiently large, the size of the spanner is easily seen to be $\Theta(i \cdot OPT)$, and thus one can approximate {\sc Set-Cover} by approximating {\sc Directed $k$-Spanner}, so the problem is $\Omega(\log n)$-inapproximable. %It is just as easy to get $2^{\log^{1-\epsilon}n}$-inapproximability for {\sc Directed $k$-Spanner} for any $k \geq 3$.

However, there is a trivial \spshort{k} for this instance! Indeed, by transitivity we can simply connect the head of $L$ to each of the elements of $G$. This is a $k$-TC-spanner of size proportional to the number of vertices in $G'$. Thus, the best one could hope for with this instance is to show $\Omega(1)$-hardness for {\sc $k$-TC-Spanner}. For similar reasons, the instance showing $2^{\log^{1-\epsilon} n}$-inapproximability for {\sc Directed $k$-Spanner} also cannot establish anything beyond $\Omega(1)$-hardness for {\sc $k$-TC-Spanner}.
%cannot be used to show anything -inapproximability for {\sc $k$-TC-spanner}.

%As these examples suggest, inapproximability for {\sc $k$-TC-spanner} is much more delicate than that of {\sc Directed $k$-Spanner}. This should be expected, as we obtain an $\tilde{O}(1)$-approximation for the former problem for any $k = \Omega(\sqrt{n})$, whereas the latter problem remains $2^{\log^{1-\epsilon} n}$-inapproximable up to $k = n^{1-\delta}$ for any constants $\epsilon, \delta > 0$.
In the example above there are many paths to cover (those from $I$ to elements of $G$), but a few ``shortcut'' edges cover them all. Ideally, we would have many paths to cover, and each shortcut edge could only cover a single path. Hesse's digraph requiring a large number of shortcuts to reduce its diameter \cite{h03} satisfies the desired condition. His idea was to associate vertices with a subset $V$ of vectors in $\mathbb{R}^d$ such that $(u,v) \in E$ iff $u-v$ is an extreme point of the $d$-dimensional ball of integer points. By the properties of an extreme point, a shortcut can cover at most one path from a large family of shortest paths. %for a large family of shortest paths, a shortcut can cover at most one path from the family. %This yields a digraph with many shortcuts required to reduce its diameter to $k$.
%In the appendix, we observe an alternate construction based on additive combinatorics achieving the same qualitative %result as Hasse. Namely, vertices are associated with the integers $[n]$ and we choose a dense set $X \subset [n]$ so %that the only solution to $i + x_1 + x_2 + \cdots + x_k = i + kx$ for $i \in [n]$ and $x, x_1, \ldots, x_k \in X$ is %$x = x_1 = \cdots = x_k$.

However, to achieve an inapproximability result, we need better structured graphs. We use {\it generalized butterflies} defined in \cite{woo06}. In these digraphs vertices are identified with coordinates $[n^{1/k}]^k \times [k+1]$, and an edge connects $u = (u_1, \ldots, u_k, i)$ to $v = (v_1, \ldots, v_k, i+1)$ iff for all $j \neq i$, $u_j = v_j$. We say a vertex $(u_1, \ldots, u_k, i)$ is in {\it strip} $i$. It is easy to see that there is a unique shortest path of length $k$ from any $u$ in strip $1$ to any $v$ in strip $k+1$. Moreover, any shortcut is on at most $n^{1-2/k}$ such paths because if it connects a vertex in strip $i$ with a vertex in strip $i+\ell$ (where $\ell\geq 2$) it fixes all but $i-1$ coordinates of $u$ and all but $k+1-(i+\ell)$ coordinates of $v$. Thus, $ \geq n^{1+2/k}$ shortcuts are needed to reduce the diameter to $k-1$.

%{\bf Generalized MIN-REP.} %S: why is it generalized?
{\bf %Noise-resilient
Reduction from {\sc MIN-REP}.}
To get $2^{\log^{1-\epsilon} n}$-inapproximability, we reduce
from the {\sc MIN-REP} \ problem. An $(n, r, d, \maxisol )$-{\sc MIN-REP} instance is a bipartite graph
of maximum degree $d$ in which the left part can be partitioned into sets ${\cal A}_1, \ldots, {\cal
  A}_r$ and the right part into sets ${\cal B}_1, \ldots, {\cal B}_r$, so that $|{\cal A}_i| =
|{\cal B}_i| = n/r$ for all $i\in[r]$. To describe the last parameter $\maxisol$, call a vertex {\em
  isolated} if its degree is $0$, and {\em non-isolated} otherwise. Let $\maxisol({\cal A}_i)$ be the
inverse of the fraction of non-isolated vertices in ${\cal A}_i$.
%ratio between the total number of vertices and the number of non-isolated vertices in ${\cal A}_i$.
Then $\maxisol$ is the minimum such $\maxisol({\cal A}_i)$. Define the {\it supergraph} %$H$ % H is a spanner
to have nodes ${\cal A}_1,
\ldots, {\cal A}_r, {\cal B}_1, \ldots, {\cal B}_r$, with a {\em
  superedge} $({\cal A}_i, {\cal B}_j)$ iff there is a node in
${\cal A}_i$ adjacent to a node in ${\cal B}_j$. A {\it rep-cover}
is a vertex set $S$ in the graph such that whenever $({\cal A}_i,
{\cal B}_j)$ is an edge in the supergraph, there is an edge between some $u, v
\in S$ with $u \in {\cal A}_i$ and $v \in {\cal B}_j$. A solution to
{\sc MIN-REP} is a smallest rep-cover, and its size is denoted by OPT. The problem is
$2^{\log^{1-\epsilon}n}$-inapproximable~\cite{ElkPel}. %for any $\epsilon \in (0,1)$.

%In {\sc MIN-REP}, there is a bipartite graph $G$ for which the left part can be partitioned into {\it clusters} ${\cal A}_1, \ldots, {\cal A}_r$ and the right part into {\it clusters} ${\cal B}_1, \ldots, {\cal B}_r$, so that $|{\cal A}_i| = |{\cal B}_i| = n/r$ for all $i$. Define the {\it supergraph} $H$ to have vertices ${\cal A}_1, \ldots, {\cal A}_r, {\cal B}_1, \ldots, {\cal B}_r$, with an edge $({\cal A}_i, {\cal B}_j)$ iff there is a vertex in ${\cal A}_i$ adjacent to a vertex in ${\cal B}_j$. %We can assume that $H$ is $d$-regular for some (non-constant) value of $d$, which we call the {\it superdegree}.
%A {\it rep-cover} is a set $S$ of vertices in $G$ so that whenever $({\cal A}_i, {\cal B}_j)$ is an edge in $H$, there is an edge between some $u, v \in S$ with $u \in {\cal A}_i$ and $v \in {\cal B}_j$. The goal is to find the smallest {\it rep-cover}, and the problem is $2^{\log^{1-\epsilon} n}$-inapproximable.

As a first attempt, we construct a graph of diameter $k+2$ as follows. We attach a disjoint copy of a generalized butterfly of diameter $k-1$ to each ${\cal A}_i$ in the {\sc MIN-REP} instance graph; that is, we identify the vertices in ${\cal A}_i$ with the last strip of the butterfly. We call the vertices in the butterfly at distance $x$ from ${\cal A}_i$ the $x$-th {\it shadow} of ${\cal A}_i$. Next, for each ${\cal B}_j$, we attach what we call a {\it broom}. This is a 3-layer graph, where the two leftmost layers form a bipartite clique, and the right layer consists of degree-1 nodes, called {\em broomsticks}, attached to vertices in the middle layer. Each vertex in the middle layer has the same number of broomsticks attached to it. Each ${\cal B}_j$ is identified with the left layer of a disjoint broom.
% All of the brooms are disjoint.
All edges are directed from the shadows of the ${\cal A}_i$ towards the broomsticks (left to right). Call the resulting digraph $G$.

We would like to argue that the minimum \spshort{k} $H$ of $G$ is formed as follows. Let $S$ be a minimum rep-cover of the underlying MIN-REP instance. For each $s \in S$, if $s$ is in an ${\cal A}_i$, include all shortcuts from the $2$-shadow of ${\cal A}_i$ to $s$ which are in the transitive closure of $G$. Otherwise ($s$ is in a ${\cal B}_j$), include all shortcuts from $s$ to the broomsticks of ${\cal B}_j$. By balancing the number of broomsticks with the size of $2$-shadows, one can show $H$ has size $|S|f(n, k)$, where $f(n, k)$ is an easily computable function. Since $S$ is a rep-cover, $H$ is a \spshort{k}. If $H$ were optimal, then approximating its size within some factor would approximate {\sc MIN-REP} within the same factor.

It turns out that $H$ is not optimal, and so our first attempt does not work. However, by modifying $G$ via the transformations below, and by looking at a related \spshort{k} $H$ of the modified $G$, we can show that any \spshort{k} must have size $\Omega(|H|/\log n)$ for constant $k$. %We will prove that just to reduce the distance from vertices in the $k$-shadows of the ${\cal A}_i$ to vertices in the broomsticks of the ${\cal B}_j$ requires the addition of $\Omega(|H|/\log n)$ shortcuts.
Since {\sc MIN-REP} is $2^{\log^{1-\epsilon} n}$-inapproximable, this still gives $2^{\log^{1-\epsilon} n}$-hardness.

To prove this, we need to argue that most vertices $v$ in the $k$-shadows do not ``benefit'' from traversing other shortcuts to reach the broomsticks. This requires a classification of all alternative routes from such $v$ to broomsticks. Given that $v$ is in a generalized butterfly, these routes are well-understood. However, for a generic {\sc MIN-REP} instance, most of these routes do indeed lead to a much smaller \spshort{k}!

To rule out the alternative routes, we ensure that OPT and the
four parameters of the {\sc MIN-REP} instance each lie in a narrow
range. In Theorem~\ref{thm:transformations}, we prove that {\sc
  MIN-REP} with the required parameter restrictions is inapproximable
by giving a reduction from an unrestricted {\sc MIN-REP} instance. It works by
carefully interleaving the following five operations on a ``base''
{\sc MIN-REP} instance with unrestricted parameters: (1) disjoint copies, (2) dummy vertices inside clusters, (3) blowup inside clusters with matching supergraph, (4) blowup inside clusters with complete supergraph, and (5) tensoring. Each operation increases one or several parameters by a prespecified factor, and together they give us five degrees of freedom to control the range of OPT and the four parameters of {\sc MIN-REP}.
%To ``rule out'' these alternative routes, the first thing we need to do is control the number of clusters $r$, have an upper bound on the degree of a vertex in the {\sc MIN-REP} instance, have the condition that $OPT$ lie in a narrow, yet polynomially-sized, range, as well as some additional properties detailed below. To achieve these properties, we carefully interleave the following five operations on a ``base'' {\sc MIN-REP} instance: (1) disjoint copies, (2) dummy vertices inside clusters, (3) blowup inside clusters with matching supergraph, (4) blowup inside clusters with complete supergraph, and (5) tensoring.

\begin{theorem}[{\bf Noise-Resilient MIN-REP is hard}]\label{thm:transformations}
 Fix any $\kappa\in(0, 1)$ and $R, D, M, F \in (0,
  1-\kappa)$ satisfying  $F\in (R, 2R)$ and $D + M + F < 1$. {\em Noise-Resilient {\sc MIN-REP}} is a family of  $(n,r,d,m)$-{\sc
    MIN-REP} instances with $r\in[n^R, n^{R+\kappa}]$, $d\in[n^D, n^{D+\kappa}]$, $m\in[n^M,
  n^{M+\kappa}]$, and $OPT\in[n^F, n^{F+\kappa}]$. This problem is $2^{\log^{1-\eps}n}$-inapproximable for all
  $\eps \in (0,1)$ unless $NP \subseteq DTIME(n^{\polylog n})$.
% Let $(n,r,d,\maxisol)$-{\sc MIN-REP} with $r \in (n^{R_1}, n^{R_2})$, $d \in (n^{D_1},n^{D_2})$,
% $\maxisol \in (n^{L_1},n^{L_2})$ and with $OPT \in (n^{F_1}, n^{F_2})$, for some given parameters $R_i, D_i, L_i, F_i\in (0,1)$, $i\in [2]$ satisfying the following constraints: $R_1 < R_2< F_1 < F_2< R_1 + D_1$, $D_1 < D_2$, $L_1 < L_2$,
%  $F_2 < 2\cdot R_1$, and $D_2 + L_2 + F_2 < 1$. Then  $(n,r,d,\maxisol)$-{\sc MIN-REP} is $2^{\log^{1-\eps} n}$-inapproximable for any $\eps \in (0,1)$.
\end{theorem}
\newcommand{\parone}{\ensuremath{{\delta_1}}} %parameter of T1
\newcommand{\partwo}{\ensuremath{{\delta_2}}} %parameter of T2
\newcommand{\parthree}{\ensuremath{{\delta_3}}} %parameter of T3
\newcommand{\parfour}{\ensuremath{{\delta_4}}} %parameter of T4
\newcommand{\parfive}{\ensuremath{{\delta_5}}} %parameter of T5
\newcommand{\genpar}{\ensuremath{\delta}} %general transformation parameter
\newcommand{\exepsilon}{\ensuremath{\epsilon'}} %parameter I do not understand

\newcommand{\nrinstance}{\ensuremath{{\cal I}_0}} %noise-resilient instance of MIN-REP
\newcommand{\spinstance}{\ensuremath{{\cal I}}} %specialized instance of MIN-REP (after T4 is applied)
\newcommand{\tcinstance}{\ensuremath{{\cal G}}} %TC-spanner instance

%If a {\sc MIN-REP} instance has parameter values that satisfy the constraints in the theorem
%statement above, we say that it is an instance of the ``noise-resilient'' {\sc MIN-REP} problem.
 The variant of {\sc MIN-REP} in Theorem~\ref{thm:transformations} is called ``noise-resilient''
%The nomenclature arises from the fact that
because even if many vertices in the sets ${\cal A}_i$ and ${\cal
  B}_j$ are adversarially deleted in an instance of this problem, the
minimum rep-cover does not shrink significantly. This property helps
us rule out many alternative routes in the TC-spanner, though we will
need to change our graph $G$. Our reduction from noise-resilient {\sc
  MIN-REP} to \tcspanner\ for $k>2$ consists of two steps: first we
produce a {\em specialized} {\sc MIN-REP} instance $\spinstance$ from an arbitrary instance $\nrinstance$ of noise-resilient {\sc MIN-REP}, and then we construct a \tcspanner\ instance $\tcinstance$ by carefully adjoining generalized butterfly graphs on the left and broom graphs on the right of $\spinstance$.

{\bf From Noise-resilient {\sc MIN-REP} to Specialized {\sc MIN-REP}.}
% Fix parameters $\mu = \frac{1}{3(k-1)}, \delta = \frac{k+\mu}{k+1}, \gamma = \frac{\mu}{2(k-1)}$,
% $\eta = \frac{1}{2} \min(\frac{\gamma}{3},(\delta -
% \frac{1}{2})\frac{1}{k-1})$ and a sufficiently small constant $\kappa <
% \frac{\gamma-3\eta}{10\delta}$.
Set $\delta = \frac{k-1}{k-\frac{1}{4}}$, $\eta = \frac{\delta}{2(4k-4)(4k-2)}$, and $\zeta =
\delta \left(\frac{4k-5}{4k-4} + \frac{1}{4k-2}\right)$.  Let $\kappa$ be a sufficiently small
positive constant which will be chosen in the course of the proof. We start from an $(n_0,r_0,d_0,\maxisol_0)$-instance $\nrinstance$ of
noise-resilient {\sc MIN-REP} with optimum $OPT_0$, where $n_0 = n^\delta, r_0 \in [n^{\delta/2},
n^{\delta/2 + \kappa}], d_0 \in [n^\eta, n^{\eta + \kappa}]$, $m_0 \in [n^{2\eta},
n^{2\eta + \kappa}]$, and $OPT_0\in[n^{\zeta},n^{\zeta + \kappa}]$.  By instantiating \thmref{transformations} with $R =\frac{1}{2}, D
= \frac{\eta}{\delta}, M = \frac{2\eta}{\delta},$ $F = \frac{\zeta}{\delta}$ and $\kappa$, we obtain that the
$(n_0,r_0,d_0,\maxisol_0)$-{\sc MIN-REP} problem is $2^{\log^{1-\eps}  n}$-inapproximable unless $NP
\subseteq DTIME(n^{\polylog n})$.  The conditions on the parameters in \thmref{transformations} are satisfied because $\zeta \in
(\frac{\delta}{2},\delta)$ and $\eta +2\eta + \zeta < \delta$.

We transform $\nrinstance$ to a specialized $(n,r,d,\maxisol)$-{\sc MIN-REP} instance $\spinstance$ by applying
on $\nrinstance$
the transformation $T_4$ defined in the proof of \thmref{transformations} in
Appendix~\ref{app:khardness}.  More precisely, set $\spinstance=T_4(\nrinstance, n^{1-\delta})$.
By definition of $T_4$, graph $\spinstance$ has $n$ vertices, $r = r_0$, $d = d_0n^{1-\delta}$ and $m=m_0$.
The transformation results in a bipartite graph $\spinstance$ with nodes partitioned into clusters
${\cal A}_1,\ldots, {\cal A}_r$
on the left, and ${\cal B}_1, \ldots, {\cal B}_r$ on the right. Each ${\cal A}_i$ and ${\cal B}_j$
is a union of $n^{1-\genpar}$ {\em groups} $A_{i, s}$ and $B_{j, s}$, respectively, with $s\in
[n^{1-\genpar}]$. Each group $A_{i, s}$ and $B_{j, s}$, for $i, j\in [r]$, $s\in [n^{1-\genpar}],$
is a copy of ${\cal A}_i$ and, respectively, ${\cal B}_j$, from the original instance
$\nrinstance$. For each edge $(u, v)$ with $u \in {\cal A}_i$ and $v \in {\cal B}_j$ of
$\nrinstance$, graph $\spinstance$ has edges between  the  copy of $u$ in $A_{i, k_1}$ and the copy
of $v$ in $B_{j, k_2}$, for all $k_1, k_2 \in [n^{1-\genpar}]$. This completes the description of
the specialized {\sc MIN-REP} instance ${\spinstance}$.

\begin{figure}[t]
\centering
\includegraphics[trim = 0mm 50mm 0mm 70mm, clip, width=10cm]{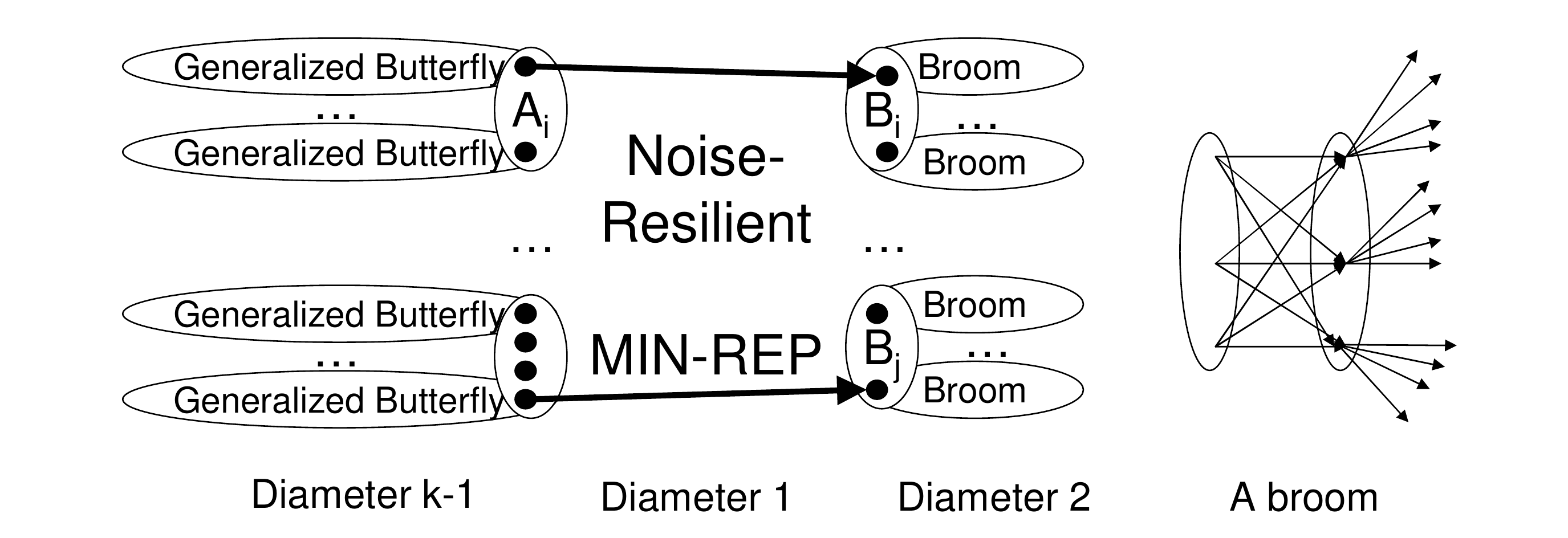}
\caption{The TC-spanner instance $\tcinstance$, and an example of a broom.}
\label{fig:butterfly}
\end{figure}

{\bf From Specialized {\sc MIN-REP} to $k$-{\sc TC-Spanner}.}
From $\spinstance$, we construct a graph $\cal G$ of diameter $k+2$ as follows.
We first attach a disjoint generalized butterfly of diameter $k-1$, denoted $BF(A_{i,s})$, to
each group $A_{i, s}$ in $\spinstance$, for all $i\in [r]$, $s\in
[n^{1-\genpar}]$. That is, we identify vertices in $A_{i,s}$ with the
last strip of $BF(A_{i,s})$ in the way discussed below. Denote by $BF({\cal A}_i)=\cup_{s} BF(A_{i, s})$ the set of all the vertices attached in this manner to the cluster ${\cal A}_i$.
Let $BF^j(A_{i, s})$ be the vertices in strip $j$ of the butterfly $BF(A_{i, s})$, where $BF^k(A_{i, s})=A_{i, s}$, and let $BF^j({\cal A}_i)= \cup_s BF^j(A_{i, s})$. We call the vertices in the butterfly $BF(A_{i, s})$ at distance $x$ from $A_{i, s}$ the $x$-th {\it shadow} of $A_{i,s}$.  Call the in-degree as well as out-degree of the vertices in the butterflies  $d_* \defeq (\frac{n^{\genpar}}{r})^{\frac{1}{k-1}}$.

Next, for each ${\cal B}_{i, s}$, we attach a broom, denoted $BR(B_{i,s})$. \mnote{Are $p$ and $t$ used later?}\enote{ Fixed} More specifically, each vertex in $ B_{i, s}$ is connected to the vertices of a set $BR^{k+2}(B_{i,s})$ of size $d_*$, and each vertex $v\in BR^{k+2}(B_{i,s})$ is connected to a disjoint set %$C_{i, s, v}$
of nodes, called broomsticks, of size $d_*$. Let $BR^{k+3}(B_{i,s})$ be the set of broomsticks adjacent to $BR^{k+2}(B_{i,s})$.
Let $BR^{k+2}({\cal B}_i)=\cup_s BR^{k+2}(B_{i,s})$ %Also, $BR^{k+3}(B_{i,s})= \cup_{v\in P_{i,s}} C_{i, s, v}$
and $BR^{k+3}({\cal B}_i)=\cup_s BR^{k+3}(B_{i,s})$. Identify {\em layer} $V_j$ with $\cup_{i, s} BF^j( A_{i,s})$ for $j\in [k]$, layer $V_{k+1}$ with $\cup_{i,s} B_{i,s}$, and layer $V_j$ with $\cup_i BR^{j}({\cal B}_i)$ for $j\in\{k+2,k+3\}$. Direct all the edges from $V_i$ to $V_{i+1}$.  See Figure \ref{fig:butterfly}.

%To achieve this, we take a {\sc MIN-REP} instance of size
%$n^{\genpar}$, copy the clusters ${\cal A}_i$ and ${\cal B}_j$ each
%$s = n^{1-\genpar}$ times for some parameter $\genpar$, obtaining $s$
%{\it groups} $A_{i, 1}, \ldots, A_{i,s}$ and $B_{j, 1}, \ldots,
%B_{j,s}$. Next, for each $a, b \in [s]$, we connect $A_{i, a}$ to
%$B_{j,b}$ as ${\cal A}_i$ and ${\cal B}_j$ are connected in the
%underlying {\sc MIN-REP} instance. For each $i \in [s]$, we create a
%new cluster ${\cal A}_i'$ (${\cal B}_i'$) by taking the disjoint
%union of the $A_{i,j}$ ($B_{i,j}$), i.e., ${\cal A}_i' = \cup_{j=1}^s
%A_{i,j}$ (${\cal B}_i' = \cup_{j=1}^s B_{i,j}$). Thus, while $r$, and
%$OPT$ are the same, if we delete a vertex in some ${\cal A}_i'$ (or
%${\cal B}_i'$), we do not change $OPT$ until we delete a vertex from
%each of the groups. In $G$, we now have a disjoint butterfly (broom)
%attached to each $A_{i,j}$ ($B_{i,j}$). The value of $H$ now changes,
%since we must take a rep-cover $S$ and duplicate it  once for each level. In each level we connect $v$ to its $2$-shadow if $v$ is in an $A_{i,j}$, and otherwise to its broomsticks.

{\bf Attaching butterflies}. Recall that we identify vertices in $A_{i,s}$ with the last strip $BF^k(A_{i,s})$ of a disjoint butterfly, for all $i\in [r]$, $s\in [n^{1-\genpar}]$. The mapping from $A_{i,s}$ to $BF^k(A_{i,s})$ is constructed in Appendix~\ref{app:khardness}. Here we explain the requirements we impose on the mapping.
Recall that each each group $A_{i,s}$ has $\leq\frac{n^{\genpar}}{rm}$ non-isolated vertices.  For our analysis, \textit{each} vertex in $BF^{k-1}(A_{i,s})$ must be adjacent to  $\leq\frac{d_*}{m}$ non-isolated vertices in $A_{i,s}$.
% To handle one family of routes, we also need many isolated vertices in each level. We can achieve this by placing many isolated vertices in our base {\sc MIN-REP} instance, and observing that our other transformations preserve isolation.
The isolated vertices help us control the number of routes with shortcut edges from the $x$-shadows to the
$(x-2)$-shadows, for some $x > 2$,
% Indeed, such routes necessarily traverse vertices in the $1$-shadow. Connecting vertices in the $1$-shadow to many isolated vertices %in each level %what does that mean?
%decreases the number of routes that can be extended from such shortcuts. We choose the density of isolated vertices to rule out this family of routes.
since connecting vertices in the $1$-shadow to many isolated vertices decreases the number of comparable
pairs in the first and last layers of $\tcinstance$ connected by a path containing such a shortcut
edge.

{\bf A sparse TC-spanner $\cal H$ for the $k$-{\sc TC Spanner} instance $\tcinstance$. } \snote{Corrected; please check} \enote{Checked} Let $S_0$ be a smallest rep-cover of $\nrinstance$ of size OPT. Recall that each ${\cal A}_i$ and ${\cal B}_j$ is replicated $n^{1-\delta}$ times in $\spinstance$. Let $S$ be the set of all replicas in $\spinstance$ of vertices in $S_0$.
 Consider a $k$-TC-spanner $\cal H$ of $\tcinstance$ that contains shortcuts from the nodes in layer $V_{k-2}$ to their descendants in $S\cap V_k$, and from the nodes in $S\cap V_{k+1}$ to their descendants in $V_{k+3}$.
 %, copied across the different groups.
The  Rep-cover Spanner Lemma (Lemma~\ref{lem:ubspshort}) shows that
 $|{\cal H}|= O(OPT ~n^{1-\genpar} (\frac{n^{\genpar}} {r} )^{\frac{2}{k-1}})$.

\subsection{Path Analysis and Rerandomization}
The next lemma shows that the \spshort{k} $\cal H$ %in the Rep-cover Spanner Lemma (Lemma~\ref{lem:ubspshort})
defined above and analyzed in \lemref{ubspshort} is nearly optimal.

\begin{lemma}\label{lemma:lowerbound}
Any \spshort{k} ${\cal K}$ of ${\tcinstance}$ has $|{\cal K}|  =
\Omega\left(OPT n^{1-\genpar} \left ( \frac{n^{\genpar}} {r} \right )^{\frac{2}{k-1}}/\log n\right)$.
\end{lemma}

We introduce a bit of notation.  A \spshort k for ${\tcinstance}=V_1\cup V_2 \cup \cdots \cup V_{k+3}$  is
  built by adding shortcut edges $(u, v)$ between comparable $u$ and
  $v$, where $u\in V_i, v\in V_{i+\ell}$ and $\ell\geq 2$. For given $\ell,i$, we classify such a shortcut edge as {\em type $\ell\&i$}. Since $\tcinstance$ has diameter $k+2$, a $k$-TC-spanner for $\tcinstance$ remains a $k$-TC-spanner when a type $\ell\&i$
% Given a \spshort k $\cal H$ on $\tcinstance$, observe that any type $l\&i$
edge $(u, v)$ with $\ell\geq 4$ is replaced by a type $3\& i$ edge $(u, v')$, where $v'$ is a
predecessor of $v$. Therefore, it is enough to consider $k$-TC-spanners with shortcut edges only of
types $2\&i$ for $1 \leq i \leq k+1$ and $3\&i$ for $1\leq i\leq k$.
%Notice that any two comparable vertices $v_1$ and $v_f$ in a \spshort k $\cal H$ must be connected by some length-$k$ %path $\pi=(v_1, v_2, \ldots v_f)$, with $v_i\in V_{i_j}$ for some $i_j$'s, and that these paths could use at least one %edge of type $3\&i$, or at least one edge of type $2\& i$ and one edge of type $2\& j$, for some $i<j<k-1.$
 Say a path $\pi$ from $V_1$ to $V_{k+3}$ is of {\em type $(\ell\& i)$ } if it uses an edge of type $\ell\&i$ (with $\ell\in \{2, 3\})$, and $\pi$ is of  {\em type
$(2\&i, 2\& j)$ } if it uses edges of types $2\&i$ and $2\&j$, $i<j$. \snote{Almost not used (other than for canonical paths; might be worth removing.}
%Notice that a path can be of multiple types.
Notice that the \spshort{k}
constructed in \lemref{ubspshort} contains only edges of type $2 \& (k-2)$ and $2 \& (k+1)$.

\smallskip
\noindent{\em Proof of Lemma~\ref{lemma:lowerbound}.}
Given a \spshort{k} ${\cal K}$ of ${\cal G}$ with
$o\left(\frac{n^{1-\genpar} d_{*}^2}{\log n}\right)OPT$ edges, we show
that we can construct a {\sc MIN-REP} cover for ${\spinstance}$ of
size $o(OPT)$, which is a contradiction (recall that $d_*= (\frac
{n^{\genpar}}{r})^{\frac{1}{k-1}}$). We will accomplish this by a series of transformations which
modify ${\cal K}$ into a \spshort{k} that uses only shortcut edges of the form $2 \& (k-2)$ and
$2 \& (k+1)$.   The process increases the size of the \spshort{k} only by a logarithmic factor.
Finally, we show that from the modified \spshort{k}, one can extract a {\sc MIN-REP} cover of size
$o(OPT)$ for $\spinstance$, the desired contradiction.

We call a superedge $({\cal A}_i,{\cal B}_j),$ where  $i,j\in [r]$, \textit{deletable with
  respect to} ${\cal K}$ if at
least $1/4$ of the vertex pairs $(u,v)  \in BF^1({\cal A}_i) \times BR^{k+3}({\cal B}_j)$ have a
path between them in ${\cal K}$ of length at most $k$ and of type other than $(2\&(k-2),2\&(k+1))$.  Our first
step is to show that such cluster pairs can be essentially ignored.

% The proof can be divided into two components: a path analysis lemma and a rerandomization lemma.
\begin{lemma}[{\bf Path Analysis Lemma}]\label{lemma:deletion}
The number of deletable superedges with respect to ${\cal K}$ is $o(OPT)$.
\end{lemma}
\begin{proof}[Proof Sketch]
We call a path {\em canonical} if it contains shortcut edges of types both $2\&(k-2)$ and $2\&(k+1)$; otherwise, a path is {\em alternative}. \snote{Renamed good and bad paths; note: used also for paths NOT between $V_1$ and $V_{k+3}$}
Observe that every alternative path contains a shortcut edge from one of the following three categories: (1) edges that connect vertices in $V_{i}$ and $V_{j}$, where $i\le k$ and $j\ge k+1$; (2) edges of type $3\&i$  where $i\le k-3$; (3) edges of type $2\&i$  where $i\le k-3$.
\ignore{
 Let ${\cal S}$ be the set of all the shortcut types that belong to one of the above three categories. For each shortcut type $S\in {\cal S}$, let $Del(S) = \{(\alpha,\beta) \in [r]^2 | \mbox{ at least }\frac{1}{4|S_B|} \mbox{ fraction of pairs }
(u,v) \in BF^1({\cal A}_\alpha) \times BR^{k+3}({\cal B}_\beta) \mbox{ have an alternative path with a shortcut of
type }S \}.$  By a union bound and the definition of the deletable superedges, the total number of deletable superedges is at most $\sum_{S\in S_B} Del(S)$. Hence it suffices to show that for all $S\in {\cal S}$, $Del(S)=o(1) OPT.$ For $S\in {\cal S}$, let $C(S)=\{ (u,v)\in V_1 \times V_{k+3}
\mid \exists \mbox{ an alternative path between } u \mbox{ and } v \mbox{ using
  a shortcut of type } S\}.$ From the definition of $Del(S)$ and
the facts that  $|BF^1({\cal A}_\alpha)|=\frac{n}{r}$ and $|BR^{k+3}({\cal B}_{\beta})|=n^{1-\genpar} d_{*}^2$, we have
$|C(S)|\ge |Del(S)| \frac{1}{4|S_B|}\frac{n}{r}~ n^{1-\genpar}~ d_{*}^2.$ Then we obtain upper bound of $|C(S)|$ in terms of $OPT$ for each type $S\in {\cal S}$ by considering the three categories separately, which in tern proves that $|Del(S)|=o(1)OPT$.
%\end{proof}
 }
 %
 %an alternative path contains at least one shortcut among the following three cases. The first case is
 %shortcuts crossing both $V_k$ and $V_{k+1}$, i.e. one of the shortcut types from $3\&(k-2)$ ,
 %$3\&(k-1)$ , $2\&(k-1)$ , $2\&k$ , and $2\&k$. The second case is shortcuts of type  $3\&\ell$
 %where $\ell\le k-3$. The third case is shortcuts of type $2\&\ell$ where $\ell\le k-3$. Let $S_B$
 %be the set of all the above shortcuts types.
Let $S_B$ be the set of all shortcut edge types included in the three cases.
 By analyzing the three cases separately, we show that for any $S\in S_B$, the number of superedges
 $(A_i,B_j)$, $(i,j)\in [r]^2$, such that at least a $\frac{1}{4|S_B|}$ fraction of pairs $(u,v) \in
 BF^1({\cal A}_i) \times BR^{k+3}({\cal B}_j)$ have an alternative path containing a shortcut of type $S$, is
 $o(OPT)$. Then by a union bound over $S\in S_B$, we prove the lemma. In the analysis for the first
 type, we use the fact that the degree of each non-isolated vertex of $V_k$ is at least
 $n^{1-\genpar}$ which is bigger than $d_{*}$. When $S$ is of the second type, we need the facts
 that the out-degree of each vertex in $V_k$ is at most $d_0n^{1-\genpar}$ and that
 $n^{\eta}=o(d_{*})$. For the third case, we use the facts that for any vertex $v$ in $V_{k-1}$ the
 number of non-isolated vertices in $V_k$ that $v$ is connected to is at most $\frac{d_{*}}{m}$, and
 that $n^{\eta}=o(n^{2\eta})$.
\ignore{
Let ${\cal S}$ be the set of all the shortcut types that belong to one of the above three categories. For each shortcut type $S\in {\cal S}$, let $Del(S) = \{(\alpha,\beta) \in [r]^2 | \mbox{ at least }\frac{1}{4|S_B|} \mbox{ fraction of pairs }
(u,v) \in BF^1({\cal A}_\alpha) \times BR^{k+3}({\cal B}_\beta) \mbox{ have an alternative path with a shortcut of
type }S \}.$  By a union bound and the definition of the deletable superedges, the total number of deletable superedges is at most $\sum_{S\in S_B} Del(S)$. Hence it suffices to show that for all $S\in {\cal S}$, $Del(S)=o(1) OPT.$ For $S\in {\cal S}$, let $C(S)=\{ (u,v)\in V_1 \times V_{k+3}
\mid \exists \mbox{ an alternative path between } u \mbox{ and } v \mbox{ using
  a shortcut of type } S\}.$ From the definition of $Del(S)$ and
the facts that  $|BF^1({\cal A}_\alpha)|=\frac{n}{r}$ and $|BR^{k+3}({\cal B}_{\beta})|=n^{1-\genpar} d_{*}^2$, we have
$|C(S)|\ge |Del(S)| \frac{1}{4|S_B|}\frac{n}{r}~ n^{1-\genpar}~ d_{*}^2.$ Then we obtain upper bound of $|C(S)|$ in terms of $OPT$ for each type $S\in {\cal S}$ by considering the three categories separately, which in tern proves that $|Del(S)|=o(1)OPT$.}
\end{proof}

Next, form the graph $\tcinstance'$ from $\tcinstance$ by  deleting all edges of $\tcinstance$ connecting ${\cal A}_i$ to ${\cal B}_j$, for all the deletable superedges
$({\cal A}_i,{\cal B}_j)$ with respect to ${\cal K}$. Similarly,
obtain a graph ${\cal K'}$ from ${\cal K}$ as follows: for all deletable superedges $({\cal
  A}_i,{\cal B}_j)$ with respect to ${\cal K}$, delete all edges of $\cal K$ connecting
${\cal A}_i$ to ${\cal B}_j$,  and also delete all shortcuts in $\cal K$ of types other than
$2\&(k-2)$ and $2\&(k+1)$.  Note that for any cluster pair
$({\cal A}_i,{\cal B}_j)$ of $\tcinstance'$, either there are no edges between vertices in ${\cal
  A}_i$ and ${\cal B}_j$ or at least  $\frac{3}{4}$ of the pairs in
$BF^1({\cal A}_i)\times BR^{k+3}({\cal B}_j)$ are connected by a canonical path. Also define a {\sc
  MIN-REP} instance ${\spinstance'}$ from ${\spinstance}$ by deleting all edges in $\spinstance$
corresponding to all the deletable superedges with respect to ${\cal K}$.

For $\mu\in [0,1]$, we say a subgraph of $TC(\mathcal{G})$  is a $\mu$-{\em good \spshort{k} \snote{Rename? E.g., $\mu$-cover} for $\mathcal{G}$} if for every $(i,j) \in [r]^2$
such that $\mathcal{A}_i$ and $\mathcal{B}_j$ are comparable in $\mathcal{G}$, at least a $\mu$ fraction of pairs
$(u,v) \in BF^1(\mathcal{A}_i) \times BR^{k+3}(\mathcal{B}_j)$ are connected by
canonical paths in the subgraph. E.g., the graph $\mathcal{K}'$ is a $\frac{3}{4}$-good \spshort{k} for $\mathcal{G}$.

\begin{lemma}[{\bf Rerandomization Lemma}]
If a $\frac{3}{4}$-good \spshort{k}  $\mathcal{K}'$  for $\mathcal{G}'$ is given, then there exists
$\mathcal{K}''$, a $1$-good \spshort{k} for $\mathcal{G}'$, such that $|\mathcal{K}''| \leq
O(|\mathcal{K}'|\cdot \log n)$.
\end{lemma}
\begin{proof}[Proof Sketch]
% Consider any pair $(i,j)$ such that $\mathcal{A}_i$ and $\mathcal{B}_j$ are comparable.  We know
% that at least $\frac{3}{4}$ of the comparable pairs in $BF^1(A_i)$ and $BR^{k+3}(B_j)$ have a
% $(2\&(k-2), 2\&(k+1))$ path between them.  In fact, by a Markov argument, at least $\frac{1}{2}$ of
% the groups $A_{i,s}$ in $\mathcal{A}_i$ must have at least half of the vertices in $BF^1(A_{i,s})$
% comparable to a shortcut edge from $V_{k-2}$ to $V_k$ and similarly, at least $\frac{1}{2}$ of the
% groups $B_{j,s'}$ in $\mathcal{B}_j$ must have at least half of the vertices in $BR^{k+3}(B_{j,s'})$
% incident on a shortcut edge from $V_{k+1}$.
To construct $\mathcal{K}''$ from $\mathcal{K}'$, we let $\mathcal{K}''$ be the union of
$O(\log n)$ random transformations of the edges of $\mathcal{K}'$.  Each transformation $\Pi_r$ will
keep the edges of $\mathcal{G}'$ invariant but move the shortcut edges.  Thus, when we let
$\mathcal{K}'' = \cup_{r=1}^{O(\log n)} \Pi_r(\mathcal{K}')$, the edges of $\mathcal{K}''$ are still
a subset of the edges in $TC(\mathcal{G}')$.   The goal of the random transformations is to ensure that in
$\Pi_r(\mathcal{K}')$, with a constant probability, each vertex
in $BF^1(\mathcal{A}_i)$ can reach a vertex in $V_{k-2}$ incident to a shortcut edge, and each vertex in
$BR^{k+3}(\mathcal{B}_j)$ is incident to a shortcut edge from $V_{k+1}$.  We achieve this by randomly
permuting the groups inside the clusters $\mathcal{A}_i$ and $\mathcal{B}_j$ and by randomly
permuting the edges of the butterfly and broom graphs attached to $\mathcal{A}_i$ and
$\mathcal{B}_j$.  After these random transformations, any two vertices $u$ and $v$ in
$BF^1(\mathcal{A}_i)$ and $BR^{k+3}(\mathcal{B}_j)$ are connected by a canonical path
with probability at least $\frac{1}{16}$.  Hence, $\mathcal{K}''$ has such a path
between them with probability $1-\frac{1}{\poly(n)}$.  The union bound over all possible
$(u,v)$ and $(i,j)$ shows that the desired $\mathcal{K}''$ with the claimed size exists.
\end{proof}

\snote{New summary; pls check.} Now that the \spshort{k} is 1-good, it is easier to reason about rep-covers of the underlying {\sc MIN-REP} instance. Recall that ${\cal G}'$ has $n^{1-\delta}$ copies of {\sc MIN-REP} instance $\spinstance'$  embedded in it. Moreover, many pairs of vertices in layers $V_1$ and $V_{k+3}$ rely on each instance to connect.
We partition the shortcut edges of $\mathcal{K}''$ into $n^{1-\delta}d^2_*$ parts, according to which groups of vertex pairs in $V_1\times V_{k+3}$ they can help to connect. By averaging, one of the parts has $o(OPT)$ shortcut edges, and can be used to extract a rep-cover of $\spinstance'$ of size $o(OPT)$. By including two vertices for each of the $o(OPT)$ deleted superedges, we obtain a rep-cover for $\spinstance$ of size $o(OPT)$. This is a contradiction.
%
%For $s \in [n^{1-\genpar}]$, define $\mathcal{K}''_s$ to be the subgraph of $\mathcal{K}''$ induced by $\cup_{i=1}^r BF(A_{i,s}) \cup BR(B_{i,s})$.  The $\mathcal{K}''_s$ are clearly disjoint.  By averaging, there exists an $s$ such that $|\mathcal{K}''_s| \leq o(OPT \cdot d_*^2)$. Also, let $\mathcal{G}'_s$ be the restriction of $\mathcal{G}'$ to $\cup_{i=1}^r BF(A_{i,s}) \cup BR(B_{i,s})$. Now we perform a few local edge transformations of $\mathcal{G}'_s$ to extract a rep-cover of $\spinstance$ of size $|\mathcal{K}''_s|/d_*^2 = o(OPT)$, and by including two vertices for each of the $o(OPT)$ deleted superedges, obtain a rep-cover for $\spinstance$ of size $o(OPT)$. This is a contradiction.

% Hence, we have a $1$-good \spshort{k} for $\mathcal{G}'$ of size $o(OPT\cdot
% n^{1-\genpar} \cdot d_*^2)$.  Finally, we claim we can use this graph to obtain a rep-cover for $\spinstance$ of size $o(1)\cdot OPT$, a contradiction. This concludes the proof.

\begin{lemma}[{\bf Rep-cover Extraction Lemma}]
Given $\mathcal{K}''$, a $1$-good \spshort{k} for $\mathcal{G}'$, of size $o(OPT \cdot n^{1-\genpar}
\cdot d_*^2)$, there exists a {\sc MIN-REP} cover of $\spinstance$ of size $o(OPT)$. \hfill $\Box$
\end{lemma}

Our $\Omega(\log n)$-inapproximability for {\sc 2-TC-Spanner}, described in Appendix~\ref{app:2hardness}, is based on a reduction from {\sc Set-Cover} instead of {\sc MIN-REP}. Our hard instance is a generalized butterfly of diameter $2$ attached to an instance of transformed {\sc Set-Cover}. We identify strip $3$ of the butterfly with the sets in the %{\sc Nice Set-Cover}
instance, and using ideas similar to our proof for $k > 2$ for ruling out alternative routes, show that up to a constant factor, the optimal \spshort{2} contains only shortcuts from strip $1$ to a minimum set-cover in strip $3$.
%\subsection{Obtaining specialized hard {\sc MIN-REP} instances}
%We obtain specialized hard to approximate {\sc MIN-REP} instances satisfying the parameters specified in the construction, by starting
%with any unconstrained hard instance and applying several transformations to it. These transformations take turns modifying some of the parameters, while fixing the others. By composing these transformations, we show that {\sc MIN-REP} is hard to approximate for a wide range of parameters. %We formalize this statement in \thmref{transformations} below.

\section{Overview of Structural Results}\label{sec:minor}
%As formalized in \lemref{spanners-to-monotonicity-tests}, a sparse \spshort{2}
%for a DAG yields a low-query monotonicity tester with respect to the
%partial order defined by the DAG.
%This was the technique used in \cite{FLNRRS02} to obtain
%monotonicity testers for planar posets.
In \cite{FLNRRS02}, the authors implicitly give \spshorts{2} for
planar digraphs of size $O(n^{3/2} \log n)$ using Lipton-Tarjan separators. For planar digraphs, our first idea is to instead use Thorup's planar separators \cite{thorup-separators} in conjunction with the efficient \spshorts{k} for the directed line of Alon and Schieber \cite{AlonSchieber87} to recursively construct \spshorts{k} of size $O(n\log^2 n)$.
%The
%\spshort{2} construction is recursive: partition the graph into smaller pieces
%using the path separator,  construct \spshort{2}s for the planar graphs induced
%on each partition recursively, and add a small number of edges to the graph so
%that comparable vertex pairs not already considered have a path of length at
%most $2$ between them.
More generally, for $H$-minor-free graphs, using an idea in \cite{thorup-separators}, we take an
arbitrary rooted spanning tree $T$ of the digraph $G$ and use it to partition $G$ into a union of
edge-disjoint digraphs so that in each part $G_i$, if one undirects the edges of $G_i$, any
undirected root path of $T$ restricted to $G_i$ is the union of at most two dipaths. Next, instead of Thorup's planar separators, we use a result of Abraham and Gavoille \cite{AbrahamGavoille} that provides a ``path separator'' for undirected $H$-minor-free graphs.

However, the Abraham-Gavoille separators cannot be directly applied, since they do not provide enough flexibility in the structure of the separators.  That is, these separators consist of a sequence of unions of minimum cost paths, where the cost function on the edges is arbitrary but specified in advance. We, however, need to adaptively change the cost function during the construction of the separator. Indeed, in the outermost level of recursion we need the path separator to lie on $T$, as otherwise the path separator may be the union of $\Omega(n)$ dipaths in the underlying digraph, and therefore we cannot use the efficient \spshort{k} of Alon and Schieber \cite{AlonSchieber87} for the directed line in order to efficiently recurse. Thus, we specify the cost of an edge in $T$ to be $1$, while outside of $T$ it is $\infty$. However, when we partition $G$ into subgraphs in the recursion, it may be that two vertices in the same subgraph no longer have a path contained in $T$. Since the cost function is fixed and the cost of any path between these two vertices is now $\infty$, a path separator in the recursive step need not be contained in $T$, and so it may not be the union of a small number of dipaths. Thus, we again cannot efficiently recurse. If, however, we could change the cost function in the recursive step, we could define a new rooted tree in each subgraph and base our cost function on that. We observe that the proof of the Abraham-Gavoille separators can be used to show that their path separators satisfy this stronger property. \snote{Is this cost-function discussion reflected in the appendix? } \enote{No, it doesn't seem so. It might be difficult to relate this general description to the actual construction indeed.}

\begin{theorem}
If $G$ is an $H$-minor-free graph, then it has a
$2$-TC-spanner of size $O(n \log^2 n)$ and, more generally, a $k$-TC-spanner of size $O(n\cdot \log
n \cdot \lambda_k(n))$ where $\lambda_k(\cdot)$ is the $k$-row
inverse Ackermann function.
\end{theorem}

\newcommand{\T}[1]{$\mathcal{T}(#1)$}
\newcommand{\spannersize}[2]{\ensuremath{S_{#1}\left({#2}\right)}}

\paragraph{Acknowledgments.} We would like to thank Michael Elkin, Vitaly Feldman, Cyril Gavoille, Piotr Indyk, T.S. Jayram, Elad Hazan, Ronitt Rubinfeld, and Adam Smith for helpful discussions.

\iffalse
\begin{table}
\begin{tabular}{|l||l|l|l|}
  \hline
  % after \\: \hline or \cline{col1-col2} \cline{col3-col4} ...
  Setting of $k$ & Implied by previous work & This paper & Notes \\
  \hline
  \hline
 $k=2$  & $O(\log n)$ \cite{KortPel}  & $\Omega(\log n)$ % unless P=NP
 &\\
  \hline
 constant $k>2$    & &   $\Omega(2^{\log^{1-\eps} n})$ %unless  NP$\subseteq $DTIME($n^{poly \log n}$)
 &\\
 \hline
 \hline
 $k=3$  & $O(n^{2/3} \polylog n)$ \cite{ElkPel}  & $O((n\log n)^{2/3})$
 %,   $\Omega(2^{\log^{1-\eps} n})$ unless  NP$\subseteq $DTIME($n^{poly \log n}$)
 & \multirow{2}{*}{applies to {\sc Directed $k$-Spanner}}\\
 \cline{1-3}
 $k>3$  & $O(n)$ [trivial]  & $O((n\log n)^{1-1/k})$
 %,   $\Omega(2^{\log^{1-\eps} n})$
 %unless  NP$\subseteq $DTIME($n^{poly \log n}$)
 &\\
  \hline
 $k=\Omega\left(\frac{\log n}{\log\log n}\right)$  & $O(n)$ [trivial]  & $O \left (\frac{n \log n}{k^2 + k \log n}\right )$ & separation from {\sc Directed $k$-Spanner}\\
  \hline
\end{tabular}
\caption{Summary of Results on Approximability of {\sc $k$-TC-Spanner}}\label{table:results}\end{table}

\fi

\ifnum\final=0
\bibliographystyle{abbrv}
\else
%Following line adds "References" to the table of contents
\addcontentsline{toc}{section}{References}

\bibliographystyle{abbrv}
\fi

\appendix
\section{Missing Details from Section~\ref{sec:intro}}\label{app:intro-details}
\subsection{Previous Work on Other Related Problems}\label{app:previous-work}
 Dodis and Khanna \cite{dk99} study the problem of finding the minimum-cost subset of missing edges that can be added to a (directed) graph $G$, with costs and
lengths associated to the missing edges, so as to ensure that that there is a path of length at most $k$ between every pairs of nodes (not only those connected in $G$). Observe that {\sc $k$-TC-Spanner} is a special case of that problem: we can let $G$ be the transitive reduction (see definitions below)%in Section~\ref{sec:notation})
of the input graph to {\sc $k$-TC-Spanner}, for all edges in the transitive closure of $G$ set the length to 1 and cost to 1, and for the remaining edges set the length to $k$ and cost to 0. Given this instance, the algorithm of Dodis and Khanna will produce a $k$-TC-spanner.
%We can run their algorithm on transitively-reduced digraphs (see definitions below)
%$G$ to obtain \spshort{k}s by assigning lengths to $1$, the cost of transitive-closure edges to $1$,
%and the remaining costs to $\infty$.
However, the guarantee on the resulting \spshort{k} size is only
 $\leq |G| + O(OPT n \log k)$, where $OPT$ is the number of missing edges that need to be
added. If $|G| = OPT = \Theta(n)$, their algorithm may return a \spshort{k} with $\Omega(n^2)$
edges.
%, which gives a trivial $O(n)$ for the ratio of the size of the optimal \spshort{k} to $|G|$.
Thus, in general, the resulting approximation ratio is no better than $O(n)$.
Since their problem is more general, their hardness results do not apply to TC-spanners.

Chekuri {\it
et al} \cite{cegs08} give an $O(p^{1/2 + \epsilon})$-approximation algorithm for the directed
Steiner network problem where, given a digraph and node pairs $(s_1, t_1), \ldots, (s_p, t_p)$,
the goal is to connect all pairs with as few edges as possible. We can reduce {\sc $k$-TC-Spanner} to this problem by specifying all comparable pairs of nodes in levels $1$ and $k+2$ in the $k+1$-extension
of $G$ (see definition 5.5 of \cite{dk99}). However, their ratio is only $O(n^{1+\epsilon})$ when $p =
\Omega(n^2)$, and thus the resulting ratio for {\sc $k$-TC-Spanner} is no better than $O(n)$.

\subsection{ Sparse \spshort{2}s Imply Efficient Monotonicity Testers}\label{app:monotonicity}
In this section we restate and prove Lemma~\ref{lem:spanners-to-monotonicity-tests}, referred
to in the introduction. The proof of the lemma
explains how to use \spshort{2}s to obtain efficient monotonicity testers.
\begin{lemma}
If a directed acyclic graph $G_n$ has a \spshort{2} $H$ with $s(n)$ edges, then
there exists a monotonicity tester on $G_n$ that runs in time $O\left(\frac{s(n)}{\epsilon n}\right)$.
\end{lemma}
\begin{proof}
The tester selects $\frac{4 s(n)}{\epsilon n}$ edges of the \spshort{2} $H$
uniformly at random. It queries function $f$ on the
endpoints of all the selected edges and rejects if and only if one of the selected edges is
{\em violated} by $f$, that is, $f(x)>f(y)$ for an edge
$(x,y)$.

If the function $f$ is monotone on $G_n$, the algorithm always accepts. The crux of the proof
is to show that functions that are $\eps$-far from monotone are rejected with probability at least
$\frac 2 3$. Let $f: V_n\to\mathbb{R}$ be a function that is $\eps$-far from monotone. It is
enough to demonstrate that $f$ violates at least $\frac{\eps n}2$ edges in $H$. Then each selected
edge is violated with probability $\frac{\eps n}{2s(n)},$ and the lemma follows by elementary
probability theory.

Denote the transitive closure of $G$ by $TC(G)$. We say a vertex $x\in V_n$ is assigned a {\em bad}
label by $f$ if $x$ has an incident violated edge in $TC(G_n)$; otherwise, $x$ has a {\em good} label. Let $V'$ be a set of vertices with
good labels. Observe that $f$ is monotone on the induced subgraph $G'=(V',E')$ of $TC(G)$.
This implies (\cite{FLNRRS02}, Lemma 1) that $f$ can be changed into a monotone function by
modifying it on at most $|V_n-V'|$ vertices. Since $f$ is $\eps$-far from monotone, it shows
that there are at least $\eps n$ vertices with bad labels.

Every function that is $\eps$-far from monotone has a matching $M$ of $\frac{\eps n}2$ violated edges in $TC(G)$ \cite{DGLRRS99}.
We will establish an injection from the set of edges in $M$ to the set of violated edges in $H$.
For each edge $(x,y)$ in the matching, consider the corresponding path from $x$ to $y$ of length at most 2 in
the \spshort{2} $H$. If the path is of length 1, $(x,y)$ is the violated edge in $H$ corresponding
to the matching edge $(x,y)$. Otherwise, let $(x,z,y)$ be a path of length 2 in $H$. At least
one of the edges $(x,z)$ and $(z,y)$ is violated, and we map $(x,y)$ to that edge. Since $M$ is a matching, all edges
in $M$ have distinct endpoints. Therefore, each edge in $M$ is mapped to a unique violated edge in $TC(G)$.
Thus, the \spshort{2} $H$ has at least $\frac {\eps n} 2$ violated edges, as required.
\end{proof}

The fact that $H$ is a \spshort{2} is crucial for the proof. If it was a
\spshort{k} for $k>2$, the path of length $k$ from $x$ to $y$ might not have any violated edges incident to
$x$ or $y$, even if $f(x)>f(y)$. Consider $G_{2n}=(V_{2n},E)$ where
$V_{2n}=\{x_1,\ldots,x_{2n}\},E=\{(x_i,x_{n})\ | \ i< n \}\cup(x_n,x_{n+1})\cup \{(x_{n+1},x_j)\ | \ j>n+1 \}.$
$G_n$ is a \spshort{3} for itself. Now set $f(x_i)=1$ for $i\leq n$ and $f(x_i)=0$ otherwise.
Clearly, this function is $\frac 1 2$-far from monotone, but only one edge, $(x_n,x_{n+1})$ is
violated in the \spshort{3}.

\subsection{Partial Products in a Semigroup}\label{app:partial-products}
Chazelle \cite{cha87} and Alon and Schieber also consider a generalization of the above problem, where the input is an (undirected) tree $T$ with
an element $s_i$ of a semigroup associated with each vertex $i$. The goal is to create a space-efficient data structure
that allows to compute the product of elements associated with all
vertices on the path from $i$ to $j$, for all vertex pairs $i,j$ in $T$.
%The more general problem reduces to finding a sparsest \spshort{k} for a rooted directed tree.
The generalized problem reduces to finding a sparsest \spshort{k} for a certain directed tree $T'$ obtained from $T$ by appending a new vertex to each leaf, and then selecting an arbitrary root and directing all edges away from it. A \spshort{k} for $T'$
with $s(n)$ edges yields a preprocessing scheme with space complexity $s(n)$ for computing products on $T$ with at most $2k$
queries as follows. The database stores a product $s_{v_1}\circ \cdots \circ s_{v_t}$ for each \spshort{k} edge
$(v_1,v_{t+1})$ if the endpoints of that edge are connected by the path $v_1,\cdots,v_t,v_{t+1}$ in $T'$.
Let $LCA(u,v)$ denote the lowest common ancestor of $u$ and $v$ in $T$. Compute the product corresponding
to a path from $u$ to $v$ in $T$ as follows: (1) if $u$ is an ancestor of $v$ (or vice versa) in $T$,
 query the products corresponding to the \spshort k edges on the shortest path from $u$ to a child of $v$ (from $v$ to a
 child of $u$, respectively); (2) otherwise, make queries corresponding to the \spshort k edges on the shortest path from
 $LCA(u,v)$ to a child of $u$ and on the shortest path from a child of $LCA(u,v)$ nearest to $u$ to a child of $u$. This
 gives a total of at most $2k$ queries.

\section{Approximation Algorithms for \tcspanner\ and Related Problems}\label{app:algorithms}
\subsection{Algorithm for \directedspanner}
We give the algorithm for {\sc Directed $k$-Spanner}, which is a more general problem than {\sc $k$-TC-Spanner}. We then mention the extensions to other problems.
%In this section, we study the problem of approximating $S_k(G)$, given an input digraph $G$
%and an integer $k$.
\iffalse
We give an $O((n \log n)^{1-1/k})$-approximation ratio for \textsf{\spshort{k}} by
providing an approximation algorithm for the more general problem of {\sc Directed
$k$-Spanner}. Previously, it was unknown how to achieve an $o(n)$-approximation ratio for the {\sc
Directed $k$-Spanner} problem for $k > 3$. Our algorithm can be modified to obtain the first sublinear approximation ratio for $k > 3$ for the related problems considered in \cite{ElkinPeleg05}.
\fi
\begin{theorem}\label{thm:lpmain}
For any (not necessarily constant) $k > 2$, there is a deterministic polynomial-time algorithm achieving an $O((n\log n)^{1-1/k})$-approximation for {\sc Directed $k$-Spanner}.%, the client/server directed $k$-spanner problem, the $k$-diameter spanning subgraph problem, and {\sc $k$-TC Spanner}.
\end{theorem}

\begin{proof}
Consider the following integer programming formulation. Let $OPT$ be the size of an optimal $k$-spanner of $G$. For each edge $e$ in the input digraph $G$, we have a variable $x_e$ indicating whether $x_e$ occurs in the $k$-spanner. Also, for each
(not necessarily simple) path $P$ containing at most $k$ edges, we have a variable $y_P$ indicating whether all of the
edges of $P$ occur in the $k$-spanner.
\begin{center}
$\min \sum_{e \in G} x_e$
\begin{eqnarray}
& \textrm{s.t.} & \forall e = (u,v) \in G, \ \sum_{P \textrm{ from }u \textrm{ to }v, \ |P| \leq k}y_P \geq 1\\
& & \forall P = (e_1, e_2, \ldots, e_r), \ y_P \leq e_1, \ y_P \leq e_2, \ldots, y_P \leq e_r\\
& & \forall e \ \forall P, \ x_e, \  y_P \in \{0,1\}
\end{eqnarray}
\end{center}
The first constraint ensures that there is at least one path of length at most $k$ spanning each edge $(u,v)$ in the spanner, while the second constraint only allows a path to be included if each of its edges is also in the spanner. Thus, any solution to this program is a $k$-spanner, and vice versa. Notice, however, that the number of path variables grows exponentially with $k$. We can instead write this as as the following integer program:
\begin{center}
$\min \sum_{e \in G} x_e$
\begin{eqnarray}
& \textrm{s.t.} & \forall e = (u,v) \in G, \sum_{P = (e_1, \ldots, e_r,) \textrm{ from }u \textrm{ to }v, \ |P| \leq k} - \min(x_{e_1}, x_{e_2}, \ldots, x_{e_r}) \leq -1\\
& & \forall e, \ x_e \in \{0, 1\}
\end{eqnarray}
\end{center}
Now the number of variables is $m$, where $m$ is the number of edges of $G$. We relax the constraints $x_e \in \{0,1\}$ to $x_e \in [0,1]$. The resulting set $K$ we optimize over is convex since $\vec{x} \in [0,1]^m$ and the functions $-\min(x_{e_1}, \ldots, x_{e_r})$ are convex (as is their sum). We reduce the problem to a feasibility one by taking the convex set $K' = K \cap \{\vec{x} : \sum_{e \in G}x_e \leq t\}$, for a parameter $t$ which we do binary search over.

The problem is still that we sum over a number of terms which can be exponential in $k$. However, we design a separation oracle $A$ which does the following: given a point $\vec{x} \in \mathbb{R}^m$, $A$ decides whether $\vec{x} \in K'$, and if not, provides an $\vec{a} \in \mathbb{R}^m$ and $b \in \mathbb{R}$ for which $\langle \vec{a}, \vec{x} \rangle < b$ but $\langle \vec{a}, \vec{y} \rangle \geq b$ for all $\vec{y} \in K'$. We will design an oracle for this task running in time $\poly(n)$. Later, we explain the details of algorithm $A$.

There are several folklore polynomial-time algorithms for solving a convex program given a separation oracle $A$. We use the ellipsoid algorithm, which, given $\epsilon > 0$, runs in time $\poly(n)\log \frac{1}{\epsilon}$ and, if the program is feasible, returns an $\vec{x^*}$ for which the $\ell_2$-norm $|\vec{x^*} - \vec{x}|_2$ is at most $\epsilon$, where $\vec{x}$ is a feasible solution. Setting $\epsilon = n^{-\Theta(k)}$ guarantees that $\vec{x^*} \in [0,1]^m$, that $\sum_{v \in G} x^*_e \leq n^{-\Theta(k)} + \sum_{v \in G} x_e$, and for all $e = (u,v) \in G$,
\begin{eqnarray*}
\sum_{P = (e_1, \ldots, e_r) \textrm{ from }u \textrm{ to }v, \ |P| \leq k} - \min(x^*_{e_1}, x^*_{e_2}, \ldots, x^*_{e_r}) & \leq & n^k\epsilon - \sum_{P = (e_1, \ldots, e_r) \textrm{ from }u \textrm{ to }v, \ |P| \leq k} \min(x_{e_1}, \ldots, x_{e_r})\\
& \leq & n^{-\Theta(k)} -1.
\end{eqnarray*}
Assuming we have an oracle $A$ described above, the following is our algorithm {\sc $k$-Spanner Generation} to construct a directed $k$-spanner $H$ of $G$. For a vertex $v \in G$, we use $BFS(v)$ to denote the set of edges along a shortest path tree\footnote{For a directed graph, this means we take a shortest path tree of edges directed away from $v$, together with a shortest path tree of edges directed towards $v$.} rooted at $v$. Clearly $|BFS(v)| = O(n)$.
\begin{center}
\fbox{
\parbox{6.5in} {
\underline{{$k$-Spanner Generation}($G$)}:
\begin{enumerate}
\addtolength{\itemsep}{-2mm}
\item $H \leftarrow \emptyset$.
\item For each edge $e \in G$, if $x_e^* \geq \frac{1/2}{(n \log n)^{1-1/k}}$, $H \leftarrow H \cup \{e\}$.
\item Randomly sample $r = O((n \log n)^{1-1/k})$ vertices $z_1, z_2, \ldots, z_r \in G$.
\item $H \leftarrow H \cup \left (\cup_i BFS(z_i) \right )$.
%\item
 Output $H$.
\end{enumerate}
}}
\end{center}

\begin{lemma}\label{lem:spanner}
With probability at least $1-1/n$, $H$ is a $k$-TC-spanner of $G$.
\end{lemma}

%Our algorithm is based on an integer linear programming (ILP) formulation for the %{\sc Directed
% $k$-Spanner} problem.

\begin{proof}
Consider an edge $(u,v) \in G$. Suppose there are at most $(n \log n)^{1-1/k}$ different $u-v$ paths $P$ of length at most $k$. By constraint (4) of the convex program and the relationship between $\vec{x}$ and $\vec{x^*}$, there exists such a $P = (e_1, \ldots ,e_r)$ for which $\min(x^*_{e_1}, \ldots, x^*_{e_r}) \geq 1/(n \log n)^{1-1/k} - n^{-\Theta(k)} \geq 1/(2(n \log n)^{1-1/k})$, for some $r \leq k$. Thus, this path $P$ is included in $H$ in step 2 of {\sc $k$-Spanner Generation}.

Now suppose there are more than $(n \log n)^{1-1/k}$ different $u-v$ paths $P$ of length at most $k$. Let $W_{u,v} = \{w_1, \ldots, w_s\}$ be the set of vertices lying on at least one such path. The number of $u-v$ paths of length at most $k$ that can be formed from $s$ vertices is at most $s^{k-1}$. So, $s^{k-1} = \Omega((n \log n)^{1-1/k})$, or $s = \Omega((n \log n)^{1/k})$. The probability that $\{z_1, z_2, \ldots, z_r\} \cap W_{u,v} = \emptyset$ is at most $(1-s/n)^r \leq e^{-rs/n} \leq e^{-\Omega(\log n)} \leq 1/n^3$, for an appropriate choice of constants.

By a union bound, with probability at least $1-1/n$, all edges $(u,v) \in G$ for which there are more than $(n \log n)^{1-1/k}$ different $u-v$ paths $P$ of length at most $k$ satisfy $\{z_1, z_2, \ldots ,z_r\} \cap W_{u,v} \neq \emptyset$. Conditioned on this event, for each such $(u,v) \in G$ let $z(u,v)$ be an arbitrary element in $\{z_1, z_2, \ldots, z_r\} \cap W_{u,v}$. Then the path $u \leadsto z(u,v) \leadsto v$ along the edges of $BFS(z(u,v))$ is of length at most $k$. Indeed, there is a path $P$ of length at most $k$ from $u$ to $v$ which contains $z(u,v)$, and the path from $u$ to $v$ along the edges of $BFS(z(u,v))$ cannot be any longer than the length of $P$.
\end{proof}

% Again, we defer the proof of the following lemma to Appendix {\ref{sec:lpproofs}}

\begin{lemma}\label{lem:size}
$|H| = O((n \log n)^{1-1/k}OPT)$.
\end{lemma}

\begin{proof}
Let $OPT'$ be the optimum of the convex program. Clearly $OPT' \leq OPT$. In step (2) of {\sc
$k$-Spanner Generation}, at most $2(n \log n)^{1-1/k}OPT' \leq 2(n \log n)^{1-1/k}OPT$ edges are added
to $H$. In step (4), $O(rn) = O((n \log n)^{2-1/k})$ edges are added to $H$. So, $|H| = O((n \log
n)^{1-1/k})(OPT + n)$. We may assume that $OPT \geq n-1$, as otherwise $G$ is not connected and we
can run {\sc $k$-Spanner Generation} on each of its connected components. Therefore, $|H| = O((n \log
n)^{1-1/k}OPT).$
\end{proof}
As $|H| \geq OPT$, these lemmas show that {\sc $k$-Spanner Generation} is a randomized
$O((n\log n)^{1-1/k})$-approximation algorithm for {\sc Directed $k$-Spanner}. The algorithm can
be derandomized by greedily choosing the $z_i$ in step 3.
\begin{lemma}\label{lem:lpderam}
For any constant $k > 2$, {\sc Directed $k$-Spanner} has a deterministic $O((n\log n)^{1-1/k})$-approximation algorithm.
\end{lemma}
\begin{proof}%[Proof of \lemref{lpderam}]
In step (3) of {\sc $k$-Spanner Generation}, instead of sampling $r$ random vertices, we do the following. For each edge $(u,v) \in G$ with more than $(n \log n)^{1-1/k}$ simple $u-v$ paths $P$ of length at most $k$, find the set $W_{u,v}$ of all vertices lying on such a path between $u$ and $v$. This can be done by computing $BFS(w)$ for each vertex $w \in G$, and checking if $u \leadsto w \leadsto v$ along the edges of $BFS(w)$ is a path of length at most $k$. By averaging, there is a vertex $z_1$ which occurs in an $\Omega((\log n)^{1/k}/n^{1-1/k})$ fraction of the sets $W_{u,v}$. Choose $z_1$, delete the sets $W_{u,v}$ containing $z_1$, and repeat. This greedy algorithm finds $z_1, \ldots, z_r$ with $r = O((n \log n)^{1-1/k})$.
\end{proof}
The technique can also be extended to other spanners variants.
\begin{lemma}\label{lem:lpextend}
For all constant $k > 2$, there are deterministic $O((n\log n)^{1-1/k})$-approximation algorithms for {\sc Client/Server Directed $k$-Spanner}, {\sc $k$-Diameter Spanning Subgraph}, and \tcspanner.
\end{lemma}
\begin{proof}%[Proof of \lemref{lpextend}]
In the client/server problem, we only wish to span a subset of edges of $G$, called {\it client edges}, and we may only use a subset of edges of $G$ for spanning, called {\it server edges}. To modify our algorithm, we have a constraint in the linear program for each client edge rather than for all edges, and we only consider paths along server edges. In {\sc $k$-Diameter Spanning Subgraph}, all pairs of vertices $(u,v)$ for which $v$ is reachable from $u$ need to be connected by a path of length at most $k$. For this we impose constraint (1) for all pairs rather than just all edges. Finally, {\sc $k$-TC-Spanner} is a special case of {\sc Directed $k$-Spanner} when the input is transitively closed.
\end{proof}
Moreover, {\sc $k$-Spanner Generation} is polynomial time provided that we can find $\vec{x^*}$ in polynomial time. For this, it suffices to show that the running time of the separation oracle is polynomial.

The separation oracle $A$ first checks whether $\vec{x} \in [0,1]^m$ and $\sum_{e \in G} x_e \leq t$ in $\poly(n)$ time, and provides an appropriate hyperplane if any of these constraints are violated. Assume, then, that all of these constraints are satisfied. Let Count$(G,u,v,k)$ be an algorithm which outputs the number of $u-v$ paths of length at most $k$. The number of $u-v$ paths of length exactly $i$ is just the $(u,v)$-th entry of $M^i$, where $M$ is the adjacency matrix of $G$. Thus, we can implement Count$(G, u, v, k)$ in $\poly(n)$ time. For a subset $S$ of edges of $G$, Count$(G \setminus S, u, v, k)$ counts the number of $u-v$ paths of length at most $k$ in $G$ which do not use the edges in $S$.

The oracle sorts the coordinates of $x$, obtaining $x_{e_1} \leq x_{e_2} \leq \cdots \leq x_{e_m}$. Let $S_0 = \emptyset$, and for $i \geq 1$, $S_i = S_{i-1} \cup \{e_i\}$. The oracle computes Count$(G \setminus S_i, u, v, k)$ for all $i \geq 0$. From this information, for each $j \geq 1$ the oracle can extract $c_j$, the number of $u-v$ paths of length at most $k$ whose minimum is achieved by $x_{e_j}$. Indeed, observe that $c_j$ is just the number of $u-v$ paths of length at most $k$ in $G \setminus S_{j-1}$ minus the number of $u-v$ paths of length at most $k$ in $G \setminus S_j$. Algorithm $A$ can now check if the constraint corresponding to $(u,v)$ is satisfied, and it does this for each $(u,v) \in G$. If the constraint for some $(u,v)$ is not satisfied, for all $i$ we set the $i$-th coordinate of the hyperplane $\vec{a}$ to be $c_i$, and the scalar $b$ to be $1$. This $\vec{a},b$ pair satisfy the desired constraints, and are output by $A$. Note that $A$ runs in $\poly(n)$ time for any $k$.
\end{proof}
\subsection{\tcspanner\ Algorithm for Large $k$}
For large $k$, we have the following better approximation, which is specific to the $\textsf{$k$-TC Spanner}$ problem.
\begin{theorem}\label{thm:largek-approx}
For any $k \geq 6$, there exists a deterministic approximation algorithm for the \textsf{$k$-TC
  Spanner} problem with approximation ratio $O((n \log n)/(k^2 + k \log n))$.
\end{theorem}

\begin{proof}
Let $G$ be the input digraph.  Assume, w.l.o.g., that $G$ is connected. We construct $S$, a $k$-TC-spanner for the graph $G$ such that
$|S|/S_k(G) \leq O((n \log n)/(k^2 + k \log n))$.
Set $k' = ck$ for $c$ to be determined.
Let $G'$ be $G$ with each directed cycle
contracted to a vertex, and let $H = TR(G')$.  For each vertex $v \in V(H)$, define a set $S_v$ of vertices
such that: (i) $|S_v| \leq 4n/k'$, (ii) for each $u \in S_v$, $v \leadsto_H u$, and (ii) for any
vertex $w$ such that $v \leadsto_H w$, there
exists $w' \in S_v$ with $d_H(w,w') \leq k'/4$.  One can easily see such a set exists by averaging
and it can be efficiently constructed. Next, we define another set of vertices $W \subseteq V(H)$
such that for every pair of vertices $u$ and $v$ such that $u \leadsto_H v$, either there is a path
of length at most $3k'/8$ from $u$ to $v$ in $H$ or there is a path from $u$ to $v$ that contains
a vertex in $W$.  The natural greedy algorithm for this problem constructs $W$ to be of size at most
$O(\frac{n \log n}{k+ \log n})$.  To construct the \spshort{k} $S$, add to $S$ the edges in $H$ and
for each vertex $w \in W$, add edges from $w$ to all vertices in $S_w$.   Also, for each contracted
cycle in $G$, add an undirected star $T_v$ centered at one arbitrary vertex of the cycle.  The size
of $S$ is at most $|H| + O((n^2 \log n)/(k^2 + k \log n)) + O(n)$.  Since $S_k(G) = \Omega(n)$ if
$G$ is connected, $|S|/S_k(G) =  O((n \log n)/(k^2 + k \log n))$.  To see that $S$ is a
\spshort{k}, observe that for any pair of $(u,v)$ with $u \leadsto_G v$, there will
be a vertex $w \in W$ within distance $3k'/8$ of $u$ and a vertex $w' \in S_w$ within distance $k'/4$
of $v$; so, if none of the involved vertices are cycles, the distance between $u$ and $v$ in $S$ is
at most $3k'/8 + k'/4 + 1 = 5k'/8 + 1$.  If the vertices $u, v$ correspond to contracted cycles,
it is easy to see that the path length will be at most $5k'/8+5$.  We choose $k' = ck$ to ensure
$5k'/8+5$ is at most $k$; this is always possible because $k \geq 6$.
\end{proof}

\section{$2^{\log^{1-\epsilon} n}$-Hardness of \tcspanner\ for constant $k > 2$}\label{app:khardness}

%\begin{theorem}[{\bf Noise-Resilient MIN-REP is hard}]
% Fix any $\kappa\in(0, 1)$ and $R, D, M, F \in (0,
%  1-\kappa)$ satisfying  $F\in (R, 2R)$ and $D + M + F < 1$.
  \snote{can we have $D+M+F +\kappa<1$ to get rid of $\kappa'$ in the proof?}
%  Let~{\em Noise-Resilient {\sc MIN-REP}} consist of  $(n,r,d,m)$-{\sc
%    MIN-REP} instances with $r\in[n^R, n^{R+\kappa}]$, $d\in[n^D, n^{D+\kappa}]$, $m\in[n^M,
%  n^{M+\kappa}]$, and $OPT\in[n^F, n^{F+\kappa}]$. This problem is $2^{\log^{1-\eps}n}$-inapproximable for all
%  $\eps \in (0,1)$ unless $NP \subseteq DTIME(n^{\polylog n})$.
%% Let $(n,r,d,\maxisol)$-{\sc MIN-REP} with $r \in (n^{R_1}, n^{R_2})$, $d \in (n^{D_1},n^{D_2})$,
%% $\maxisol \in (n^{L_1},n^{L_2})$ and with $OPT \in (n^{F_1}, n^{F_2})$, for some given parameters $R_i, D_i, L_i, F_i\in (0,1)$, $i\in [2]$ satisfying the following constraints: $R_1 < R_2< F_1 < F_2< R_1 + D_1$, $D_1 < D_2$, $L_1 < L_2$,
%%  $F_2 < 2\cdot R_1$, and $D_2 + L_2 + F_2 < 1$. Then  $(n,r,d,\maxisol)$-{\sc MIN-REP} is $2^{\log^{1-\eps} n}$-inapproximable for any $\eps \in (0,1)$.
%\end{theorem}

\begin{proofof}{\thmref{transformations}}
We give a reduction from {\sc MIN-REP} with unrestricted parameters, considered in~\cite{ElkPel}:
\begin{fact}[\cite{ElkPel}]\label{fact:minrep}
For all $\eps\in(0,1)$, there is no polynomial time algorithm for the {\sc MIN-REP} problem  with
approximation ratio $2^{\log^{1-\eps}n}$ unless $NP \subseteq DTIME(n^{\polylog n})$.
\end{fact}

% Choose arbitrary $R \in (R_1,R_2)$, $D \in (D_1,D_2)$, $L \in (L_1,L_2)$, and $F \in (F_1,
% F_2)$. Let $\kappa \ll \min(R_2-R, D_2-D, L_2 - L,  F_2 - F)$.
We reduce an arbitrary {\sc MIN-REP} instance on
$n^{\kappa'}$ vertices to a {\sc MIN-REP} instance on $n$ vertices with parameters in the desired
range (where $\kappa'$ is a suitably small constant).  Since
{\sc MIN-REP} with unrestricted parameters is $2^{\log^{1-\eps}n}$-inapproximable and
the reduction is polynomial time, the theorem follows.  The reduction consists of a sequence of five
transformations on the original instance. We describe each of the transformations and specify how
the parameters of the input and output {\sc MIN-REP} instances are related.
\begin{itemize}
\item[1.]
\textit{(Disjoint copies)}

Given an $(n_0,r_0,d_0,\maxisol_0)$-{\sc MIN-REP} instance $G_0$ with $OPT_0$ as the solution value,
$T_1(G_0, n^\parone)$ is defined to be the {\sc MIN-REP} instance $G_1$ with $n^\parone$ disjoint copies of
$G_0$.  $G_1$ is a $(n_1,r_1,d_1,\maxisol_1)$-{\sc MIN-REP} instance with $n_1 = n^\parone n_0$, $r_1 =
n^\parone r_0$, $d_1 = d_0$, and $\maxisol_1 = \maxisol_0$. The solution value of $G_1$ is
$OPT_1 = n^\parone OPT_0$ because if $OPT_1 < n^\parone OPT_0$, one could, by averaging over the
$n^\parone$ copies of $G_0$, extract a {\sc MIN-REP} cover for $G_0$ of size smaller than $OPT$.

\item[2.]
\textit{(Dummy vertices inside clusters)}

Given an $(n_1,r_1,d_1,\maxisol_1)$-{\sc MIN-REP} instance $G_1$ with $OPT_1$ as the solution value,
$T_2(G_1, n^\partwo)$ is defined to be the {\sc MIN-REP} instance $G_2$ obtained by increasing the
size of each cluster by a factor of $n^\partwo$ and not attaching any edges to the new vertices.  $G_2$
is a $(n_2,r_2,d_2,\maxisol_2)$-{\sc MIN-REP} instance with $n_2 = n^\partwo n_1$, $r_2 = r_1$, $d_2 =
d_1$, and $\maxisol_2 = n^\partwo \maxisol_1$.  The solution value of $G_2$ remains $OPT_2 = OPT_1$ because the
minimum cover of $G_2$ does not include any isolated vertices.

\item[3.]
\textit{(Blowup inside clusters with matching supergraph)}

Given an $(n_2,r_2,d_2,\maxisol_2)$-{\sc MIN-REP} instance $G_2$ with $OPT_2$ as the solution value,
$T_3(G_2, n^\parthree)$ is defined to be the {\sc MIN-REP} instance $G_3$ obtained as follows.  For
each cluster ${\cal A}_i$ in $G_2$, construct a cluster ${\cal A}_i'$ in $G_3$ consisting of $n^\parthree$ copies of
${\cal A}_i$.  Let $({\cal A}_i')_k$ denote the $k$th copy of ${\cal A}_i$ inside ${\cal A}_i'$.  Whenever there is an edge in $G_2$
between $u \in {\cal A}_i$ and $v \in {\cal B}_j$, for each $1\leq k \leq n^\parthree$, add an edge between the
copy of $u$ in $({\cal A}_i')_k$ and the copy of $v$ in $({\cal B}_j')_k$.  This procedure yields a $(n_3,
r_3,d_3, \maxisol_3)$-{\sc MIN-REP} instance $G_3$ where $n_3 = n^\parthree n_2$, $r_3 = r_2$, $d_3 =
d_2$, and $\maxisol_3 = \maxisol_2$.  The solution value of $G_3$ remains $OPT_3 = OPT_2$ because the
supergraph corresponding to $G_3$ and $G_2$ are identical.

\item[4.]
\textit{(Blowup inside clusters with complete supergraph)}

Given an $(n_3,r_3,d_3,\maxisol_3)$-{\sc MIN-REP} instance $G_3$ with $OPT_3$ as the solution value,
$T_4(G_3, n^\parfour)$ is defined to be the {\sc MIN-REP} instance $G_4$ obtained as follows.  For
each cluster ${\cal A}_i$ in $G_3$, construct a cluster ${\cal A}_i'$ in $G_4$ consisting of $n^\parfour$ copies of
${\cal A}_i$.  Let $({\cal A}_i')_k$ denote the $k$th copy of ${\cal A}_i$ inside ${\cal A}_i'$.  Whenever there is an edge in $G_3$
between $u \in {\cal A}_i$ and $v \in {\cal B}_j$, for each $1\leq k_1,k_2 \leq n^\parfour$, add an edge between the
copy of $u$ in $({\cal A}_i')_{k_1}$ and the copy of $v$ in $({\cal B}_j')_{k_2}$.  This procedure yields a $(n_4,
r_4,d_4, \maxisol_4)$-{\sc MIN-REP} instance $G_4$ where $n_4 = n^\parfour n_3$, $r_4 = r_3$, $d_4 =
n^\parfour d_3$, and $\maxisol_4 = \maxisol_3$.  The solution value of $G_4$ remains $OPT_4 = OPT_3$ because the
supergraph corresponding to $G_3$ and $G_4$ are identical.

\item[5.]
\textit{(Tensoring)}

Given an $(n_4,r_4,d_4,\maxisol_4)$-{\sc MIN-REP} instance $G_4$ with $OPT_4$ as the solution value,
$T_5(G_4, n^\parfive)$ is defined to be the {\sc MIN-REP} instance $G_5$ obtained by repeating the
following construction $\log_2 n^\parfive$ times\footnote{For simplicity, we assume $n^\parfive$ is a power of $2$.}. For each cluster ${\cal A}_i$ in $G_4$, construct two clusters ${\cal A}_i'$ and ${\cal A}_i''$ in $G_5$.
Furthermore, ${\cal A}_i'$ contains two copies of ${\cal A}_i$ and ${\cal A}_i''$ contains two copies of ${\cal A}_i$.  Denote the
two copies inside ${\cal A}_i'$ as $({\cal A}_i')_1$ and $({\cal A}_i')_2$ and similarly the two copies inside ${\cal A}_i''$ as
$({\cal A}_i'')_1$ and $({\cal A}_i'')_2$.  For each edge $(u,v)$ in $G_4$ with $u \in {\cal A}_i$ and $v \in {\cal B}_j$, add
the following four edges in $G_5$: between the copy of $u$ in $({\cal A}_i')_1$ and copy of $v$ in
$({\cal B}_j')_1$, between the copy of $u$ in $({\cal A}_i')_2$ and copy of $v$ in $({\cal B}_j'')_2$, between the copy
of $u$ in $({\cal A}_i'')_1$ and copy of $v$ in $({\cal B}_j'')_1$, and between the copy of $u$ in $({\cal A}_i'')_2$ and
copy of $v$ in $({\cal B}_j')_2$.

The procedure yields a $(n_5,r_5,d_5,\maxisol_5)$-{\sc MIN-REP} instance $G_5$ where $n_5 = n^{2\parfive}
n_4$, $r_5 = n^\parfive r_4$, $d_5 = d_4$, and $\maxisol_5 = \maxisol_4$.  Also, we argue that $OPT_5 =
n^{2\parfive} OPT_4$.  Clearly, $OPT_5 \leq n^{2\parfive} OPT_4$ because one could choose copies of the
vertices in the cover for $G_4$ in each of the $n^\parfive$ copies of the clusters of $G_4$.  For the
other direction, notice that $G_5$ contains $n^{2\parfive}$ vertex disjoint copies of $G_4$, and so, if
$OPT_5 < n^{2\parfive} OPT_4$, then by averaging, there would be a copy of $G_4$ covered using less than
$OPT$ vertices, a contradiction.
\end{itemize}

For some positive $\kappa'$ sufficiently smaller than $\kappa$, consider an arbitrary
$(n^{\kappa'},r_0,d_0,\maxisol_0)$-{\sc MIN-REP} instance $G_0$ with optimum $OPT_0$, where the
only constraints on the parameters are nontriviality conditions: $r_0 \in [1, n^{\kappa'}]$, $d_0\in[1, n^{\kappa'}]$,  $\maxisol_0 \in[1, n^{\kappa'}]$, and $OPT_0 \in [1, 2 n^{\kappa'}]$.  Let $G =
T_5(T_4(T_3(T_2(T_1(G_0, n^\parone), n^\partwo),\linebreak n^\parthree), n^\parfour), n^\parfive)$.  We choose
$\parone, \partwo, \parthree, \parfour, \parfive$ such that $G$ is a $(n,r,d,\maxisol)$-{\sc MIN-REP} instance with
$r \in [n^{R}, n^{R + \kappa'}]$, $d \in [n^{D},n^{D + \kappa'}]$, $\maxisol \in [n^{M}, n^{M +
  \kappa'}]$ and $OPT \in  [n^{F},n^{F + \kappa'}]$.  By
definitions of transformations, $n = n^{\kappa' + \parone +  \partwo + \parthree + \parfour + 2\parfive}$, $r \in
[n^{\parone +  \parfive}, n^{\kappa' + \parone + \parfive}]$, $d \in [n^\parfour, n^{\parfour +
  \kappa'}]$, $\maxisol \in [n^\partwo, n^{\kappa' + \partwo}]$, and $OPT \in [n^{\parone + 2 \parfive}, n^{\kappa' + \parone +
 2\parfive}]$.  Therefore, choose $\parfour = D$, $\partwo = \Maxisol$, $\parfive = F - R$, and $\parone = 2R-F$.  All
of these values are in $(0,1)$ by restriction of the parameters in the theorem statement.  Now,
since $\kappa' + \parone +  \partwo + \parthree + \parfour
+ 2\parfive = D + \Maxisol + F + \parthree + \kappa'$ and since $D+\Maxisol+F < 1$ and $\kappa'$ can be made as
small as we want, we can  choose $\parthree \in (0,1)$ such that $n = n^{\kappa' + \parone +  \partwo + \parthree
  + \parfour + 2\parfive}$.  Therefore, $G$ is a {\sc MIN-REP} instance with parameters in the desired range.
\end{proofof}

{\bf Attaching butterflies in the construction of a hard \tcspanner\ instance $\tcinstance$.} \snote{Move to the main body?}
 We further discuss the way the butterflies are attached to the
 groups.  Recall that $\nrinstance$ is a $(n_0, r, d_0, m)$-{\sc MIN-REP} instance with $n_0 =
 n^\delta, r \in [n^{\delta/2},n^{\delta/2 + \kappa}], d_0 \in [n^\eta, n^{\eta + \kappa}]$ and
 $m \in [n^{2\eta},n^{2\eta + \kappa}]$. Thus, for each group $A_{i, j}$ there are at most $\frac{n^{\genpar}}{rm}$
non-isolated vertices. We will attach the butterfly $BF(A_{i, j})$ in such a way that each vertex in $BF^{k-1}(A_{i, j})$ is adjacent to at most $\frac{d_*}{m}$ non-isolated vertices in $A_{i, j}$, out of a total out-degree of size $d_*$.
This is the crucial property exploited  later in the proof. We can achieve this property in the following way. Recall that \snote{Make notation consistent with the definition of butterflies}
each vertex of $BF^s(A_{i, j})$ is labeled $(a_1, \ldots, a_{k-1}, s)$, where $a_l\in [d_*]$ for all $l\in[k-1]$, $s\in [k]$, and
each vertex $v=(a_1, \ldots, a_{k-1}, k)$ connects to $v'=(a_1, \ldots, a_{k-2},a'_{k-1}, k-1)$. Thus, for a fixed prefix $b=(b_1, b_2, \ldots, b_{k-2})$ all
vertices $(b_1, \ldots, b_{k-2}, b_{k-1}, k-1)$ connect to the same set $A_b$ of vertices in $A_{i,
  j}$, and $|A_b|=d_*$. Choose the set $A_b$ to contain at most $d_*/m$ non-isolated vertices, which
is possible since the total fraction of non-isolated vertices in $A_{i, j}$ is $\leq \frac{1}{m}$.

\begin{lemma}[{\bf Rep-cover Spanner Lemma}]\label{lem:ubspshort}
There is a \spshort k $\cal H$ s.t. $|{\cal H}|= O(OPT ~n^{1-\genpar} (\frac{n^{\genpar}} {r} )^{\frac{2}{k-1}})$, where $OPT$ is the minimum rep-cover of the underlying $\nrinstance$ (and of $\spinstance$ as well) . Moreover, ${\cal H}$ contains only paths of type $(2\& (k-2), 2\& (k+1))$.
\end{lemma}

\begin{proof}
%Construct the following \spshort k ${\cal H}$ which only contains edges of type $2\& (k-2)$ and $2 \& k+1$.
 \snote{Also need to argue that nodes at distance $k+1$ are connected by a length-$k$ path.} \enote{Done}
 We construct the graph $\mathcal{H}$ by adding some shortcut edges to $\mathcal{G}$. Let $S_0$ be a minimum rep-cover of $\nrinstance$ of size OPT. Recall that each ${\cal A}_i$ and ${\cal B}_j$ is replicated $n^{1-\delta}$ times in $\spinstance$. Let $S$ be the set of all replicas in $\spinstance$ of vertices in $S_0$.\snote{In the main body $S$ is used instead of $O$}\anote{Fixed}
 Let $A_{i,j}$ and $B_{k, l}$ be two comparable groups of vertices.
 %Denote the in-degrees as well as out-degree of each vertex of  the butterflies $BF(A_{i,s})$ by $d_*\defeq (\frac{n^{\genpar}} {r})^{\frac{1}{k-1}}$.
 Recall that $d_*=(\frac{n^{\genpar}} {r})^{\frac{1}{k-1}}$.
  To get a \spshort k on $BF(A_{i,j}) \cup  BR(B_{k, l})$ connect each vertex $v$ from the restriction of $S$ to $A_{i,j}$ with all its $d_*^2$ comparable vertices in $BF^{k-2}(A_{i, j})$.\enote{corrected from $BF^{k-1}(A_{i, j})$}
   Similarly, connect each vertex in the restriction of $S$ to $B_{k, l}$ to its %$pt$
  $d_*^2$ comparable vertices in $BR^{k+3}(B_{k, l})$.
 Since every vertex $u\in BF^1(A_{i, j})$ is comparable to every vertex $v\in A_{i, j}$, it follows that there is a vertex $w\in BF^{k-2}(A_{i, j})$ comparable to both $u$ and $v$. Thus, between any such $u$ and $v$ there is a path using an edge of type $2\&(k-2)$. Similarly, every vertex in $BR^{k+3}(B_{k, l})$ is comparable to every vertex of $B_{k, l}$. By our construction, any pair of vertices $(u_1, u_{k+3})\in BF^1(A_{i, j})\times BR^{k+3}(B_{k, l})$ is connected by a path of type $(2\& (k-2), 2\& (k+1))$. In addition, any pair of vertices $(u_1, u_{k+2})\in BF^1(A_{i, j})\times BR^{k+2}(B_{k, l})$, as well as $(u_2, u_{k+3})\in BF^2(A_{i, j})\times BR^{k+3}(B_{k, l})$ are connected by a path of length at most $k$ using shortcut edges of types $2\&(k-2)$ and $2\&(k+1)$, respectively.
  By connecting all the comparable groups $A_{i, j}$ and $B_{k, l}$ in this manner, we obtain a \spshort k on $\cal G$.

 Since there are $n^{1-\genpar}$ copies of each $A_{i, j}$ and $B_{k, l}$ the total number of shortcut edges added  is $OPT n^{1-\genpar} d_*^2 = OPT ~n^{1-\genpar} (\frac{n^{\genpar}} {r} )^{\frac{2}{k-1}}$ and we only used shortcut edges of types $2\&(k-2)$ and $2\& (k+1)$. In addition, since $\mathcal{G}$ is transitively reduced, $\mathcal{H}$ must include all the edges of $\mathcal{G}$. We bound the size of $\cal G$ by inspecting the total number of edges in the butterflies ($knd_*$),  the  {\sc MIN-REP} instance ($\leq nd$), and the brooms ($nd_*+n^{1-\delta}rd_*^2$).
Thus, $|{\cal G}| \leq k~ r~ n^{1-\genpar}~ (\frac{n^{\genpar}}{r}) ^{1+ \frac{1}{k-1}}~+
~n^{2+\eta+\kappa-\genpar}~+~n ~(\frac{n^{\genpar}}{r})^{\frac{1}{k-1}}~+ ~n^{1-\genpar}~ r~
(\frac{n^{\genpar}}{r})^{\frac{2}{k-1}}$.  The following conditions, satisfied by the
parameters of our construction, suffice to show that each term of the preceding sum is
respectively $o(|{\cal H}|)$. (The parameter $\kappa$ is omitted from the conditions, since if the inequalities are satisfied without $\kappa$ then $\kappa$ can be made sufficiently small to ensure that they are satisfied with $\kappa$.)\snote{As a compromise to our margin notes discussion with Arnab, I added the remark in (), but it came out lengthy.}\anote{This is better now.  Thanks!}
%\snote{Added $\kappa$ to equations for MIN-REP size; $\kappa$
%needed elsewhere?}\anote{Deleted $\kappa$ from the second inequation because we can always make it small enough to %not change the strict inequalities.}
\begin{align}
&\zeta + (1 - \delta) + \frac{2}{k-1}\left(\delta - \frac{\delta}{2}\right) >
\frac{\delta}{2} + (1-\delta) + \frac{k}{k-1}\left(\delta - \frac{\delta}{2}\right) &\text{, or }
\zeta > \delta \frac{2k-3}{2(k-1)}\\
&\zeta + (1 - \delta) + \frac{2}{k-1}\left(\delta - \frac{\delta}{2}\right) > 2 + \eta - \delta
&\text{, or } \zeta > 1 + \eta  - \frac{\delta}{k-1} \\
&\zeta + (1 - \delta) + \frac{2}{k-1}\left(\delta - \frac{\delta}{2}\right) > 1 + \frac{1}{k-1}\left(\delta -
  \frac{\delta}{2} \right) &\text{, or } \zeta > \delta \frac{2k-3}{2(k-1)}\\
&\zeta + (1 - \delta) + \frac{2}{k-1}\left(\delta - \frac{\delta}{2}\right) > (1 - \delta) +
\frac{\delta}{2} + \left(\delta - \frac{\delta}{2}\right)\frac{2}{k-1} &\text{, or } \zeta > \frac{\delta}{2}
\end{align}
\end{proof}

\begin{lemma}[{\bf Path Analysis Lemma}]
There are $o(OPT)$ deletable superedges $({\cal A}_i,{\cal B}_j),$ where $i,j \in [r]^2$.
\end{lemma}
\begin{proof}
We call a path {\em canonical} if it contains shortcut edges of types both $2\&(k-2)$ and $2\&(k+1)$; otherwise, a path is {\em alternative}. Observe that any alternative path contains at least one shortcut edge from among the following three cases:
(1)  shortcut edges crossing both $V_k$ and $V_{k+1}$, i.e. one of the shortcut edge types:
$3\&(k-2)$, $3\&(k-1)$, $2\&(k-1)$, $2\&k$, and $3\&k$; (2) shortcut edges of type  $3\&\ell$
where  $\ell\le k-3$; (3) shortcut edges of type $2\&\ell$  where $\ell\le k-3$. Let $S_B$ be the
set of all the shortcut edge types contained in the above three cases. Then $|S_B|=\Theta(k)$. Now, for each  shortcut edge type
$S\in S_B$, let $Del(S) = \{(i,j) \in [r]^2 | \mbox{ at least }\frac{1}{4|S_B|} \mbox{ fraction of pairs }
(u,v) \in BF^1({\cal A}_i) \times BR^{k+3}({\cal B}_j) \mbox{ have an alternative path containing a shortcut edge of
type }S \}.$  By a union bound, the total number of deletable superedges is at most $\sum_{S\in S_B}
Del(S)$. Hence it suffices to show that for all $S\in S_B$, $Del(S)=o(OPT).$

Let $C(S)=\{ (u,v)\in BF^1({\cal A}_i) \times BR^{k+3}({\cal B}_j)
\mid \exists \mbox{ an alternative path between } u \mbox{ and } v$ containing
  a shortcut edge of type  $S\}$. By the definition of $Del(S)$, since for all $i\in [r]$,
$|BF^1({\cal A}_i)|=\frac{n}{r}$ and $|BR^{k+3}(B_{i})|=n^{1-\genpar} d_{*}^2$, we have
\begin{equation}\label{eq:case0}
|C(S)|\ge |Del(S)| \frac{1}{4|S_B|}\frac{n}{r}~ n^{1-\genpar}~ d_{*}^2.
\end{equation}

Now we will obtain upper bounds of $|C(S)|$ in terms of $OPT$ for each of three cases of shortcut
edges, thus obtaining upper bounds on $Del(S)$.  Recall that $\delta =
\frac{k-1}{k-\frac{1}{4}}$, $\eta = \frac{\delta}{2(4k-4)(4k-2)}$, and $\zeta = \delta
\left(\frac{4k-5}{4k-4} + \frac{1}{4k-2}\right)$.  Also, recall $r \in [n^\frac{\delta}{2},
n^{\frac{\delta}{2} + \kappa}]$, $d \in [n^{(1-\delta) + \eta}, n^{(1-\delta) + \eta + \kappa}]$, and $m \in [n^{2\eta},
n^{2\eta + \kappa}]$ for some small enough constant $\kappa$ and $d_* =
\left(\frac{n^\delta}{r}\right)^{1/(k-1)}$. We mostly ignore $\kappa$ below since we can make it as small a
constant as we like.

Suppose that $S$ is a shortcut edge from the first case.
Then $S$ is a shortcut of type $\ell_1 \& (k-\ell_2)$, where $ 2\le\ell_1\le 3, 0\le \ell_2$, and $\ell_1-\ell_2\ge 1$.
Now we obtain that for any shortcut of type $S$, the shortcut can be used for at most
$d_{*}^{k-1-\ell_2} d_{*}^{3+\ell_2-\ell_1} = d_*^{k+2-\ell_1}$ many pairs $(u,v)\in C(S)$. Hence,
$|C(S)|\le d_{*}^{k+2-\ell_1} \cdot OPT \frac{n^{1-\genpar} d_*^2}{\log n}  \le d_{*}^{k} \cdot OPT
\frac{n^{1-\genpar}~d_{*}^2}{\log n}$.
From (\ref{eq:case0}), we obtain that

$$|Del(S)|\le 4|S_B|\frac{ d_{*}^{k} OPT~ n^{1-\genpar}~ d_{*}^2}{\frac{n}{r}~  n^{1-\genpar}~ d_{*}^{2}\log n}=O\left(\frac{d_{*}}{n^{1-\genpar}\log n}\right) OPT.$$
Then because $\genpar < \frac{k-1}{k-\frac{1}{2}}$, $1 - \delta >
\frac{\delta}{2(k-1)}$, and so we obtain that $n^{1-\genpar}$ is a polynomial factor larger than
$d_{*} = \left(\frac{n^{\genpar}}{r}\right)^{1/(k-1)}$, which  proves that $|Del(S)|=o(OPT)$.

Now suppose that $S$ is a shortcut type of the second case.
Let $S$ be type $3\&\ell$, where $1\le\ell\le k-3$.
Now, from the fact that out-degree of each vertex in $V_k$ is at most $n^{1-\genpar + \eta + \kappa}$, we
obtain that for any shortcut of type $S$, the shortcut can be used for at most $d_{*}^{\ell-1}
d_{*}^{k-3-\ell} n^{(1-\genpar) + \eta} d_{*}^2 = d_{*}^{k-4} n^{(1-\genpar)+\eta}
d_{*}^2$ many pairs  $(u,v)\in C(S)$ (ignoring $\kappa$ as mentioned above). Hence, upto small
polynomial factors,
\begin{equation}|C(S)| \le d_{*}^{k-4}
  n^{(1-\genpar)+\eta} d_{*}^2 \cdot OPT \frac{n^{1-\genpar}d_{*}^2}{\log n} =\frac{d_{*}^k
    n^{2-2\genpar+\eta}}{\log n} OPT\label{eq:case02}\end{equation}
From (\ref{eq:case0}) and (\ref{eq:case02}), we obtain $|Del(S)|\le 4|S_B|\frac{n^{\eta}}{d_{*}\log
  n} OPT$.  Now, $\eta < \frac{\delta}{2(k-1)}$, and so, $|Del(S)|=o(OPT).$

Now suppose that $S$ is a shortcut type of the third case.
Let $S$ be type  $2\&\ell$, where $\ell\le k-3$. Note that for any vertex $v$ in $V_{k-1}$ the
number of non-isolated vertices \snote{Strange terminology} in $V_k$ that $v$ is connected to is $\frac{d_{*}}{m}$. Hence,
together with the fact that out-degree of each vertex in $V_k$ is at most $n^{1-\genpar + \eta + \kappa}$, we obtain that for any shortcut of type $S$, the shortcut can be used for at most
$\frac{d_{*}^{k-3}}{n^{2\eta}} n^{(1-\genpar)+\eta} d_{*}^2$ many pairs $(u,v)$ in $C(S)$ (upto
small polynomial factors). Then
\begin{equation}|C(S)|\le  \frac{d_{*}^{k-3}}{n^{2\eta}} n^{(1-\genpar)+\eta} d_{*}^2 \cdot OPT \frac{n^{1-\genpar}d_{*}^2}{\log n}=\frac{d_{*}^{k+1}n^{2-2\genpar+\eta}}{n^{2\eta}\log n}OPT.\label{eq:case2}\end{equation}
From (\ref{eq:case0}), (\ref{eq:case2}), and the fact that $n^{\eta}=o(n^{2\eta}\log n)$, we get $|Del(S)|\le 4|S_B|\frac{n^{\eta}}{n^{2\eta}\log n} OPT=o(OPT).$
\end{proof}

\begin{lemma}[{\bf Rerandomization Lemma}]
If a $\frac{3}{4}$-good \spshort{k}  $\mathcal{K}'$  for $\mathcal{G}'$ is given, then there exists
$\mathcal{K}''$, a $1$-good \spshort{k} for $\mathcal{G}'$, such that $|\mathcal{K}''| \leq
O(|\mathcal{K}'|\cdot \log n)$.
\end{lemma}

\begin{proof}
First, we fix some notation.
Consider some $(i,j) \in [r]^2$ such that there is an edge between a vertex in $\mathcal{A}_i$ and
a vertex in $\mathcal{B}_j$ in $\mathcal{G}'$.  Let $S_{i,j}$ be the set of vertices in $\mathcal{A}_i$ that are
adjacent to $\mathcal{B}_j$, and let $T_{i,j}$ be the set of vertices in $\mathcal{B}_j$ that are
adjacent to $\mathcal{A}_i$.  We know that at least $\frac{3}{4}$ of the vertices in
$BF^1(\mathcal{A}_i)$ have a path of type $(2\&(k-2))$ to $S_{i,j}$ and at least
$\frac{3}{4}$ of the vertices in $BR^{k+3}({\cal B}_j)$ have a path of length $1$ from $T_{i,j}$.  By a Markov
argument, for at least $\frac{1}{2}$ of the groups $A_{i,s}$ in $\mathcal{A}_i$, at least
$\frac{1}{2}$ of the vertices in $BF^1(A_{i,s})$ must have a path of type $(2\&(k-2))$ to $S_{i,j}$.
Call the butterfly attached to such a group $A_{i,s}$ an $(i,j)$\textit{-good butterfly}, and call the set of
vertices in $BF^{k-2}(A_{i,s})$ that have shortcut edges to $S_{i,j}$ $(i,j)$\textit{-helpful vertices}.
Similarly, for at least $\frac{1}{2}$ of the groups $B_{j,t}$, at least $\frac{1}{2}$
of the vertices in $BR^{k+2}(B_{j,t})$ have shortcut edges to $T_{i,j}$.  We call the brooms
attached to such groups $B_{j,t}$ $(i,j)$\textit{-good brooms} and we again call the vertices in
$BR^{k+2}(B_{j,t})$ that have shortcut edges to $T_{i,j}$ $(i,j)$-\textit{helpful vertices}.  It will be
clear from context whether a helpful vertex is to the left or right of the {\sc MIN-REP} instance.

Our construction of $\mathcal{K}''$ ensures that in $\mathcal{K}''$, for any two comparable
clusters $({\cal A}_i,{\cal B}_j)$, each vertex in $BF^1({\cal A}_i)$ is comparable to a helpful vertex in
$BF^{k-2}({\cal A}_i)$ and each vertex in $BR^{k+3}({\cal B}_j)$ is a helpful vertex.  This is enough to ensure
that $\mathcal{K}''$ is $1$-good \spshort{k} for $\mathcal{G}'$.  We will construct
$\mathcal{K}''$ to be equal to $\bigcup_{r=1}^{O(\log n)} \Pi_r(\mathcal{K}')$ where each
$\Pi_r$ is a random transformation of $\mathcal{K}'$ that moves the shortcut edges.

Each $\Pi_r$ will be the composition of several transformations on the edges of $\mathcal{K}'$. The
transformations move only shortcut edges, but not transitive reduction edges, in $\mathcal{K}'$.  Informally, the first transformation randomly permutes the groups in each cluster
on the left side of the {\sc MIN-REP} instance, the second randomly permutes the groups in each
cluster of the right side of the {\sc MIN-REP} instance, the third randomly permutes the edges of
the butterfly graph, and the fourth randomly permutes the broomsticks.  Formally:

\begin{itemize}
\item
\textit{Left Group permutations: $\Pi^{lg}$}

For each $i \in [r]$, independently choose a random permutation $\pi_i : [n^{1-\genpar}] \to
[n^{1-\genpar}]$.  For each cluster $\mathcal{A}_i$, if $(u,v)$ is an edge in $\mathcal{K}'$ with $u,
v \in BF(A_{i,s})$, then there is an edge $(u',v')$ in $\Pi^{lg}(\mathcal{K}')$, where $u'$ and $v'$
are the copies of $u$ and $v$ respectively in $BF(A_{i,\pi_i(s)})$.

\item
\textit{Right Group permutations: $\Pi^{rg}$}

For each $j \in [r]$, independently choose a random permutation $\pi_j : [n^{1-\genpar}] \to
[n^{1-\genpar}]$.  For each cluster $\mathcal{B}_j$, if $(u,v)$ is an edge in $\mathcal{K}'$ with $u,
v \in BR(B_{j,s'})$, then there is an edge $(u',v')$ in $\Pi^{rg}(\mathcal{K}')$, where $u'$ and $v'$
are the copies of $u$ and $v$ respectively in  $BR(B_{j,\pi_j(s')})$.

\item
\textit{Butterfly permutations: $\Pi^{bf}$}

For each $i \in [r]$ and $s \in [n^{1-\genpar}]$, label a vertex $u$ in $BF(A_{i,s})$ as
$(a_1,a_2,\dots, a_{k-1}, m) \in [d_*]^{k-1} \times [k]$, where $u \in V_m$ and
$(a_1,a_2,\dots,a_{k-1})$ is the usual vertex labelling that defines a generalized butterfly graph.
Now, for every $(i,s)$ and every $(a_1,\dots,a_{k-3}) \in [d_*]^{k-3}$, independently choose two random permutations
$\pi_{i,s}^{(a_1,\dots,a_{k-3})} : [d_*] \to [d_*]$  and $\sigma_{i,s}^{(a_1,\dots,a_{k-3})} : [d_*]
\to [d_*]$.  For any edge $(u,w) \in BF^{k-2}(A_{i,s}) \times BF^k(A_{i,s})$ where $u =
(a_1,\dots,a_{k-3}, a_{k-2},a_{k-1}, k-2)$ and $w = (a_1, \dots, a_{k-3}, a'_{k-2}, a'_{k-1}, k)$,
there exists the edge $(u',w)$ in $\Pi^{bf}(\mathcal{K}')$ where $u' =
\left(a_1,\dots,a_{k-3},\right.$ $\pi_{i,s}^{(a_1,\dots,a_{k-3})}(a_{k-2}),$
$\sigma_{i,s}^{(a_1,\dots,a_{k-3})}(a_{k-1}),$ $\left.k-2 \right)$.   All other edges in the
butterfly stay fixed.

% For any edge $(u,v) \in BF^{k-2}(A_{i,s}) \times BF^{k-1}(A_{i,s})$ where $u =
% (a_1,\dots,a_{k-3}, a_{k-2}, a_{k-1}, k-2)$ and $v = (a_1,\dots,a_{k-3},a'_{k-2},a_{k-1},k-1)$, there
% exists the edge $(u',v')$ in $\Pi^{bf}(\mathcal{K}')$
% where $u' = \left(a_1,\dots,a_{k-3},\right.$ $\pi_{i,s}^{(a_1,\dots,a_{k-3})}(a_{k-2}),$
% $\sigma_{i,s}^{(a_1,\dots,a_{k-3})}(a_{k-1}),$ $\left.k-2 \right)$ and $v' =
% \left(a_1,\dots,a_{k-3},a'_{k-2}, \sigma_{i,s}^{(a_1,\dots,a_{k-3})}(a_{k-1}),k-1 \right)$.  Also,
% for any edge $(v,w) \in BF^{k-1}(A_{i,s}) \times BF^{k}(A_{i,s})$ where $v =
% (a_1,\dots,a_{k-3}, a_{k-2}, a_{k-1}, k-1)$ and $w = (a_1,\dots,a_{k-3},a_{k-2},a'_{k-1},k)$, there
% exists the edge $(v',w)$ in $\Pi^{bf}(\mathcal{K}')$
% where $v' = \left(a_1,\dots,a_{k-3}, a_{k-2},\right.$
% $\left.\sigma_{i,s}^{(a_1,\dots,a_{k-3})}(a_{k-1}),k-1 \right)$.

\item
\textit{Broom permutations: $\Pi^{br}$}

For each $j \in [r]$ and $s' \in [n^{1-\genpar}]$, independently choose a random permutation
$\pi_{j,s'} : [p] \to [p]$ and $\sigma_{j,s'} : [t] \to [t]$.  Label a vertex $v \in
BR^{k+2}(B_{j,s'})$ as an element of $[p]$ and label a vertex $w \in BR^{k+3}(B_{j,s'})$ as an element
of $[p] \times [t]$ in the natural way.  If $(u,w) \in BR^{k+1}(B_{j,s'}) \times BR^{k+3}(B_{j,s'})$
is an edge in $\mathcal{K}'$,  then $(u,w') \in BR^{k+1}(B_{j,s'}) \times BR^{k+3}(B_{j,s'})$ is an edge in
$\Pi^{br}(\mathcal{K}')$, where $w' = (\pi_{j,s'}(w_1), \sigma_{j,s'}(w_2))$ if the label of $w$ is
$(w_1,w_2)$.  All other edges in the broom stay fixed.
\end{itemize}

Now, for each $r = 1,\dots,O(\log n)$, define $\Pi_r$ to be the composition of $\Pi^{lg}$,
$\Pi^{rg}$, $\Pi^{bf}$, and $\Pi^{br}$.  For each $r$, choose all the permutations independently.
As we said before, we set $\mathcal{K}'' = \cup_r \Pi_r(\mathcal{K}')$.

\begin{claim}
For each $(i,j) \in [r]^2$ such that $\mathcal{A}_i$ and $\mathcal{B}_j$ are comparable,
for any $u \in BF^1({\cal A}_i)$ and $v \in BR^{k+3}({\cal B}_j)$,
\begin{align*}
&\Pr_{\Pi_r} [u \text{ is in a }(i,j)\text{-good
  butterfly in }\Pi_r(\mathcal{K'})] \geq \frac{1}{2},
 &\Pr_{\Pi_r}[v \text{ is in a }(i,j)\text{-good broom in } \Pi_r(\mathcal{K}')] \geq \frac{1}{2}
\end{align*}
\end{claim}
\begin{proof}
At least half the butterflies attached to $\mathcal{A}_i$ are good from above, and hence, for every
vertex $u \in BF^1({\cal A}_i)$, the left
group permutations ensure that with probability at least $\frac{1}{2}$, the edges of a good
butterfly are mapped to the butterfly that $u$ belongs to.  The right group permutations provide the
same function for a $v \in BR^{k+3}({\cal B}_j)$.
\end{proof}

\begin{claim}
For each $(i,j) \in [r]^2$ such that $\mathcal{A}_i$ and $\mathcal{B}_j$ are comparable, then for
any $v \in BR^{k+3}({\cal B}_j)$:
\begin{align*}
\Pr_{\Pi_r} [v \text{ is a }(i,j)\text{-helpful vertex } | v \text{ is in a }(i,j)\text{-good broom}] \geq \frac{1}{2}
\end{align*}
\end{claim}
\begin{proof}
At least half of the broomsticks of a good broom are helpful (i.e., incident to a shortcut edge), and
hence for every vertex $v \in BR^{k+3}({\cal B}_j)$, the broom permutations ensure that with
probability at least $\frac{1}{2}$, $v$ is incident to a shortcut edge.
\end{proof}

\begin{claim}
For each $(i,j) \in [r]^2$ such that $\mathcal{A}_i$ and $\mathcal{B}_j$ are comparable, then for
any $u \in BR^1({\cal A}_i)$:
\begin{align*}
\Pr_{\Pi_r} [u \text{ is comparable to a }(i,j)\text{-helpful vertex } | u \text{ is in a
}(i,j)\text{-good butterfly}] \geq \frac{1}{2}
\end{align*}
\end{claim}
\begin{proof}
Suppose $BF(A_{i,s})$ is a $(i,j)$-good butterfly.  For any $(a_{k-2},a_{k-1}) \in [d_*]^2$, let
$S_{(a_{k-2},a_{k-1})}$ be the set of vertices $u$ in $BF^1(A_{i,s})$ such that $u$ is labelled as
$(a_1, \dots, a_{k-3}, a_{k-2}, a_{k-1}, 1)$ where $(a_1,\dots,a_{k-3})$ are arbitrary elements of
$[d_*]^{k-3}$.  Note that all the vertices in a given $S_{(a_{k-2},a_{k-1})}$ are comparable to the
same set of vertices in $BF^{k-2}(A_{i,s})$ and hence, either they are all comparable to an
$(i,j)$-helpful vertex or none of them are.  Hence, at least $\frac{1}{2}$ of the $S_{(a_{k-2},a_{k-1})}$'s must
have every vertex comparable to a helpful vertex in $BF^{k-2}(A_{i,s})$.  Now, because of the random
butterfly permutations, a given vertex $v \in BF^1(A_{i,s})$ falls in such a $S_{(a_{k-2},a_{k-1})}$
  with probability at least $\frac{1}{2}$.
\end{proof}

Thus, for two vertices $u$ and $v$ in comparable $\mathcal{A}_i$ and
$\mathcal{B}_j$ respectively, the probability that $u$ and $v$ are connected by a canonical
path  in $\Pi_r(\mathcal{K}')$ is at least $\frac{1}{16}$.  Since we take $O(\log n)$ independent random
transformations $\Pi_r$, the probability that $u$ and $v$ will be connected by a canonical
path in at least one $\Pi_r(\mathcal{K}')$ is at least $1-\frac{1}{\poly(n)}$.  Taking a union bound
over all vertex pairs in $\mathcal{A}_i$ and $\mathcal{B}_j$ as well as all possible $i$ and $j$, we
find that with probability at least $\frac{1}{2}$, $\mathcal{K}''$ has a canonical path between any comparable $u\in V_1$ and $v \in V_{k+3}$.  Therefore, the desired
$\mathcal{K}''$ exists and is of size at most $O(|\mathcal{K}'| \cdot \log n)$.
\end{proof}

\begin{lemma}[{\bf Rep-cover Extraction Lemma}]
Given $\mathcal{K}''$, a $1$-good \spshort{k} for $\mathcal{G}'$, of size $o(OPT \cdot n^{1-\genpar}
\cdot d_*^2)$, there exists a {\sc MIN-REP} cover of $\spinstance$ of size $o(OPT)$.
\end{lemma}
\begin{proof}
%sofya: I wasn't sure how to fix the old proof; seemed easier to change it
For $s \in [n^{1-\genpar}]$, define $\mathcal{K}''_s$ to be the subgraph of $\mathcal{K}''$ induced by
$\cup_{i=1}^r \big(BF(A_{i,s}) \cup BR(B_{i,s})\big)$.  The $\mathcal{K}''_s$ are clearly disjoint.  By
averaging, there exists an $\bar s$ such that $|\mathcal{K}''_{\bar s}| \leq o(OPT \cdot d_*^2)$.
%Also, let $\mathcal{G}'_s$ be the restriction of $\mathcal{G}'$ to $\cup_{i=1}^r BF(A_{i,s}) \cup BR(B_{i,s})$.

We further partition the shortcut edges in $\mathcal{K}''_{\bar s}$ into $d_*^2$ parts. For each $x,y\in[d_*]$, let $U_{x,y}$ denote the set of all the nodes in $\cup_{i=1}^r BF^1(A_{i,\bar s})$ with butterfly coordinates $(u_1,\dots, u_{k-2},x,y,1)$, where $u_1,\dots,u_{k-2}\in[d_*]$. To partition the corresponding broomsticks, identify the nodes in $BR^{k+2}(B_{i,s})$ with $[d_*]$, and for each such node $x\in [d_*]$, identify its descendants in $BR^{k+3}(B_{i,s})$ with $(x,1),\dots, (x,d_*)$.
For each $x,y\in[d_*]$, let $U'_{x,y}$ denote the set of all the broomsticks $\cup_{i=1}^r BR^{k+3}(B_{i,\bar s})$ with coordinates $(x,y)$. Define $\mathcal{K}''_{\bar s,x,y}$ to be the subgraph of $\mathcal{K}''_{\bar s}$ induced  by the nodes comparable to the nodes in $U_{x,y}\cup U'_{x,y}$.

Observe that the shortcut edges in different $\mathcal{K}''_{\bar s,x,y}$ are disjoint because (a) different $U'_{x,y}$ are disjoint and (b) the descendants in $V_{k-2}$ of different $U_{x,y}$ are also disjoint. Thus, by averaging, there exist  $\bar x, \bar y$ such that $\mathcal{K}''_{\bar s,\bar x, \bar y} $ contains $o(OPT)$ shortcut edges.

Let $S$ be the set of vertices in $V_k$ and $V_{k+1}$ that are incident to shortcut edges in $\mathcal{K}''_{\bar s,\bar x, \bar y}$.  Then $|S| \leq o(OPT)$.
Observe that $S$ is a rep-cover for the {\sc MIN-REP} instance $\spinstance'_{\bar s}$ obtained
by restricting $\spinstance'$ to the edges between $A_{i,\bar s}$ and $B_{j,\bar s}$. This holds because in $\mathcal{K}''_{\bar s}$, each comparable pair of nodes  in $U_{\bar x,\bar y}\times U'_{\bar x,\bar y}$ is connected by a canonical path. But a {\sc MIN-REP}
cover for $\spinstance'_{\bar s}$ is also a a {\sc MIN-REP}
cover for $\spinstance'$ by definition of $\spinstance$.  Finally, given a rep-cover $S$
of $\spinstance'$, we can get a rep-cover of $\spinstance$ by adding at most $2$ vertices per
super-edge deleted from $\spinstance$ to obtain $\spinstance'$.  Since $o(OPT)$
super-edges were deleted and since $|S| \leq o(OPT)$, we obtain a {\sc MIN-REP} cover for
$\spinstance$ of size $o(OPT)$.
\end{proof}

\section{ $\Omega(\log n)$-Hardness of {\sf 2-TC-Spanner}}\label{app:2hardness}
\begin{theorem}\label{thm:khard}
For any $k \geq 2$, it is \textsc{NP}-hard to approximate the size of the sparsest \spshort k within a ratio of $O( \frac1k ~\log n )$. In particular, {\sc $2$-TC-Spanner} is $\Omega(\log n)$-inapproximable.
\end{theorem}
 Our proof uses a reduction from a variant of \textsf{Set Cover}, called $(a, b, c)$-\textsf{Nice Set Cover}. Before defining this problem we define other variants of \textsf{Set Cover}. An instance of $(a, b)$-\textsf{Set Cover}, consists of a bipartite graph $G=A\cup B$, with $|A|= a$ and $ |B|= b$. An instance of $a$-\textsf{Balanced Set Cover} consists of a bipartite graph $G=A\cup B$, with $|A|= |B|= a$. An instance of $(a, c)$-\textsf{Balanced Bounded Set Cover} consists of a bipartite graph $G=A\cup B$, with $|A|= |B|= a$ and such that the degrees of the vertices in $A$ are at most $c$. Finally, an instance of $(a, b, c)$-\textsf{Nice Set Cover} consists of a bipartite graph $G=A\cup B$, with $|A|=a, |B|=b$.
$B$ can be partitioned into disjoint sets $B_i$ such that $B=\cup_{i=1}^{a} B_i$, $|B_i|= \frac{b}{a}$, assuming $\frac{b}{a}$ is an integer. $G$ must satisfy the property that if $v\in A$ is adjacent to $w \in B_i$, for some $1\leq i\leq a$,  then $v$ is adjacent to every element of $B_i$. Moreover, $v$ is adjacent to at most $c$ sets $B_i$. A solution to all these \textsf{Set Cover} variants is a minimum number of vertices in $A$ that cover all the vertices in $B$.

\begin{lemma}\label{lem:hardnsc}
It is \textsc{NP}-hard to approximate a solution to $(n^a, n^b, n^c)$-\textsf{Nice Set Cover} to
within a ratio of $\gamma~ a~c \log n $ for some constant $\gamma$, where $0< c\leq a\leq b$.
\end{lemma}
\begin{proof}%[Proof of \lemref{hardnsc}]
 We will need the following fact, proved in \cite{rs97}. Earlier, this result
was shown under the weaker assumption that $NP \not \subseteq DTIME(n^{O(\log \log n)})$ \cite{f96,
ly93}.
\begin{fact}\label{claim:0} There is a $d > 0$ for which it is \textsc{NP}-hard to approximate a
solution to $(n^d, n)$-\textsf{Set Cover} to within a ratio of $ \gamma \log n$, for some $\gamma
>0$.
\end{fact}

\begin{claim}\label{claim:1} It is \textsc{NP}-hard to approximate a solution to
$n$-\textsf{Balanced Set Cover} to within a ratio of $\gamma \log n$, for the same $\gamma$ as above.
\end{claim}
\begin{proof} By Fact \ref{claim:0}, $(n^d, n)$-\textsf{Set Cover} is not approximable within a
factor of $\gamma \log n$, unless \textsc{P} $=$ \textsc{NP}. Using a reduction from $(n^d ,
n)$-\textsf{Set Cover}, if $|A|=n^d< n$, transform this instance into an instance where $|A|=|B|$ by
padding $A$ with dummy vertices. If $|A|>n $, transform this instance into an instance where
$|A|=|B|$ by padding the set $B$ with dummy vertices and connecting them to all vertices in $A$.
\end{proof} Applying Lemma 2.3 of \cite{k01} to an instance of $n^c$-\textsf{Balanced Set Cover},
and using Claim \ref{claim:1}, we obtain the following.

\begin{claim}\label{claim:2} It is \textsc{NP}-hard to approximate a solution to $(n^a,
n^c)$-\textsf{Balanced Bounded Set Cover} to within a ratio of $\gamma~ a~c~\log n$, where $\gamma $
is from Claim \ref{claim:0} above.
\end{claim}

To complete the proof of the lemma, notice that a set $M$ is a solution to an instance of $(n^a,
n^b, n^c)$-\textsf{Nice Set Cover} iff $M$ is a solution to the instance of $(n^a,
n^c)$-\textsf{Balanced Bounded Set Cover}, resulted from compressing each set $B_i$ into a single
vertex $b_i$, $1\leq i\leq n^a$. By Claim \ref{claim:2} above, it follows that $(n^a, n^b,
n^c)$-\textsf{Nice Set Cover} is not approximable within a ratio of $\gamma~ a~c \log n$, unless
\textsc{P}$=$\textsc{NP}.
\end{proof}

We now prove the main theorem of this section.
\begin{proof}[Proof of Theorem \ref{thm:khard}]

Let $\alpha=1+\frac3{2k}$ and $\beta=\frac1{5k}$.
Given $G_1=V_{k+1}\cup V_{k+2}$, an instance of $(n, n^{\alpha}, n^{\beta})$-\textsf{Nice Set
Cover}, transform it into the following $k+2$-partite graph $G=V_1\cup V_2\cup \ldots V_{k+1} \cup
V_{k+2}$, with edges directed from $V_i$ to $V_{i+1}$. Let $|V_i|= n$, $1\leq i\leq k+1$, and
$V_{k+2}=n^{\alpha}$. The induced subgraph on $V_1\cup V_2\cup \ldots V_{k+1}$ is $BF(k,n)$, the butterfly
graph of diameter $k$ and width $n$.  Then $|G|\leq k n^{1+\frac1k}+n^{\alpha+\beta}=\Theta~( n^{1+
{\frac{17}{10k}}})$ edges. Notice that there are indeed at most $n^{\alpha+\beta}$ edges from
$V_{k+1}$ to $V_{k+2}$ since there are $n$ vertices in $V_{k+1}$, each of degree at most $n^{\beta +
\alpha - 1}$.

%We will next show that the size of the sparsest \spshort k of $G$, denoted by $OPT_S$, is polynomially related to $OPT_{NSC}$, the optimal
%solution to the $(n, n^{\alpha}, n^{\beta})$-\textsf{Nice Set Cover} instance $G_1$.

\begin{lemma}\label{lem:opt}
$OPT_S=\Theta( OPT_{NSC} ~n^{\frac2k})$.
\end{lemma}

\begin{proof}
First, we show that there is a \spshort k $H$ of $G$ s.t. $|H|=\Theta(
OPT_{NSC} ~n^{\frac2k})$ edges. Then we show that any \spshort k of $G$ must have $\Omega(OPT_{NSC}
~n^{\frac2k})$ edges.

Notice that the only pairs of vertices of $G$ that are not already at distance at most $k$ are the
comparable vertices $u, v$, with $u\in V_1$ and $v \in V_{k+2}$.  In order to connect such pairs by
a directed path of length at most $k$, we need ``shortcut'' edges between different levels $V_i$ and
$V_j$, $i+2\leq j$.  W.l.o.g., we may assume that the only shortcut edges used are those connecting
vertices in $V_i$ to $V_{i+2}$, for some $i$'s. Indeed, a shortcut edge connecting a vertex $u\in
V_i$ to a vertex $v\in V_j$, where $j>i+2$ can be replaced with one edge connecting $u\in V_i$ to a
vertex $w\in V_{i+2}$ that is an ancestor of $v$. In this way, all paths from $V_1$ to $V_{k+2}$
that previously had a path of length at most $k$ still have a path of length at most $k$. Define an
edge $e=(u, v)$ to be a {\it type $i$ edge} if $u\in V_i$ and $v\in V_{i+2}$. We next build a \spshort k of $G$ with $\Theta( OPT_{NSC} ~n^{\frac2k})$ edges.  Let $H$ be the
smallest \spshort k of $G$ which only uses shortcut edges of type $k-1$.

\begin{claim} $|H|=\Theta (n^{\frac2k}~OPT_{NSC})$

\end{claim}
\begin{proof}

Let $O$ be a set of vertices in $V_{k+1}$ that is an optimal solution to the $(n, n^{\alpha}, n^{\beta})$-\textsf{Nice Set Cover} instance.
Connect each vertex $v\in O$ to the set $A_v$ of all the $n^{2/k}$ ancestors of $v$ from level $V_{k-1}$. Direct these edges from $A_v$ to $v$.  Notice that we added $OPT_{NSC}n^{\frac2k}$ edges and the new graph $H'$ is a $k$-TC-spanner. Indeed, each vertex $u\in V_1$ is comparable to each vertex $v\in O$, and thus, there is a vertex $w\in A_v$ that is comparable to $u$. This implies that for every $u\in V_1$ there is a path of length $k-1$ to each of the vertices of $O$, resulting in  a path of length $k$ to each vertex in $V_{k+2}$.
To show that $H$ (the minimum size $k$-TC-spanner with  shortcuts only of type $k-1$) needs at least $OPT_{NSC}n^{\frac2k}$ edges on top of those in $G$, assume otherwise. For $v\in V_{k-1}$, let $n(v)$ be the number of type $k-1$ edges leaving from $v$. By assumption, $\sum_{v\in V_{k-1}} n(v) < OPT_{NSC}~ n^{\frac2k}$. Each vertex in $v\in V_{k-1}$ has exactly $a(v)=n^{1-\frac2k}$  ancestors in $V_1$.
For $u\in V_1$, let $e(u)$ be the total number of type $k-1$ shortcuts leaving from its descendants in $V_{k-1}$. Since there exists a path of length $k$ from $u$ to each vertex in $V_{k+2}$, it follows that $e(u)\geq OPT_{NCS}$.
Notice that $ \sum_{v\in V_{k-1}} n(v) a(v)=   \sum_{u\in V_1} e(u)\geq OPT_{NCS}~ n$. This implies that $\sum_{v\in V_{k-1}} n(v) \geq OPT_{NCS}~ n^{ \frac2k}$, a contradiction to our assumption, concluding that $|H| = |G| + OPT_{NSC}~n^{\frac2k}$.
Next we show that $|H|=\Theta (n^{\frac2k}~OPT_{NSC})$. Indeed, $OPT_{NSC}$ is, by construction,
the same as the size of the optimal solution to an $(n, n^{\beta})$-\textsf{Balanced Bounded Set
Cover} instance, where we must cover $n$ vertices on the right with $n$ vertices of degree at most
$n^{\beta}$ on the left. This implies that $OPT_{NSC}\geq n^{1-\beta}=n^{1-\frac{1}{5k}}$.  Now,
$|G|=\Theta ~( n^{1+ {\frac{17}{10k}}})$ and $|H| = |G| + OPT_{NSC}~n^{\frac2k}$. Since
$OPT_{NSC}n^{\frac2k} \geq n^{\frac2k + 1-\frac{1}{5k}} = n^{1+\frac{18}{10k}}$, this implies that
$|H|=\Theta (OPT_{NSC}~ n^{\frac2k})$.
\end{proof}

%Since a solution to an instance of $(n, n^{\alpha}, n^{\beta})$-\textsf{Nice Set Cover} is also a
%solution to the

Let $M$ be a sparsest spanner on $G$ which possibly uses shortcut edges of types other than
$k-1$.  Assume for the sake of contradiction that $|M|< ~\frac14 ~n^{\frac2k}~OPT_{NSC}$.
A vertex $u\in V_1$ can reach $v\in V_{k+2}$ in at most $k$ steps by using shortcut edges either of
type $i\leq k-1$ or of type $k$. We will show that, under our assumption, there are many vertices in
$V_1$ that can reach at most $\frac 1 2~OPT_{NSC}$ vertices in $V_{k+1}$ by using only edges of some
types $i<k$. Moreover, there are many vertices in $V_1$ that reach only $n^{\frac{1}{k}} OPT_{NSC}$
edges of type $k$. That will be enough to argue that a contradiction must occur, allowing us to
conclude that $|M|=\Theta( n^{\frac2k}~{OPT_{NSC}})$.

 \begin{claim}\label{claim:size1}
 Let $R$ be the set of vertices in $V_1$ that can reach less than $\frac12 ~OPT_{NSC}$ vertices $v\in V_{k+1}$
 in at most $k-1$ steps in $M$. Then $|R|> \frac n2$.
 \end{claim}

%\ifnum\final=1
 \begin{proof}
 For each vertex $u \in V_1$ and $v \in V_{k+1}$, define an
indicator variable $X_{u,v}$ which is $1$ iff there is a shortcut edge along the
unique path from $u$ to $v$ in $G$. Consider a type $i$ shortcut edge $e=(v_i, v_{i+2})$,
with $v_i\in V_i$ and $v_{i+2}\in V_{i+2}$.
Then there are $n^{\frac{i-1}{k}}$ vertices $u$ in $V_1$ with $u \leq v_i$ . Moreover, there are
$n^{\frac{k-i-1}k}$ vertices $v \in V_{k+1}$ with $v_{i+2}\leq  v$. Thus, this shortcut edge $e$
can set at most $ n^{\frac{i-1}{k} + \frac{k-i-1}{k}} = n^{1-\frac2k}$ different $X_{u,v}$ to $1$. By
assumption, there are less than $\frac14~n^{\frac2k}OPT_{NSC}$ shortcut edges of types $i$, where
$i \leq k-1$. It follows that less than $\frac{1}{4}~n~OPT_{NSC}$ different $X_{u,v}$'s can be set to $1$.
For $u\in V_1$, let $n(u)$ be the number of vertices $v\in V_{k+1}$ that $u$ can reach in less than $k$
 steps. Thus, $\mathbb{E}_{u\in V_1}~[n(u)] < \frac{1}{4}~OPT_{NSC}$. By Markov's inequality,
 $Pr_{u\in V_1} [ n(u) \geq \frac{1}{2}~OPT_{NSC}] < \frac12$. This implies that more than
$\frac12$ of the vertices $u \in V_1$ can reach less than $\frac{n}{2}~OPT_{NSC}$ vertices $v \in V_{k+1}$
in less than $k$ steps. Therefore, $ |R| >\frac n 2$.
 \end{proof}
% \fi

We say that a vertex $u$ {\it reaches} an edge $e=(v, w)$, if there is a path from $u$ to $v$. For $u\in V_1$ let $t(u)$ be the number of type $k$ edges that $u$ reaches in $M$.
\begin{claim}\label{claim:size2}
Let $S$ be the set of vertices $u\in V_1$ s.t. $t(u)<\frac12 ~n^{\frac1k }~OPT_{NSC}$.
Then $|S|> \frac{n}2$.
\end{claim}

%\ifnum\final=1
\begin{proof}
Assuming $|M|< ~\frac14 ~n^{\frac2k}~OPT_{NSC}$, there are at most $\frac14 ~n^{\frac2k}~OPT_{NSC}$ edges of type $k$. Each $v\in V_k$ has exactly $n^{1-\frac1k}$ ancestors in $V_1$,
and therefore $\sum\limits_{u\in V_1} t(u) <n^{1-\frac1k}~ \frac14 ~n^{\frac2k}~OPT_{NSC}=\frac14 n^{1+\frac1k}OPT_{NSC}$. Thus, $\mathbb{E}_{u\in V_1} [ ~t(u)] < \frac14 n^{\frac1k}OPT_{NSC}$ and by Markov's inequality, $Pr_{u\in V_1} [ t(u) < \frac12 n^{\frac1k}OPT_{NTS}] > \frac12$.
\end{proof}
%\fi

Let $T = R \cap S$. The two claims above imply $ |T|\geq 1.$
Now we  argue that a vertex $v \in T$ cannot reach some vertices in $V_{k+2}.$
Recall that an instance of $(n, n^{\alpha}, n^{\beta})$-\textsf{Nice Set Cover}
 was obtained from an instance of $(n, n^{\beta})$- \textsf{Balanced Bounded Set Cover},
 by copying each vertex on the right $ n^{\alpha-1}$ times, which means that the optimal solution
 to one of them is also an optimal solution to the other.
 Suppose we remove $\frac 12 ~n^{\frac1k}OPT_{NSC}$ vertices from $V_{k+2}$. This corresponds to removing at most
$\frac12 n^{\frac1k + 1 -\alpha}OPT_{NSC}= \frac12 ~n^{-\frac1{2k} }OPT_{NSC}= o(1)~ OPT_{NSC}$ vertices from the universe of the related $(n, n^{\beta})$- \textsf{Balanced Bounded Set Cover} instance. Let $OPT_{BSC}$ be the size of a solution to this new \textsf{Set Cover} problem.
Then $OPT_{BSC}\geq (1-o(1))~ OPT_{NSC}$.

 Suppose then that  $v \in T$ could cover all of the elements in $V_{k+2}.$
 Each such vertex $v\in T$ can cover vertices in $V_{k+2}$ in exactly two ways: (1) from the $\frac12 {OPT_{NSC}}$
vertices it reaches in $V_{k+1}$ via paths of length $<k$ using type $i<k$ edges, and (2) by at most $\frac12 n^{\frac1k}OPT_{NSC}$ type $k$ edges it can reach. Thus we must have $OPT_{BSC}\leq \frac12 ~OPT_{NSC}$, which is
a contradiction since  $OPT_{BSC}\geq (1-o(1))~ OPT_{NSC}$.
Thus, $v \in T$ cannot reach all of $V_{k+2}$, and so the optimal \spshort k on $G$ must have size
at least $n^{\frac2k}OPT_{NSC}/4$.\
We can them conclude that $|M| =\Theta(n^{\frac2k}~ OPT_{NSC})$.
\end{proof}%of lemma.

 Suppose now that we could approximate the size of the sparsest \spshort k within $\gamma_1 \log n$ for some $\gamma_1 >0$.
 Then, since $|M| =\Theta(n^{\frac2k} OPT_{NSC})$,
  we could approximate a solution to $(n, n^{\alpha}, n^{\beta})$-\textsf{Nice Set Cover} within $\gamma_2 \log n$, for some $\gamma_2>0$.
  By Lemma \ref{lem:hardnsc} above, $(n, n^{\alpha}, n^{\beta})$-\textsf{Nice Set Cover} cannot be approximated within $\gamma \beta \log n= O(\frac1k) \log n$, unless \textsc{P}$=$\textsc{NP}. Therefore, the size of the sparsest \spshort k
  cannot be approximated within a factor $\gamma_3 \frac{1}k \log n,$ for some $\gamma_3 > 0$, unless \textsc{P}$=$\textsc{NP}.
\end{proof}%of theorem.

\section{Constructing Sparse $k$-TC-Spanners for Path Separable Graphs}\label{app:separable}

\ifnum\final=1
Many TC-spanner constructions can be naturally described by a divide-and-conquer approach.  We cut the
graph into two or more roughly equal-sized components, ensure that comparable vertices in different
components have a path of length at most $k$ between them, and then recurse on each component.  For
instance, this approach adequately describes the TC-spanner constructions given in \cite{AlonSchieber87} for the
line and the rooted tree.  In this section, we extend the divide-and-conquer approach to produce sparse
TC-spanners for various interesting families of digraphs.

The key property that the digraph needs to exhibit in order to admit such an approach is to have
``{good}'' separators.   For our purposes, the most useful notion of separability is path
separability  (\cite{AbrahamGavoille}).

\fi
\begin{definition}[\cite{AbrahamGavoille}]\label{def:pathsep}
Let $G$ be a connected undirected graph with $n$ vertices.  $G$ is {\em $(s,m)$-path
separable} (for $m \geq n/2$) if for any rooted spanning tree $T$ of $G$ either (1) there exists a set
$P$ of at most $s$ monotone paths\footnote{A monotone path in a rooted tree is a subpath of a path with
one endpoint at the root.} in $T$ so that each connected component of $G \backslash P$ is of
size at most $m$, or (2) for some $s' < s$, there exists a set $P$ of $s'$ monotone paths in $T$ so
that the largest connected component of $G \backslash P$ is $(s-s',m)$-path separable.  $G$ is said
to be {\em $s$-path separable} if $G$ is $(s,n/2)$-path separable.  If $G$ is
$s$-path separable, let $S$ be the union of at most $s$ paths in $G$ such that each connected
component of $G \backslash S$ is of size at most $n/2$.  $S$ is called the {\em $s$-path separator}
of $G$. A digraph $G'$ is called an {\em $s$-path separable digraph} if the undirected graph
underlying $TR(G')$ is $s$-path separable.
% Let $G$ be a connected undirected graph with $n$ vertices.  $G$ is said to be {\em $s$-path
%   separable} if for any rooted spanning tree $T$ of $G$, it is true that either 1) there exists a set $P$ of at most $s$ monotone\footnote{A monotone path in a rooted tree is a
%   subpath of a path with one endpoint at the root.} paths in $T$ so that each connected component of $G
% \backslash P$ is of size at most $n/2$, or 2) for some $r < s$, there exists a set $P$ of $r$ monotone paths in $T$ so that the largest connected component of $G \backslash P$ is $(s-r)$-path separable. If $G$ is $s$-path separable, let $S$ be the union of the at most $s$ paths in $G$ such that each connected component of $G \backslash S$ is of size at most $n/2$.  $S$ is called the {\em $s$-path separator} of $G$. A digraph $G'$ is called an {\em $s$-path separable digraph} if the undirected graph underlying $TR(G')$, the transitive reduction of $G'$, is $s$-path separable.
\end{definition}

%\ifnum\final=1
In the above definition, the number of vertices in the path separator is left unspecified.
Trees are $1$-path separable, since $S$ can be taken to be the centroid.  Similarly, graphs
of treewidth $w$ are $(w+1)$-path separable. Thorup \cite{thorup-separators} showed that every
planar graph is $3$-path separable.  Indeed, in the case of any planar graph $G$ and any rooted
spanning tree for it, Thorup proved that there exists a set of $3$ root paths of the tree whose
removal disconnects the graph into components of size at most $n/2$. Abraham and
Gavoille \cite{AbrahamGavoille} studied the more general case of $H$-minor-free graphs and proved
the following.
\begin{theorem}[Theorem 1 of \cite{AbrahamGavoille}]\label{thm:agthm}
Every $H$-minor-free\footnote{A graph is called $H$-minor-free if it belongs to a minor-closed
  graph family that excludes $H$.} graph is $s$-path separable, for $s = s(H)$, and an $s$-path separator can be
computed in polynomial time.
\end{theorem}
\noindent
The definition of path separability used by Abraham and Gavoille is slightly different from \defref{pathsep}. \snote{Explain the difference.} However, the separators produced in their proof of the theorem satisfy our notion of path separability \cite{gavoille-pers}. Our main theorem in this section is the following:
\begin{theorem}\label{thm:pathsepspanners}
If $G$ is a graph drawn from a minor-closed graph family that is $s$-path separable, for $s = \Theta(1)$, then $G$ has a
$2$-TC-spanner of size $O(n \log^2 n)$ and, more generally, a $k$-TC-spanner of size $O(n\cdot \log
n \cdot \lambda_k(n))$ where $\lambda_k(\cdot)$ is the $k$-row
inverse Ackermann function.
\end{theorem}
Since the families of bounded treewidth, planar graphs, and $H$-minor-free graphs (where $H$ is a
fixed minor) satisfy the hypotheses of the theorem, these families have $2$-TC-spanners of size $O(n \log^2 n)$.

\ifnum\final=1
  As mentioned in
\secref{results}, these results are superior to previous constructions used in the access control
literature, \cite{AFB05}.  Additionally, as pointed out in \corref{minor-free-mon-tests},
\thmref{pathsepspanners} also produces monotonicity testers  which have better query complexity than
any previously known tester, for poset families which are minor-closed and whose Hasse graphs forbid a fixed minor.

\fi

\begin{proof}[Proof of \thmref{pathsepspanners}]
First we describe a preprocessing step resembling \cite{thorup-separators} in which the digraph
is divided into subgraphs so that constructing TC-spanners for each subgraph individually results in a
TC-spanner for the entire graph.  Then, we show how to efficiently construct
sparse $2$-TC-spanners for each of these path separable subgraphs. Lastly, we give the construction for general $k$.

\noindent{\bf{Preprocessing Step.} }
Let $G$ be a transitively reduced digraph. \snote{Do you also assume that the underlying undirected graph is connected?} Choose an arbitrary vertex $r \in V(G)$. Let $L_0$ be
the set containing $r$ and all vertices reachable from $r$ by a directed path.  For $i
\geq 1$, let
$L_{2i} \defeq \{v \in G \backslash \cup_{j=0}^{2i-1} L_j : \exists u \in \cup_{j=0}^{2i-1}
L_j \text{ s.t. } u \leadsto v\}$ and
$L_{2i-1} \defeq \{v \in G \backslash \cup_{j=0}^{2i-2} L_j : \exists u \in \cup_{j=0}^{2i-2}
L_j \text{ s.t. } v \leadsto u\}$. Then $L_0,L_1,\ldots,L_t$ partition the vertices in $G$, for some integer $t\leq n$. Evidently,
\begin{claim}\label{claim:levels}
For any vertices $u, v \in G$, if $u \leadsto_G v$ and if $u \in L_i$ and $v \in L_j$, then $|i - j| \leq 1$.
\end{claim}
%\begin{proof}
%Suppose that $u \leadsto v$ in $G$.  If $u = r$, then clearly $v \in L_0$.  Else if $u \in L_i$ and
%$i$ is odd, then if $v \not \in L_i$, $v$ must be in $L_{i+1}$ by definition.  Similarly, if $u \in
%L_i$ and $i$ is even, then either $v \in L_i$ or $v \in L_{i-1}$ by definition.
%\end{proof}

For $1 \leq i \leq t$, let $G_i \defeq L_{i-1} \cup L_i$.  By claim \ref{claim:levels}, any two vertices
with a dipath between them must both be contained in some $G_i$.  Moreover, any dipath between them
must lie entirely in $G_i$.  Therefore, a $k$-TC-spanner for $G$ is the union of $k$-TC-spanners for each $G_i$.
Notice that $\sum_i |V(G_i)| \leq 2 |V(G)|$.

We next construct a spanning tree $T_G$ for the undirected graph underlying $G$ that is rooted at
$r$ and has the following property: for any undirected path in $T_G$ from the root, the restriction of
the path to a single level $L_i$ consists of a single directed path.

 $T_G$ can be constructed
inductively.  First, since by definition $r$ reaches all the vertices in $L_0$, a spanning tree of
$L_0$ rooted at $r$ can be constructed with all edges oriented away from $r$.  Now suppose we have a
tree $T_{i-1}$ that is rooted at $r$, spans all vertices in $\cup_{j=0}^{i-1} L_j$, and whose
restriction to each level $0,\ldots,i-1$ consists of a single directed path.  If $i$ is odd,
$T_{i-1}$ can be extended to a tree $T_i$ where all the new edges are oriented towards
$\cup_{j=0}^{i-1} L_j$.  The case when $i$ is even is symmetric.  Our desired spanning tree $T_G$ is
$T_t$.  The following lemma is immediate by the construction.

\begin{lemma}\label{lem:2dipath}
A monotone path in $T_G$ restricted to $G_i$, for any $i \in [t]$, is a concatenation of $\leq 2$ dipaths.
\end{lemma}

% Defining the $G_i$ as well as finding the spanning tree $T_G$ can be achieved in linear time.

%\subsection{Constructing $2$-TC-spanners for  path separable digraphs}\label{sec:2-spanner}

We assume $G$ is transitively reduced and connected.  If $G$ is not connected, we can apply our algorithm on each component.  We describe how to construct $H$, a
$2$-TC-spanner for $G$ of size $O(n \log^2 n)$.

\noindent
{\bf The recursive graph fragmentation.}
First, we apply the preprocessing step described above to $G^0 \defeq
G$; that is, we obtain a spanning tree $T_{G^0}$ and a collection of subgraphs, $G^0_1,G^0_2,\ldots$. \snote{Throughout the paragraph, the sequences should not be infinite} By
definition of path separability, there exists a set $P^0$ of monotone paths on $T_{G^0}$
such that one of two situations happens: (1) all the connected components in $G^0 \backslash P^0$ are
of size at most $n/2$, (2) the largest component of $G^0 \backslash P^0$ is of size
greater than $n/2$ and is path separable.  Let $G^1$ denote the induced subgraph of $G^0$ on the
largest component of $G^0 \backslash P^0$.  We can apply the preprocessing to $G^1$ to obtain a
collection of subgraphs $G^1_1,G^1_2,\ldots$ and a spanning tree $T_{G^1}$ rooted at some arbitrary
vertex in $G^1$.  Again, we find an appropriate set of paths $P^1$ in $T_{G^1}$ and we recurse if
necessary on the largest component of $G^1 \backslash P^1$.  The recursion ends when the graph
has been disconnected into components of size at most $n/2$.  Notice that the total number of paths
in $P^0 \cup P^1 \cup \cdots$ is at most $s=\Theta(1)$, and we then recurse only a constant number of
times.  Let $S \defeq P^0 \cup P^1 \cup \cdots$.

%end fragmentstion step
\noindent
{\bf Connecting the cut pairs in $G$.}
Call a pair of vertices $(u,v)$ a \textit{cut pair} if $u \leadsto_G v$ and every directed path from $u$
to $v$ intersects a path in $S$.
We show how to connect every cut pair by a path of length at most $2$.

Repeat the following for
every vertex $v \in V(G)$. Let $I = \{i : v \in V(G^i)\}$ and, additionally, for each $i \in I$, let
$J_i = \{j : v \in V(G^i_j)\}$.  Do the following for each $i \in I$ and each $j \in J_i$.  Let
$P^i_j$ denote the restriction of the paths in $P^i$ to $G^i_j$.  Each undirected path in $P^i_j$
is a concatenation of at most $2$ directed paths by \lemref{2dipath}. Break up the paths in $P^i_j$
into dipaths. Consider some dipath $P \in P^i_j$ which visits the vertices
$p_1, p_2, \ldots, p_m$ in that order, where $m \leq |V(G^i_j)|$. For simplicity of presentation, assume $m$ is a power of $2$.  For each $1 \leq z \leq \log_2 m$, add the following
two edges in $H$: \snote{Describe this in words: mid point is a hub at level 1, midpoints of two halves are hubs at level 2,...; $\log m$ levels; at each level, $v$ connects to the "nearest" hub} (i) an edge from $v$ to $p_{y_1\cdot m/2^z}$ where $y_1 = \min_{y} \{1 \leq y <
2^z : v \leadsto p_{y\cdot m/2^z}$ in $G\}$ and (ii) an edge to $v$ from $p_{y_2\cdot m/2^z}$ where $y_2
= \max_{y} \{1 \leq y < 2^z : p_{y \cdot m/2^z} \leadsto v$ in $G\}$.  If any of the sets inside the
$\min$ or $\max$ is empty, do not add the respective edge.  Finally,\anote{This is just because of
  bad notation, so can be improved.} \snote{Make $m$ (a power of two) - 1 to get rid of this special case.} add an edge $(v,p_m)$ if $v
\leadsto p_m$ in $G$ and $(p_m,v)$ if $p_m \leadsto v$ in $G$.  Repeat this process
for every separator dipath that is a subpath of an undirected path  in ${P}^i_j$.

\noindent
{\bf Outer Recursion.}
For each connected component $C$ of $G \backslash S$, recurse on the subgraph induced by $C$.  $C$
is also path separable since the graph family is minor-closed.

%end spanner construction
\begin{lemma}\label{lem:finalsep}
The above construction efficiently produces a \spshort 2 on $G$ of size at most $O(n \log^2 n)$.
\end{lemma}
\noindent{\em Proof Sketch.} In each $G^i_j$, the separators are nicely structured as only a constant
number of directed paths.  Hence, we can add only $O(|V(G^i_j)| \log |V(G^i_j)|)$ edges in order to connect the cut
pairs present in each $G^i_j$.  Since $\sum_j |V(G^i_j)| $ $\leq$ $ 2n$
and the number of $G^i$'s is $\Theta(1)$, the total number of edges added in each step of the outer
recursion is $O(n \log n)$. The size of the remaining connected components halves after
each graph fragmentation step. So, the outer recursion continues only $\log n$ times, making the
total number of added edges $O(n\log^2 n)$. The construction results in a \spshort{2} because every pair of related vertices $(u,v)$ is a cut pair at some level of the outer recursion.  Then, $u$ and $v$ are both contained in some $G^i$.  One can
check that the above construction ensures that both $u$ and $v$ are adjacent to the same vertex on
some separator dipath intersecting a dipath from $u$ to $v$.  The formal proof is below.

\begin{proofof}{\lemref{finalsep}}
Let us first see why connecting every cut pair by a path of length at most $2$, and recursing on smaller components produce a $2$-TC-spanner for $G$.  Indeed, if $(u,v)$ is a cut pair, then
the first step ensures a path of length at most $2$ between them.  If $(u,v)$ is not a cut pair but
there exist some dipaths from $u$ to $v$, then $u$ and $v$ are in the
same component $C$ of $G \backslash S$, and there exists a dipath between them that lies entirely within this
component.  In this case, constructing a $2$-TC-spanner on the subgraph induced by this $C$ suffices
to connect $u$ and $v$ by a path of length at most $2$.

Let us now argue that this process connects every cut pair by a path of length at most $2$. Consider
some cut pair $(u,v)$.  Let $i$ be the smallest nonnegative integer such that every dipath from $u$ to $v$
intersects a path in $\cup_{i' \leq i} P_{i'}$.  Therefore, there must be a dipath from $u$ to $v$ entirely
contained in $G^i$, and by claim \ref{claim:levels}, it follows that there is a $j$ such that
both $u$ and $v$ are in $G^i_j$.  Suppose $P \in P^i_j$ is a separator dipath of length $m$ (a power
of $2$) that intersects a dipath in $G^i$ from $u$ to $v$.  Let $y_1 = \min_y \{y : u
\leadsto_{G^i} p_y\}$ and $y_2 = \max_y \{y : p_y \leadsto_{G^i} v\}$.  $y_1 \leq y_2$ because
otherwise there cannot be a vertex on $P$ that lies on a path from $u$ to $v$.  One possibility is that
$y_1 = y_2 = m$ in which case the construction ensures that we add the edges $(u,p_m)$ and $(p_m,
v)$.  Otherwise, there exists some $z \in \{1,2,\dots,\log_2 m\}$ such that there is a unique $y \in
\{1,2,\dots,2^z - 1\}$ for which $y \cdot m/2^z$ is in the interval $[y_1,y_2]$. Moreover, $u \leadsto
p_{y_1} \leadsto p_{y\cdot m/2^z} \leadsto p_{y_2} \leadsto v$.  Therefore, the construction above adds the edges
$(u,p_{y \cdot m/2^z})$ and $(p_{y \cdot m/2^z},v)$.

In connecting the cut pairs in $G^i_j$, we add at most $O(s ~|V(G^i_j)| ~\log |V(G^i_j)|)$ because there are at
most $O(s)$ separator dipaths in $G^i_j$ and for any separator dipath $P$, each vertex in $G^i_j$ is
connected to at most $2(\log_2 |V(P)|+1) \leq O(\log |V(G^i_j)|)$ vertices on the path.  Recall that
there are only a constant number of $G^i$s and $\sum_j |V(G^i_j)| \leq 2 |V(G^i)|$.  Thus, if
$S(G)$ denotes the total number of edges in the constructed $2$-TC-spanner for $G$, we have that:
\begin{equation*}
S(G) \leq \sum_{C \text{ is a c.c. of } G\backslash S} S(C) + O\left(\max_i \sum_j O(|V(G_j^i)|\cdot \log
|V(G_j^i)|)\right)
\leq \sum_{C \text{ is a c.c. of } G\backslash S} S(C) + O(n \log n)
\end{equation*}
Since  $|V(C)| \leq n/2$ for any connected component of $G \backslash S$, it follows that $S(G) = O(n \log^2 n)$.

If the strong path separators can be found in polynomial time, as is guaranteed, for example, in
\thmref{agthm}, then it is clear that the above $2$-TC-spanner can be constructed efficiently.
\end{proofof}
We now prove the part of \thmref{pathsepspanners} concerning \spshorts{k} for general $k$. %The construction of $k$-TC-spanner is similar to the $2$-TC-spanner construction. %Essentially we will use techniques inherent in \cite{AlonSchieber87} (and later rediscovered in \cite{AFB05}).
Again, assume $G$ is transitively reduced and connected; now, we wish to construct $H$, a $k$-TC-spanner for $G$.  We
perform the same preprocessing as before in order to obtain induced subgraphs $G^0,G^1,\ldots$ and a corresponding $s$-path separator $S = P^0 \cup
P^1 \cup \cdots$.  Define a cut pair $(u,v)$ to be a pair of vertices in $G$ such that $u \leadsto
v$ and every directed path from $u$ to $v$ intersects a path in $S$.  This time, our plan is to
connect all cut pairs by a path of length at most $k$ and then to recurse on each of the connected
components that remain after removing the vertices in the paths of $S$.  By the argument used
earlier, this process produces a $k$-TC-spanner.

Now we show how to connect cut pairs $(u,v)$ with a path of length at most $k$.  Do the following
for every vertex $v \in V(G)$.  Let $I = \{i : v \in V(G^i)\}$ and for each $i \in I$, let $J_i =
\{j : v \in V(G^i_j)\}$.   Do the following for each $i \in I$ and for each $j \in J_i$.  Let
$P^i_j$ be the restriction of the paths in $P^i$ to $G^i_j$.  Break up the undirected paths into
dipaths, increasing the size of $P^i_j$ by a factor of at most $2$.  Do the following for each
dipath $P \in P^i_j$. Let $m$ be the length of $P$ which visits vertices  $p_1,  p_2, \ldots p_m$ in
that order.  Additionally, let $c(\ell)$ be a concave increasing function of $\ell$, which
satisfies $c(\ell) < \ell$ that we specify later.  For simplicity of presentation, we
omit all floors and ceilings. Let $c^*(\ell)$ denote the smallest $z$ such that $c^z(\ell) =
\Theta(1)$ where $c^z(\cdot)$ denotes the $z$th functional power of $c$.  For each  $z$ such
that $1 \leq z \leq c^*(m)$, add the following two edges to $H$: (i) an edge from $v$ to $p_{y_1\cdot c^z(m)}$ where $y_1 = \min_{y} \{1 \leq y <
m/c^z(m) : v \leadsto p_{y\cdot c^z(m)}$ in $G\}$ and (ii) an edge to $v$ from $p_{y_2\cdot c^z(m)}$ where $y_2
= \max_{y} \{1 \leq y < m/c^z(m) : p_{y \cdot c^z(m)} \leadsto v$ in $G\}$.  If any of the sets inside the
$\min$ or $\max$ is empty, do not add the respective edge.

Finally, do the following for every dipath $P \in P^i_j$ that visits vertices $p_1, p_2, \dots, p_m$
in that order.
\begin{center}
\fbox{
\begin{minipage}{6.5in}
\textsc{CONNECT-ON}$(P)$

\begin{enumerate}
\item
Add to $H$ the edges in a $(k-2)$-spanner of the induced subgraph of the transitive closure of $G$
on $\{p_{c(m)}, p_{2c(m)}, \dots, p_m\}$.

\item
Remove the points in the set $\{p_{c(m)}, p_{2c(m)}, \dots, p_m\}$ from $P$ and run
\textsc{CONNECT-ON} on each connected component of $P$ that remains.
\end{enumerate}
\end{minipage}}
\end{center}

This completes our description of $H$.  It is not too hard to see that $H$ is indeed a
$k$-TC-spanner, using reasoning as in the previous section.  The only difference is that
now, for a cut pair $(u,v)$, it could be that $u$ and $v$ are adjacent to different vertices on the separating
dipath.  But we have the guarantee by \textsc{connect-on} above that two path vertices
in the same recursion level of \textsc{connect-on} have a path of length at most $k-2$ between
them.  Hence, it follows that $u$ and $v$ have a path of length at most $k$ between them.

Now, we bound the size of $H$.  First, let us count the number of edges added in each step of the
main recursion (that is, not counting the edges needed to connect pairs within components of size at
most $n/2$).   Denote by $\ell(n,k)$ the quantity $S_k(L_n)$, the size of the optimal $k$-TC-spanner for the directed line on
$n$ vertices.  Let us count all the edges added that are incident to some separating dipath $P$ of
size $m$.  Denote this quantity $f(m)$.  By the definition of \textsc{connect-on}:
\begin{equation*}
f(m) \leq O(n) + n/c(m) f(c(m)) + \ell(m/c(m), k-2)
\end{equation*}
It can be seen that $f(m)$ is minimized when $\ell(m/c(m),k-2) = O(m) = O(n)$.  For example, for $k=4$,
$\ell(n,2) = O(n \log n)$ and in this case, $c(m)$ should be chosen to be $\log m$ since $\frac {m}{ \log m}
\cdot \log \frac{m}{\log m} = O(m)$.  In any case, once $c(m)$ is fixed, the solution to the above
functional equation turns out to be: $f(m) \leq O(n\cdot c^*(n))$.  Also, the number of edges
added  to the vertices not on the separating paths is also $O(n\cdot c^*(n))$.
Making the same arguments as in the analysis of the $2$-TC-spanner, we find that the number
of edges added in total, counting all the paths and all the $G^i$'s is still $O(n\cdot c^*(n))$.
Since there are $\log n$ levels of the recursion at the top level (at each level, the size of the
largest component is decreased by a factor of $2$), the total number of edges added to $H$ is $O(n
\cdot \log n \cdot c^*(n))$.  Finally, we use the results of \cite{AlonSchieber87} about optimal
$k$-TC-spanners of the directed line to conclude that $c^*(n) = O(\lambda_k(n))$.
%
%\noindent{\it Observation:}
% Note that a trivial lower bound on the size of a \spshort{k} on $G_1\times G_2$ is given by
% $\min( |G_1| S_k(G_2), $ $|G_2| S_k(G_1))$. If $|G_1|\ll |G_2|$ then the trivial lower and upper bounds are matching.
\end{proof}
\section{NP-Hardness of \spshort{k}}\label{sec:nphard}
Theorem \ref{thm:khard} breaks down for $k=\Omega(\log n)$. For these large values of $k$ we have the following.
\begin{theorem}\label{thm:apx-hardness}
 For any $k<n^{1-\epsilon}$ for any $\epsilon > 0$, it is NP-hard to approximate the size of the sparsest \spshort{k} within a factor of $~1+\gamma$,
 for some $\gamma = \Omega \left (\frac{1}{k} \right )$. %$0<\gamma=\Omega(\frac{1}{k})<1$.
 \end{theorem}

 \begin{proof}[Proof of Theorem \ref{thm:apx-hardness}]
 \newcommand{\nodecover}{{\sc 3NodeCover}}
 We use a reduction from \nodecover\ to show that, unless $P = NP$, \tcspanner\ cannot be approximated within a factor of $1+\Omega \left (\frac{1}{k} \right )$. That is, for constant $k$, the problem is APX-hard. An instance of $3$-Node Cover consists of a collection $D$ of subsets of a universe $X$. Each subset contains at most $3$ elements, and each element of $X$ is contained in at most $2$ subsets. The goal is to output a minimum size subcollection $M\subseteq D$ whose union is $X$. We need the following result in \cite{BerKarpi}:

 \begin{lemma}\label{lem:hard}
 \nodecover\ is NP-hard to approximate within a factor of $~1+c$, for some constant $c>0.$
 \end{lemma}

 We now give a reduction from \nodecover\ to \tcspanner.
 For a given instance $R$ of \nodecover\ we construct the following graph $G$.
 Let $V_1$ be the set of vertices representing each set $d\in D$. Let $V_2$ be the set of vertices representing each element  $t\in X$. Draw a directed edge from each vertex in $V_1$ corresponding to $d\in D$ to the vertices in $V_2$ corresponding to elements of $d$.
 Add an extra vertex $a$ and draw directed edges from $a$ to every element of
 $V_1$. For each vertex $v\in V_1$ add $k-1$ new vertices $v_1, v_2, \ldots, v_{k-1}$ and connect them via a directed
 path of length $k$ passing through $a, v_1, v_2, \ldots, v_{k-1}, v$ in the given order. Call this path $P(v)$.

 Let $OPT_{S}$ be the size of a minimum \spshort{k} on $G$ and $OPT_{3NC}$ be the size of the solution to the initial instance $R=(D, X)$ of
 \nodecover. Let $|D |=n$.

 \begin{claim}
 \label{claim:opt-calc}
 $OPT_{S} = OPT_{3NC} +kn+ \sum_{d \in D}|d|$.
 \end{claim}
 \begin{proof}
 We show that there is a sparsest spanner $H$ that contains only edges from $V_1$ to $V_2$, from vertex $a$ to some
 vertices in $V_1$ and the edges on the paths $P(v)$, for all $v\in V_1$.

 Note that all edges from $V_1$ to $V_2$ need to be included in a sparsest spanner,
 since each of them forms a unique directed path connecting its endpoints. There are $\sum_{d \in D}|d|$ such edges. Similarly, for each vertex $v\in V_1$, all the edges on the path $P(v)$ need to be included in $H$. In total, there are $kn$ such edges.

  For $v\in V_1$, suppose that $H$ contains an edge $(v_i, t)$, for some $v_i\in P(v)$
 and $t\in V_2$. We claim that such an edge $(v_i, t)$ can be removed and substituted with an edge $(a, u)$, where $u\in V_1$ s.t. $u$ is adjacent to $t$. Indeed, all vertices on $P(v)$ except for $a$ are already at distance $\leq k$ from $t$ in $H$. Thus, to reach $t$ from $a$ it is enough to include in $H$ a directed edge from $a$ to any $u\in V_1$ that is adjacent to $t$.

 Similarly, an edge $e$ between vertices on a path $P(v)$ in $H$ can be replaced by an edge $(a, v)$. Indeed, such an edge $e$ can only be useful to connect $a$ to some vertices in $V_2$, via a path that passes through $v$.

 Therefore, among the edges from $a$ to $V_1$ a sparsest spanner need only contain the edges that connect $a$ to a minimum set of vertices in $ V_1$ that cover $V_2$. Thus, there are exactly $OPT_{3NC}$ such edges.
 \end{proof}

Suppose that there exists $0<\gamma$ and an algorithm $\mathcal{A}$ that computes the size of a sparsest \spshort{k}
within $1+\gamma$. Namely, $\mathcal{A}$ outputs $s$, such that $OPT_{S}\leq s\leq (1+\gamma) OPT_{S}$.
We show that $\gamma \geq \frac{c}{19+6k}$, where $c$ is the constant from Lemma \ref{lem:hard}.

 Each set $d$ contained in an optimal solution to $R$ covers at most $3$ elements of the universe $X$. Therefore $|X| \leq 3 OPT_{3NC}$. Any element of $X$ is contained in at most $2$ sets of $D$, and therefore, $|X|\geq \frac{n}{2}.$ This implies that $n\leq 6 OPT_{3NC}$.
Let $s'=s-(k+3)n$.
Then
\begin{align*}
OPT_{3NC}\leq s' &\leq OPT_{3NC}+\gamma (OPT_{3NC}+kn+3n)\\
                 &\leq OPT_{3NC}+ \gamma(OPT_{3NC}+6k OPT_{3NC}+18 OPT_{3NC})\\
                 &=OPT_{3NC}( 1+ \gamma(19+6k)).
\end{align*}

Finally, using \ref{lem:hard}, it follows that $\gamma \geq \frac{c}{19+6k}$. Thus $\gamma=\Omega(\frac{1}{k})$.

 %Given Claim~\ref{claim:opt-calc}, it is straightforward to prove Theorem~\ref{thm:apx-hardness}. The details are deferred to Appendix~\ref{sec:missing}.
 \end{proof}
\end{document}